\documentclass[sigconf, nonacm]{acmart}

\newcommand{\kw}[1]{\mbox{\textnormal{\textsf{#1}}}}

\newcommand{\invoke}[1]{\mbox{\textnormal{\ensuremath{\textsf{#1}}}}}
\newcommand{\stitle}[1]{\vspace{1ex} \noindent{\bf #1}}
\long\def\comment#1{}
\newcommand\mdbf[1]{\textnormal{\textbf{#1}}}

\newcommand{\mycolora}{black}
\newcommand{\mycolorb}{black}
\newcommand{\mycolorc}{black}

\renewcommand{\a}{\alpha}
\renewcommand{\b}{\beta}
\renewcommand{\int}{\overline}

\newcommand{\fullversion}[1]{{#1}}
\newcommand{\submitversion}[1]{\comment{#1}}

\newtheorem{Lemma}{Lemma}
\newtheorem{Theorem}{Theorem}
\newtheorem{Example}{Example}
\newtheorem{Definition}{Definition}

\newcommand\figref[1]{Figure~\ref{#1}}
\newcommand\subfigref[1]{Figure~\ref{#1}}
\newcommand\lemref[1]{Lemma~\ref{#1}}
\newcommand\theref[1]{Theorem~\ref{#1}}
\newcommand\tabref[1]{Table~\ref{#1}}
\newcommand\algref[1]{Algorithm~\ref{#1}}
\newcommand\expref[1]{Example~\ref{#1}}
\newcommand\defref[1]{Definition~\ref{#1}}

\comment{

\AtBeginDocument{%
  \providecommand\BibTeX{{%
    \normalfont B\kern-0.5em{\scshape i\kern-0.25em b}\kern-0.8em\TeX}}}
\setcopyright{acmcopyright}
\copyrightyear{2025}
\acmYear{2025}
\acmDOI{XXXXXXX.XXXXXXX}
\acmConference[Conference acronym 'XX]{Make sure to enter the correct
  conference title from your rights confirmation emai}{June 03--05,
  2018}{Woodstock, NY}
\acmPrice{15.00}
\acmISBN{978-1-4503-XXXX-X/18/06}

}

\usepackage[linesnumbered,ruled,vlined]{algorithm2e}
\usepackage{algorithmic} \SetKwProg{Fn}{Function}{}{end}
\usepackage{xcolor}
\usepackage{subcaption}
\usepackage{mdsymbol}
\usepackage{mathrsfs}
\usepackage{multirow}
\usepackage{tikz}
\usepackage{CJKutf8}

\makeatletter
\patchcmd\algocf@Vline{\vrule}{\vrule \kern-0.4pt}{}{}
\patchcmd\algocf@Vsline{\vrule}{\vrule \kern-0.4pt}{}{}
\makeatother

\begin{document}
\sloppy

\title{Optimal $(\alpha,\beta)$-Dense Subgraph Search in Bipartite Graphs}

\author{Yalong Zhang}
\affiliation{%
  \institution{Beijing Institute of Technology}\city{Beijing}\country{China}
}\email{yalong-zhang@qq.com}

\author{Rong-Hua Li}
\affiliation{%
  \institution{Beijing Institute of Technology}\city{Beijing}\country{China}
}\email{lironghuabit@126.com}

\author{Qi Zhang}
\affiliation{%
  \institution{University of Science and Technology Beijing}\city{Beijing}\country{China}
}
\email{qizhangcs@ustb.edu.cn}

\author{Guoren Wang}
\affiliation{%
  \institution{Beijing Institute of Technology}\city{Beijing}\country{China}
}
\email{wanggrbit@126.com}

\begin{abstract}
Dense subgraph search in bipartite graphs is a fundamental problem in graph analysis, with wide-ranging applications in fraud detection, recommendation systems, and social network analysis. The recently proposed $(\alpha, \beta)$-dense subgraph model has demonstrated superior capability in capturing the intrinsic density structure of bipartite graphs compared to existing alternatives. However, despite its modeling advantages, the $(\alpha, \beta)$-dense subgraph model lacks efficient support for query processing and dynamic updates, limiting its practical utility in large-scale applications. To address these limitations, we propose \kw{BD-Index}, a novel index that answers $(\alpha, \beta)$-dense subgraph queries in optimal time while using only linear space $O(|E|)$, making it well-suited for real-world applications requiring both fast query processing and low memory consumption. We further develop two complementary maintenance strategies for dynamic bipartite graphs to support efficient updates to the \kw{BD-Index}. The space-efficient strategy updates the index in time complexity of $O(p \cdot |E|^{1.5})$ per edge insertion or deletion, while maintaining a low space cost of $O(|E|)$ (the same as the index itself), where $p$ is typically a small constant in real-world graphs. In contrast, the time-efficient strategy significantly reduces the update time to $O(p \cdot |E|)$ per edge update by maintaining auxiliary orientation structures, at the cost of increased memory usage up to $O(p \cdot |E|)$. These two strategies provide flexible trade-offs between maintenance efficiency and memory usage, enabling \kw{BD-Index} to adapt to diverse application requirements. Extensive experiments on 10 large-scale real-world datasets demonstrate high efficiency and scalability of our proposed solutions.
\end{abstract}

\comment{


\received{20 February 2007}
\received[revised]{12 March 2009}
\received[accepted]{5 June 2009}
}

\maketitle

\section{INTRODUCTION}

Bipartite graphs are widely used to represent relationships between two types of entities, such as user-page social networks \cite{social}, customer-product networks \cite{customer, customer2}, collaboration networks \cite{collaboration}, and gene co-expression networks \cite{gene, gene2}. Mining dense subgraphs in bipartite graphs has numerous practical applications. For example, in customer-product networks, users within the same dense community typically share similar product preferences, enabling e-commerce platforms to make targeted recommendations. In social networks, fraudsters may hire a large number of bot accounts to boost their visibility through likes or interactions. Such behavior tends to form dense communities around the fraudsters, making it easier for platforms to detect these fraudulent activities.

Motivated by these applications, various cohesive subgraph models have been proposed to identify dense communities in bipartite graphs, including biclique \cite{biclique1, biclique2, biclique3, biclique4}, $k$-biplex \cite{biplex1, biplex2, biplex3}, $k$-bitruss \cite{bitruss1, bitruss2, bitruss3}, $(\alpha,\beta)$-core \cite{abcore1, abcore2, abcore3}, and $(\alpha,\beta)$-dense subgraph \cite{ddbipartite}. However, despite the significant progress made by these models in capturing dense communities, they still suffer from several critical limitations. For example, there are no known polynomial-time algorithms for computing biclique and $k$-biplex, limiting their applicability in large bipartite networks. The computation of $k$-bitruss heavily relies on enumerating butterfly structures, which becomes prohibitively expensive in dense graphs \cite{bitruss1, bitruss2, bitruss3}. The $(\alpha,\beta)$-core model only considers node degrees and often fails to accurately capture the underlying density structure of the graph. {\color{\mycolora}In contrast, the $(\alpha,\beta)$-dense subgraph is a density-based model that exhibits a desirable ``dense-inside, sparse-outside'' property (\theref{The:edges}) and has been shown in \cite{ddbipartite} to capture the density structure of bipartite graphs more effectively than other cohesive subgraph models such as $(\alpha,\beta)$-core, bitruss, biplex, and biclique.}

{\color{\mycolorb}Given the favorable properties of the $(\alpha,\beta)$-dense subgraph model, it can be applied to a wide range of real-world scenarios where $(\alpha,\beta)$-dense subgraph queries are frequently issued. For example: (1) in user–product bipartite networks for recommendation systems, platforms may process millions of recommendation requests per second \cite{amazon,DBLP:conf/cikm/ZhangTWW24}; by tuning $(\alpha,\beta)$ values, the system can retrieve product groups with different densities, enabling personalized and diverse recommendations; (2) in fraud detection over financial networks, malicious behaviors may hide in subgraphs with varying density levels, requiring the system to frequently query dense subgraphs with different $(\alpha,\beta)$ combinations \cite{DBLP:conf/apweb/AllahbakhshIBBBF13, abcore2}; (3) in social networks, different $(\alpha,\beta)$ combinations reveal communities with distinct interaction patterns, such as a user–blog bipartite graph where high $\alpha$ and low $\beta$ (i.e., stronger filtering on the user side) identify highly active users, while low $\alpha$ and high $\beta$ (i.e., stronger filtering on the blog side) highlight popular blogs \cite{ddbipartite}.} Moreover, these bipartite networks are typically dynamic, with frequent edge insertions and deletions as users interact with items or each other in real time. This makes it essential to support efficient and responsive $(\alpha,\beta)$-dense subgraph computation under continuous graph updates.

{\color{\mycolora}However, existing state-of-the-art algorithm \kw{DSS++} \cite{ddbipartite} for computing $(\alpha,\beta)$-dense subgraphs struggles to meet these demands. The \kw{DSS++} algorithm adopts a network-flow-based approach that requires solving a maximum flow problem per query, resulting in a worst-case time complexity of $O(|E|^{1.5})$, where $|E|$ is the number of edges in the bipartite graph. As real-world graphs continue to grow rapidly in size and $(\alpha,\beta)$-dense subgraph queries are issued at high frequency, this online approach becomes increasingly inefficient. For instance, as shown in our experiments (Section~6, Exp-1), on the large \kw{LI} dataset ($112.3$ million edges), the \kw{DSS++} algorithm takes an average of 21.49 seconds to process a single query—far exceeding the latency requirements of many real-time applications, which typically demand a response time below 0.5 seconds \cite{amazon,DBLP:conf/cikm/ZhangTWW24,DBLP:conf/sigmod/KersbergenSS22}.} Moreover, real-world bipartite graphs are typically dynamic, with frequent updates such as edge insertions and deletions. Efficient computation of $(\alpha,\beta)$-dense subgraphs in such evolving graphs is crucial for numerous applications, including identifying closely-connected communities in social networks over time, generating real-time product recommendations in customer-product networks, and dynamically monitoring dense regions potentially indicative of fraudulent activities in financial networks. Unfortunately, existing methods for handling $(\alpha,\beta)$-dense subgraph queries in dynamic graphs remain limited to inefficiently re-executing the \kw{DSS++} algorithm whenever the graph is updated.


To address these limitations, we study efficient computation of $(\alpha,\beta)$-dense subgraphs on both static and dynamic bipartite graphs. We first propose \kw{BD-Index}, a novel index structure specifically designed to support the optimal query processing (i.e., the query time complexity is linear to the size of results). Meanwhile, the space complexity of \kw{BD-Index} is carefully bounded to $O(|E|)$, ensuring scalability to large-scale graphs. Thus, \kw{BD-Index} serves as a practical solution for real-time and large-scale dense subgraph computation. For dynamic graphs, we investigate the maintence of \kw{BD-Index} and propose two complementary strategies. The first is a space-efficient approach that preserves the linear space complexity advantage of \kw{BD-Index}. The second is a time-efficient approach that leverages additional storage for faster updates. Notably, the latter provides a practical trade-off between time and space complexity, enabling efficient maintenance of \kw{BD-Index} even for graphs with hundreds of millions of edges. In summary, the main contributions of this paper are as follows.


\stitle{A novel index structure with optimal query time.} Our design exploits the hierarchical property of $(\alpha,\beta)$-dense subgraphs, i.e., higher-density subgraphs are necessarily contained within lower-density ones. Leveraging this property, we introduce two novel concepts: $\alpha$-rank and $\beta$-rank, which represent the largest $(\a, \b)$ values for which the node belongs to a dense subgraph. Using these ranks, we elaborately organize nodes into $p$ node lists, where $p$ is the largest integer such that $(p,p)$-dense subgraph is non-empty. For each $(\alpha, \beta)$ pair, we maintain a pointer to the corresponding node list, allowing the query to retrieve all result nodes by scanning the list once. This enables an optimal query time of $O(|D_{\a, \b}|)$, where $D_{\a, \b}$ is the result set. Furthermore, by fully exploiting the nested nature of dense subgraphs, our index requires only linear space $O(|E|)$, making it both query-efficient and space-efficient. We also present an index construction algorithm with a time complexity of $O(p \cdot |E|^{1.5} \cdot \log |U \cup V|)$, which can effectively handle graphs with hundreds of millions of edges. 


\stitle{Novel space-efficient index maintenance algorithms.} To support efficient updates of the \kw{BD-Index}, we first establish several update theorems for $(\alpha,\beta)$-dense subgraphs, covering both insertion and deletion cases. These theorems reveal that when an edge is updated, the $\alpha$-rank of nodes in the lower side $V$ and the $\beta$-rank of nodes in the upper side $|U|$ may change by at most 1. Leveraging these update properties, we propose two space-efficient maintenance algorithms: \kw{BD-Insert-S} and \kw{BD-Delete-S}, which handle edge insertions and deletions, respectively. Both algorithms operate in linear space $O(|E|)$, preserving the space advantage of \kw{BD-Index}. In terms of time complexity, our \kw{BD-Insert-S} and \kw{BD-Delete-S} algorithms achieve an update complexity of $O(p \cdot |E|^{1.5})$, which is significantly lower than the $O(p \cdot |E|^{1.5} \cdot \log |U \cup V|)$ complexity of the baseline that recomputes the entire index from scratch.

\stitle{Novel time-efficient index maintenance algorithms.} Building on the established update theorems, we further propose time-efficient maintenance algorithms. We introduce a novel concept called \emph{egalitarian orientation}, transforming the maintenance of \kw{BD-Index} into maintaining a set of egalitarian orientations. Leveraging this transformation, we propose two algorithms, \kw{BD-Insert-T} and \kw{BD-Delete-T}, that first maintain the egalitarian orientations and subsequently update \kw{BD-Index} using only a few breadth-first search operations. These algorithms achieve a significant reduction in time complexity from $O(p \cdot |E|^{1.5})$ required by the space-efficient algorithms to $O(p \cdot |E|)$. While maintaining $p$ additional egalitarian orientations increases the space complexity to $O(p \cdot |E|)$, this remains within acceptable limits. Consequently, \kw{BD-Insert-T} and \kw{BD-Delete-T} achieve a favorable time-space trade-off, enabling efficient and scalable maintenance of \kw{BD-Index}.

\stitle{Extensive experiments.} We conduct extensive experiments on 10 large real-world datasets, and the results demonstrate the efficiency and scalability of our solutions. First, the index-based query algorithm, \kw{Query-BD-Index}, outperforms the state-of-the-art online algorithm by 3 to 4 orders of magnitude. On the largest dataset \kw{LI} (with over 100 million edges), it achieves an average query time of merely 2.74 milliseconds. Second, \kw{BD-Index} exhibits memory usage comparable to the original graph size, and our index construction algorithm scales effectively to large graphs like \kw{LI}. Third, our space-efficient maintenance algorithms achieve up to one order of magnitude speedup compared to index recomputation from scratch, while the time-efficient algorithms are 2-4 orders of magnitude faster than the space-efficient approaches. Although requiring approximately one order of magnitude more memory, the time-efficient maintenance algorithms only consumes 51 GB for maintaining \kw{BD-Index} on the \kw{LI} dataset, which is easily acceptable in practical applications.

\stitle{Reproducibility and full version paper.} The source code and the full version of this paper can be found at \url{https://anonymous.4open.science/r/bd-index-F7E2}.

\section{PRELIMINARIES}

We consider an undirected and unweighted bipartite graph $G = (U, V, E)$, where $U$ and $V$ are two disjoint node sets, and $E \subseteq U \times V$ represents the set of edges connecting nodes from distinct sets. For each node $x \in U \cup V$, we denote its neighbor nodes set in $G$ as $N_G(x)$, and its degree as $d_x(G) = |N_G(x)|$ (or simply $N(x)$ and $d_x$ when the context is unambiguous). Given a subset of nodes $X \subseteq U \cup V$, we define subsets $X^U = X \cap U$ and $X^V = X \cap V$ corresponding to the respective partitions. The subgraph induced by $X$ is $G(X)=(X^U, X^V, E(X))$, where $E(X)$ includes all edges between nodes in $X$.

\begin{table}[t]
\color{\mycolorc}
\caption{Frequently used notations.} \label{notations}
\vspace*{-0.3cm}
\small
\begin{tabular}{c|l}
\hline
\textbf{Notation}                            & \textbf{Definition}                                                \\ \hline
$N_G(x)$ or $N(x)$                  & The neighbor nodes of $x$ in $G$.                          \\
$d_x(G)$ or $d_x$                   & The degree of $x$ in $G$.                                  \\
$G(X)=(X^U,X^V,E(X))$               & The subgraph induced by $X$.                               \\
$\vec{G}$                           & An orientation of $G$.                                     \\
$\vec{d}_x(\vec{G})$ or $\vec{d}_x$ & The indegree of $x$ in $\vec{G}$.                          \\
$s\rightsquigarrow t$               & A path from $s$ to $t$.                                    \\
$r_{\alpha}(x)$ and $r_{\beta}(x)$  & The $\alpha$-rank and $\beta$-rank of $x$.                 \\
$p$                                 & The largest integer $p$ such that $D_{p,p}\neq \emptyset$. \\ \hline
\end{tabular}
\end{table}

By assigning each edge a specific direction, we can convert the undirected bipartite graph $G = (U, V, E)$ into a directed bipartite graph called an \emph{orientation} of $G$, denoted $\vec{G}=(U, V, \vec{E})$, where $\vec{E}$ contains the directed edges. In this orientation, the indegree of node $x$, denoted $\vec{d}_x(\vec{G})$ (or $\vec{d}_x$ for simplicity), counts its incoming edges. A \emph{path} $s \rightsquigarrow t$ in an orientation is a sequence of nodes $x_0, x_1, \ldots, x_l$, where $(x_{i-1}, x_i)\in \vec{E}$ for $i = 1,\ldots,l$, and the length of this path is $l$. If a path $x \rightsquigarrow y$ exists, we say $x$ can \emph{reach} $y$. {\color{\mycolorc}The frequently used notations are summarized in \tabref{notations}.} Next, we introduce the definition of the $(\a,\b)$-dense subgraph \cite{ddbipartite}. 

\begin{Definition} \label{ab-dense-subgraph}
    \cite{ddbipartite} \mdbf{($(\a,\b)$-dense subgraph)} Given a bipartite graph $G=(U, V, E)$, two non-negative integers $\a$ and $\b$, let $\vec{G}$ be an orientation of $G$, and let $S=\{u \in U| \vec{d}_u< \a\} \cup \{v \in V| \vec{d}_v < \b \}$ and $T=\{u \in U| \vec{d}_u> \a\} \cup\{v \in V| \vec{d}_v> \b \}$. If there is no path $s \rightsquigarrow t$ in $\vec G$ with $s\in S$ and $t\in T$, then the $(\a,\b)$-dense subgraph $G(D_{\a,\b})$ is induced by $D_{\a,\b}=T \cup \{x| x \text{ can reach a node in } T~in~\vec G\}$. 
\end{Definition}

{\color{\mycolora}
\begin{Theorem} \label{The:edges} \cite{ddbipartite}
    The $(\a,\b)$-dense subgraph $D_{\a,\b}$ has the following two properties: 
    (1) \textbf{(inside dense)} for any non-empty $X \subseteq D_{\a,\b}$, the removal of $X$ from $D_{\a,\b}$ results in a deletion of more than $\a\cdot|X^U|+\b\cdot|X^V|$ edges; (2) \textbf{(outside sparse)} for any subset $Y \subseteq (U\cup V)\setminus D_{\a,\b}$, the inclusion of $Y$ into $D_{\a,\b}$ leads to an increase of at most $\a\cdot|Y^U|+\b\cdot|Y^V|$ edges.
\end{Theorem}
}

\begin{figure}[t]
  \captionsetup[subfigure]{justification=centering}
  \begin{subfigure}{0.95\linewidth}
    \centering
    \includegraphics[width=1.0\linewidth]{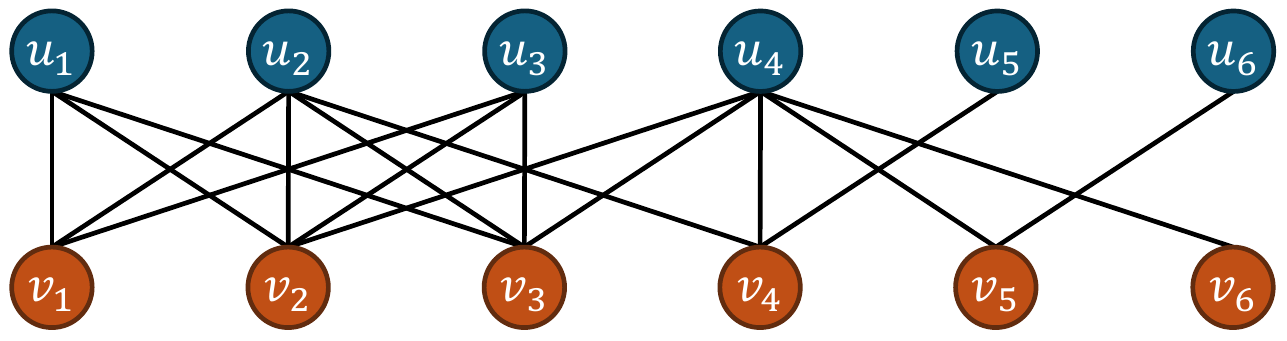}
    \subcaption{A bipartite graph $G=(U,V,E)$.} \label{example-graph-1}
  \end{subfigure}
  \begin{subfigure}{0.95\linewidth}
    \centering
    \includegraphics[width=1.0\linewidth]{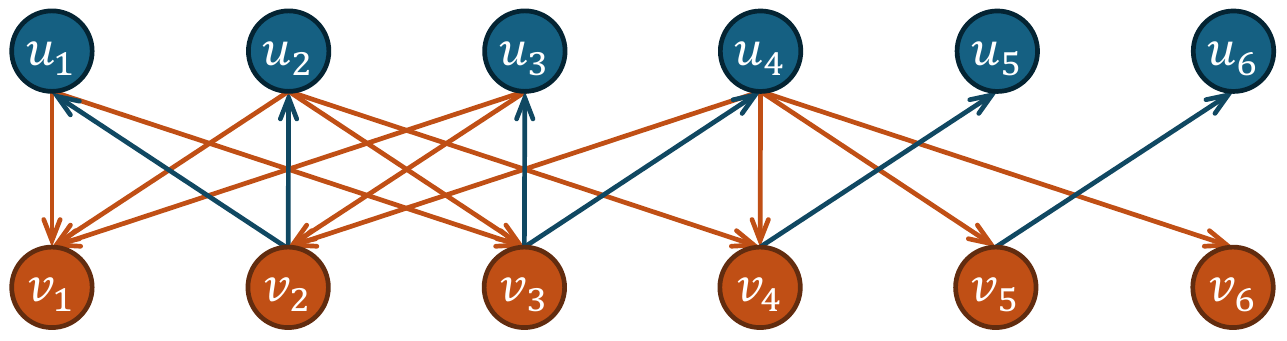}
    \subcaption{An orientation $\vec{G}=(U,V,\vec{E})$.} \label{example-graph-2}
  \end{subfigure}
  \caption{An example graph $G=(U,V,E)$ and its orientation.} \label{example-graph}
\end{figure}

{\color{\mycolora}The $(\alpha,\beta)$-dense subgraph is characterized as being {\bf dense inside} and {\bf sparse outside}~\cite{ddbipartite}. Intuitively, the parameters $\alpha$ and $\beta$ serve as thresholds that control the density of the two partitions: increasing $\alpha$ directly raises the density requirement for the upper partition $U$, and increasing $\beta$ does the same for the lower partition $V$. As a result, each $(\alpha,\beta)$ pair corresponds to a subgraph $D_{\alpha,\beta}$ with a distinct density, enabling users to flexibly adjust the density of the resulting subgraph to their specific needs and separately control the density requirements for the two partitions.
}


\begin{Example}
    As illustrated in \figref{example-graph-2}, $\vec{G}$ is an orientation of the bipartite $G$ depicted in \figref{example-graph-1}. Let $\a=1$ and $\b=2$. According to \defref{ab-dense-subgraph}, we have $S=\{v_5, v_6\}$ and $T=\{v_1\}$. Since no path exists from any node in $S$ to any node in $T$, we can conclude that $D_{1,2}$ contains $T$ and all nodes that can reach $T$, specifically $\{u_1, u_2, u_3, u_4, v_1, v_2, v_3\}$. Using the same approach, we have $D_{1,1} = \{u_1, u_2, u_3, u_4, v_1, v_2, v_3, v_4\}$ and $D_{1,3} = \emptyset$. {\color{\mycolora}Compared with $D_{1,1}$, $D_{1,2}$ excludes the relatively sparse node $v_4$ due to the higher $\beta$ value, resulting in a denser subgraph. Further increasing $\beta$ to 3 yields no subgraph satisfying the density requirement, thus, $D_{1,3} = \emptyset$.}
\end{Example}

In practical applications (such as e-commerce recommendation and fraud detection), it is often necessary to frequently query $D_{\alpha,\beta}$ for different combinations of $(\alpha,\beta)$ parameters. This necessitates the design of efficient query processing algorithms. Moreover, real-world bipartite graphs are typically dynamic, with frequent edge insertions and deletions. In such scenarios, efficiently computing $D_{\alpha,\beta}$ while keeping up with graph updates is essential for real-time responsiveness. Motivated by these, we formulate the problem studied in this paper as follows.


\stitle{Problem definition:} Given a bipartite graph $G$, we define a \emph{$D_{\a,\b}$-query} as the computation of $D_{\a,\b}$ for given two non-negative integers $\a$ and $\b$. Our problem focuses on efficiently processing $D_{\a,\b}$-queries in both static and dynamic graphs. 


\stitle{Challenges.} A straightforward approach to answering a $D_{\a,\b}$-query is to directly invoke the state-of-the-art algorithm \kw{DSS++} proposed in \cite{ddbipartite}, which we refer to as the \kw{Online} algorithm in this paper. However, this algorithm incurs a worst-case time complexity of $O(|E|^{1.5})$ since it requires performing a maximum flow computation \cite{ddbipartite}, making it impractical for large-scale graphs. To address this inefficiency, a promising approach is to design specialized index structures for efficient query processing, which presents two key challenges: \emph{(1) Query-efficient and space-efficient index design:} To achieve optimal query time, a naive indexing approach is to precompute and store all non-empty $D_{\a,\b}$ subgraphs. However, this approach incurs prohibitive space complexity of $O(|V|^3)$, as there can be $O(|V|^2)$ valid $(\a,\b)$ pairs, each requiring $O(|V|)$ storage space. The challenge lies in how to leverage the inherent relationships among $D_{\a,\b}$ subgraphs to compactly organize node information in the index, such that it supports both minimal space usage and optimal query performance. \emph{(2) Efficient index maintenance for dynamic graphs:} Real-world bipartite graphs frequently evolve through edge insertions and deletions. In such dynamic scenarios, the index must be efficiently maintained after each update to guarantee real-time query responsiveness. A naive approach that recomputes the entire index from scratch following each update is computationally prohibitive. While there are dense subgraph maintenance methods in unipartite graphs \cite{dd}, they cannot be directly applied to bipartite graphs due to the two-dimensional nature of $(\a,\b)$-dense subgraphs. Unlike unipartite settings, bipartite graphs require maintaining significantly more subgraphs due to the dual-parameter definition and must carefully handle the asymmetry between the upper and lower node partitions. These unique challenges make index maintenance in bipartite graphs substantially more complicated.

\section{A NOVEL INDEX: \textsf{BD-Index}}
In this section, we propose a novel index called \kw{BD-Index} {\color{\mycolorb}(short for {\bf B}ipartite {\bf D}ense Subgraph Index)}, which achieves optimal query processing time of $O(|D_{\a,\b}|)$. Furthermore, \kw{BD-Index} requires only $O(|E|)$ space, linear to the graph size, making it both time-efficient and space-efficient for query processing.

\subsection{Structure of \textsf{BD-Index}}

As discussed previously, we can compute all non-empty $D_{\alpha,\beta}$ subgraphs by executing the \kw{DSS++} algorithm~\cite{ddbipartite} and storing each $D_{\alpha,\beta}$ individually to construct a basic index structure. Although this naive approach achieves optimal query processing time $O(|D_{\alpha,\beta}|)$, its space complexity becomes impractical for large-scale graphs. Specifically, each $D_{\alpha,\beta}$ requires $O(|V|)$ space, and there exist $O(|V|^2)$ distinct $(\alpha,\beta)$ pairs that produce non-empty $D_{\alpha,\beta}$ subgraphs~\cite{ddbipartite}, leading to an overall storage requirement of $O(|V|^3)$. To address this space limitation, we exploit the hierarchical property of $D_{\alpha,\beta}$ subgraphs, as formalized in the following theorem.

\begin{Theorem} \label{hierarchy}
    \cite{ddbipartite} \mdbf{(Hierarchical property)} Given a graph $G$, for $\a^+\ge\a$ and $\b^+\ge\b$, we have $D_{\a^+,\b^+}\subseteq D_{\a,\b}$.
\end{Theorem}

According to the hierarchical property theorem, if a node $x$ belongs to $D_{\a^+,\b^+}$, it must necessarily belong to $D_{\a,\b}$. In the straightforward index, $x$ would be redundantly stored in both $D_{\a^+,\b^+}$ and $D_{\a,\b}$. To optimize storage, we can store $x$ only in $D_{\alpha^+,\beta^+}$. When processing queries for $D_{\alpha,\beta}$, we can directly incorporate $D_{\alpha^+,\beta^+}$ into the result set, thereby eliminating storage redundancy while maintaining query correctness. Based on the above rationale, we define $\a$-rank and $\b$-rank.

\begin{Definition}
    \mdbf{($\a$-rank and $\b$-rank)} Given a graph $G$ and a value $\a$, the $\a$-rank of a node $x\in (U\cup V)$, denoted by $r_\a(x)$, is the maximum integer $k$ such that $x\in D_{\a,k}$ if $x\in D_{\a,0}$, and -1 otherwise. Similarly, for a value $\b$, the $\b$-rank of $x$, denoted by $r_\b(x)$, is the maximum integer $k$ such that $x \in D_{k,\b}$ if $x\in D_{0,\b}$, and -1 otherwise.
\end{Definition}

With $\a$-rank and $\b$-rank defined, $D_{\a,\b}$ can be computed using the following theorem.

\begin{Theorem} \label{rank-theorem}
    Given a graph $G$ and two non-negative integers $\a$ and $\b$, we have: $D_{\a,\b} = \{x|r_\a(x) \geq \b\} = \{x|r_\b(x) \geq \a\}$.
\end{Theorem}

\begin{proof}
    By definition, the $\a$-rank $r_\a(x)$ implies $x \in D_{\a,r_\a(x)}$ and $x \notin D_{\a,r_\a(x)+1}$. Combining this with \theref{hierarchy}, if $r_\a(x) \geq \b$, then $x \in D_{\a,\b}$. Conversely, if $r_\a(x) < \b$, it follows that $x \notin D_{\a,\b}$. This establishes the equivalence $r_\a(x) \geq \b \Leftrightarrow x \in D_{\a,\b}$, which yields $D_{\a,\b} = \{x|r_\a(x) \geq \b\}$. Through symmetric reasoning, we further derive $D_{\a,\b} = \{x|r_\b(x) \geq \a\}$.
\end{proof}

According to \theref{rank-theorem}, for a fixed $\a$, we can construct a node list by sorting all nodes in ascending order of their $r_\a$. Then, given arbitrary $\b$, $D_{\a,\b}$ consists of all nodes from the first node in the list with $r_\a \geq \b$ to the end of the node list. Symmetrically, when fixing $\beta$ and sorting by $r_\beta$, we obtain an analogous property for $D_{\alpha,\beta}$. Based on this idea, we propose \kw{BD-Index} as follows.

\begin{Definition}
    \mdbf{(\textsf{BD-Index})} Let $p$ be the maximum integer such that $D_{p,p} \neq \emptyset$. The \kw{BD-Index} comprises two symmetric components:
    
    (1) $\mathbb{I}_{BD}^U$: For each $\a = 0, \ldots, p$, the index stores a node list that contains all nodes with $r_\a \geq \a$, sorted in ascending order of $r_\a$. For each $\b = \a, \ldots, \max_{x \in (U \cup V)} r_\a(x)$, $\mathbb{I}_{BD}^U[\a][\b]$ points to the first node in the node list with $r_\a \geq \b$. 
    
    (2) $\mathbb{I}_{BD}^V$: for each $\b = 0, \ldots, p$, the index stores a node list that contains all nodes with $r_\b > \b$, sorted in ascending order of $r_\b$. For each $\a = \b + 1, \ldots, \max_{x \in (U \cup V)} r_\b(x)$, $\mathbb{I}_{BD}^V[\b][\a]$ points to the first node in the node list with $r_\b \geq \a$.
\end{Definition}

\begin{figure*}[t]
  \captionsetup[subfigure]{justification=centering}
  \begin{subfigure}{0.45\linewidth}
    \centering
    \includegraphics[width=1.0\linewidth]{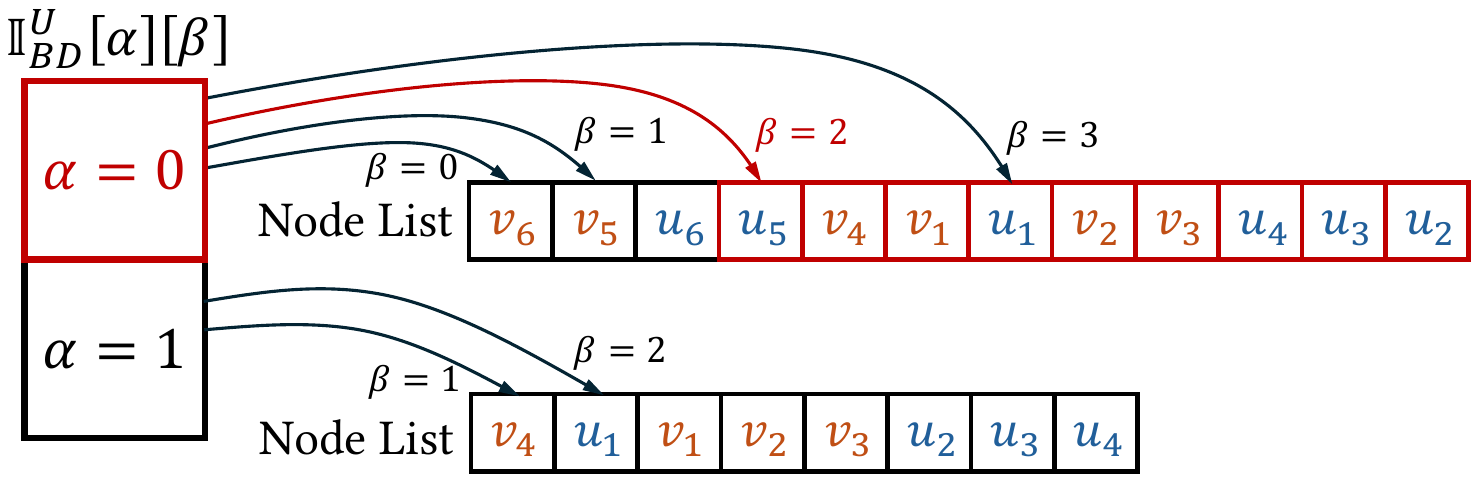}
    \subcaption{$\mathbb{I}_{BD}^U$.} \label{bd-index-example-1}
  \end{subfigure}
  \hspace*{0.3cm}
  \begin{subfigure}{0.45\linewidth}
    \centering
    \includegraphics[width=1.0\linewidth]{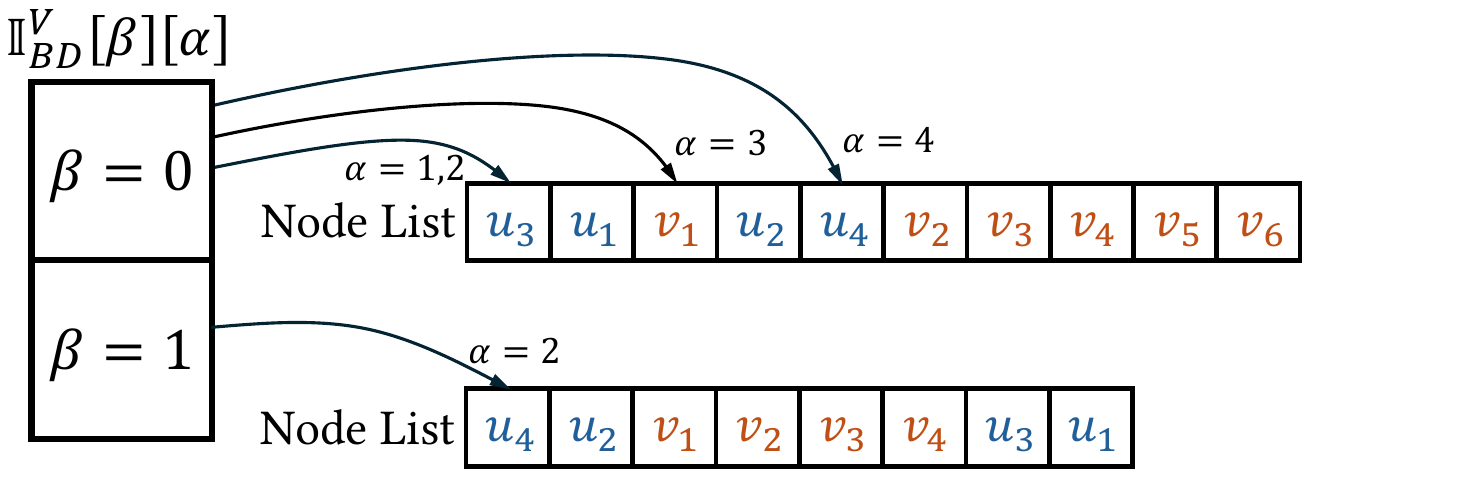}
    \subcaption{$\mathbb{I}_{BD}^V$.} \label{bd-index-example-2}
  \end{subfigure}
  \caption{Example of \kw{BD-Index} and querying $D_{0,2}$.} \label{bd-index-example}
\end{figure*}

\begin{Example}
    The \kw{BD-Index} of the graph $G$ in \figref{example-graph-1} is shown in \figref{bd-index-example}. In this \kw{BD-Index}, we have $p = 1$, resulting in both $\mathbb{I}_{BD}^U$ and $\mathbb{I}_{BD}^V$ containing $(p+1) = 2$ node lists. Taking the case of $\a = 1$ as an example, we have the following rank values: $r_\a(u_5) = r_\a(u_6) = -1$, $r_\a(v_5) = r_\a(v_6) = 0$, $r_\a(v_4) = 1$, $r_\a(u_1) = r_\a(u_2) = r_\a(u_3) = r_\a(u_4) = r_\a(v_1) = r_\a(v_2) = r_\a(v_3) = 2$. The node list for $\a = 1$ includes all nodes with $r_\a \geq \a$, as illustrated in \subfigref{bd-index-example-1}. Since $\max_{x \in (U \cup V)} r_\a(x) = 2$, $\mathbb{I}_{BD}^U[\a]$ contains two pointers to the node list: $\mathbb{I}_{BD}^U[\a][1]$ and $\mathbb{I}_{BD}^U[\a][2]$, which point to $v_4$ and $u_1$, respectively.
\end{Example}

\subsection{Query processing}
Here, we present how \kw{BD-Index} supports optimal-time query processing. The query processing algorithm, \kw{Query-BD-Index}, is detailed in \algref{query-bd-index}. The index component $\mathbb{I}_{BD}^U$ processes queries where $\a \leq \b$ (lines 1-3), while $\mathbb{I}_{BD}^V$ handles queries where $\a > \b$ (lines 4-6). For $\a \leq \b$, the algorithm first checks whether the queried $\a$ and $\b$ fall within the valid range to determine if $D_{\a,\b}$ is empty (line 2). If $D_{\a,\b}$ is non-empty, the algorithm retrieves all nodes from $\mathbb{I}_{BD}^U[\a][\b]$ to the end of the node list and adds them into $D_{\a,\b}$ (line 3). For $\a > \b$, the algorithm computes $D_{\a,\b}$ analogously (lines 5-6). 

\begin{algorithm}[t]
    \caption{\invoke{Query-BD-Index}$(\mathbb{I}_{BD},\a,\b)$} \label{query-bd-index}
    \small
    \KwIn{The \kw{BD-Index} $\mathbb{I}_{BD}$, and the query integers $\a$ and $\b$.}
    \KwOut{$D_{\a,\b}$.}
    \If{$\a\le\b$}{
        \lIf{$\a\ge\mathbb{I}_{BD}^U.size$ or $\b\ge\mathbb{I}_{BD}^U[\a].size$}{
            $D_{\a,\b}\gets\emptyset$
        }
        \lElse{
            \scalebox{0.97}{$D_{\a,\b}\gets$ nodes from $\mathbb{I}_{BD}^U[\a][\b]$ to the end of the node list}
        }
    }
    \Else{
        \lIf{$\b\ge\mathbb{I}_{BD}^V.size$ or $\a\ge\mathbb{I}_{BD}^V[\b].size$}{
            $D_{\a,\b}\gets\emptyset$
        }
        \lElse{
            \scalebox{0.97}{$D_{\a,\b}\gets$ nodes from $\mathbb{I}_{BD}^V[\b][\a]$ to the end of the node list}
        }
    }
    \Return{$D_{\a,\b}$}\;
\end{algorithm}

\begin{Example}
    Based on \kw{BD-Index} illustrated in \figref{bd-index-example}, the \kw{Query-BD-Index} processes the $D_{0,2}$ query as follows. Since $\a = 0 \leq \b = 2$, the algorithm employs the index component $\mathbb{I}_{BD}^U$. First, it verifies that $\a < \mathbb{I}_{BD}^U.size = 2$ and $\b < \mathbb{I}_{BD}^U[\a].size = 4$, confirming that $D_{\a,\b} \neq \emptyset$. The algorithm then retrieves all nodes from $\mathbb{I}_{BD}^U[\a][\b] = u_5$ to the end of the node list, which includes $\{u_5, v_4, v_1, u_1, v_2, v_3, u_4, u_3, u_2\}$, as the query result $D_{\a,\b}$.
\end{Example}

Next, we prove the correctness of \kw{Query-BD-Index} and establish its optimal query time complexity.

\begin{Theorem}
    \kw{Query-BD-Index} can correctly retrieve $D_{\a,\b}$ within optimal query time $O(|D_{\a,\b}|)$.
\end{Theorem}

\begin{proof}
    When $\a \leq \b$, according to the definition of \kw{BD-Index}, the nodes from $\mathbb{I}_{BD}^U[\a][\b]$ to the end of the node list include all nodes with $r_\a \geq \b$. By \theref{rank-theorem}, $D_{\a,\b}$ can be correctly retrieved. Using a similar proof technique, we can establish the correctness for the case when $\a > \b$. 

    Regarding the query complexity, lines 1-2 and lines 4-5 of the algorithm execute in constant time. In lines 3 and 6, the algorithm perform a  traversal of $D_{\a,\b}$, requiring $O(|D_{\a,\b}|)$ time. Therefore, the overall query complexity is $O(|D_{\a,\b}|)$.
\end{proof}

The optimal query time $O(|D_{\alpha,\beta}|)$ of \kw{BD-Index} guarantees efficient processing with minimal overhead, enabling result retrieval independent of the graph size.

\subsection{Space complexity of \textsf{BD-Index}}

In this subsection, we analyze the space complexity of \kw{BD-Index}. We begin by introducing the following lemma.

\begin{Lemma} \label{degree}
    Given a bipartite graph $G$ and $D_{\a,\b}$, for any node $u \in D_{\a,\b}^U$, we have $d_u(G) > \a$, and for any node $v \in D_{\a,\b}^V$, we have $d_v(G) > \b$.
\end{Lemma}

The correctness of this lemma can be directly derived by definition. Below, we prove that \kw{BD-Index} achieves $O(|E|)$ space complexity.

{\color{\mycolorc}
\begin{Theorem}
    Given a graph $G$, the space complexity of its \kw{BD-Index} is $\Theta(|E|)$.
\end{Theorem}
}

\begin{proof}
    We first prove that the space complexity of $\mathbb{I}_{BD}^U$ is $O(|E|)$. To begin, we show that the $(p+1)$ node lists in $\mathbb{I}_{BD}^U$ occupy $O(|E|)$ space. For a fixed $\a$, by definition, the corresponding node list contains the nodes in $D_{\a,\a}$. The total number of nodes across all node lists in $\mathbb{I}_{BD}^U$ is $\sum_{\a=0}^{p} |D_{\a,\a}| \leq \sum_{\a=0}^{p} |\{x | d_x(G) > \a\}| \leq \sum_{x \in U \cup V} d_x(G) = 2|E|$, where the first inequality follows from \lemref{degree}. Therefore, the space occupied by the $(p+1)$ node lists in $\mathbb{I}_{BD}^U$ is $O(|E|)$.  

    Next, we prove that the space occupied by all pointers $\mathbb{I}_{BD}^U[\cdot][\cdot]$ is also $O(|E|)$. Let $r_\a^{\max} = \max_{x \in U \cup V} r_\a(x)$. Fixing a specific $\a$, \lemref{degree} implies that the nodes in $D_{\a,r_\a^{\max}}^V$ have degrees greater than $r_\a^{\max}$, which leads to $|D_{\a,r_\a^{\max}}^U| > r_\a^{\max}$. Thus, by \lemref{degree}, we obtain $r_\a^{\max} < |D_{\a,r_\a^{\max}}^U| \le |\{u \in U | d_u > \a\}|$. By the definition of $\mathbb{I}_{BD}^U$, the number of pointers in $\mathbb{I}_{BD}^U$ is $\sum_{\a=0}^{p} (r_\a^{\max} - \a + 1) \leq \sum_{\a=0}^{p} (r_\a^{\max} + 1) \leq \sum_{\a=0}^{p} |\{u \in U | d_u > \a\}| \leq \sum_{u \in U} d_u = |E|$. Thus, the space occupied by the pointers in $\mathbb{I}_{BD}^U$ is also $O(|E|)$. Consequently, the total space complexity of $\mathbb{I}_{BD}^U$ is $O(|E|)$.  

    Using a similar approach, we can prove that $\mathbb{I}_{BD}^V$ also occupies $O(|E|)$ space. Therefore, the total space complexity of $\mathbb{I}_{BD}$ is $O(|E|)$.

    {\color{\mycolorc}
    Next, we prove that the space complexity of \kw{BD-Index} is $\Omega(|E|)$. It suffices to show that the number of nodes contained in $\mathbb{I}_{BD}^U$ is $\Omega(|E|)$, i.e., to prove that $\sum_{k=0}^p |D_{k,k}| = \Omega(|E|)$. First, define $E_k^\Delta = |E(D_{k,k})| - |E(D_{k+1,k+1})|$. According to \theref{The:edges}, we have $E_k^\Delta \le (k+1)|D_{k,k} \setminus D_{k+1,k+1}|$ by letting $Y = D_{k,k} \setminus D_{k+1,k+1}$. Then, we obtain $\sum_{k=0}^p |D_{k,k}| = \sum_{k=0}^p \sum_{t=k}^p |D_{t,t} \setminus D_{t+1,t+1}| = \sum_{t=0}^p (t+1)|D_{t,t} \setminus D_{t+1,t+1}| \ge \sum_{t=0}^p E_t^\Delta = |E|$, which implies $\sum_{k=0}^p |D_{k,k}| = \Omega(|E|)$. Combining the above results, we conclude that the space complexity of \kw{BD-Index} is $\Theta(|E|)$.
    }
\end{proof}

{\color{\mycolorc}The linear space complexity $\Theta(|E|)$ of \kw{BD-Index} indicates that its space usage grows asymptotically proportional to the number of edges, ensuring highly stable and predictable memory consumption. Our experiments further show that the actual space usage is almost equal to $8|E|$ bytes in practice. This property enables \kw{BD-Index} to efficiently handle large bipartite graphs while maintaining fast query performance.
}

\subsection{Index construction}

In this subsection, we introduce the construction algorithm for \kw{BD-Index}. The construction relies on three existing algorithms for computing $(\a,\b)$-dense subgraphs: \kw{DSS++}, \kw{Divide-a}, and \kw{Divide-b} \cite{ddbipartite}. The \kw{DSS++} algorithm computes a single $D_{\a,\b}$ subgraph. Given a bipartite graph $G$ and parameters $\a$ and $\b$, it constructs a network flow model and performs a minimum cut computation to obtain $D_{\a,\b}$, with the worst-case time complexity of $O(|E|^{1.5})$. The \kw{Divide-a} algorithm computes all non-empty $D_{\a,\b}$ for a fixed $\a$ across all possible $\b$ values. It utilizes \kw{DSS++} along with a divide-and-conquer strategy for pruning, achieving a time complexity of $O(|E|^{1.5} \log|U\cup V|)$. Its symmetric counterpart, \kw{Divide-b}, computes all non-empty $D_{\a,\b}$ subgraphs for a fixed $\b$ across all $\a$ values.

With the three algorithms \kw{DSS++}, \kw{Divide-a}, and \kw{Divide-b}, we propose the \kw{Build-BD-Index} algorithm for constructing \kw{BD-Index}, as shown in \algref{build-bd-index}. The algorithm proceeds as follows. First, it computes $p$ as the maximum integer satisfying $D_{p,p} \neq \emptyset$ (line 1), which can be achieved by performing a binary search with invocations to \kw{DSS++} \cite{ddbipartite}. Then, for each $\a$, the algorithm first invokes \kw{Divide-a} to identify all non-empty $D_{\a,\b}$ subgraphs across all $\b$ values (line 3). Based on these results, it determines the $r_\a$ value for each node (lines 4-5), sorts the nodes according to $r_\a$, and constructs the corresponding node list (line 6). The algorithm subsequently sets the pointers $\mathbb{I}_{BD}^U[\a][\cdot]$ in $\mathbb{I}_{BD}^U$ using this node list (lines 7-8). Once $\mathbb{I}_{BD}^U[\a][\cdot]$ is constructed for all $\a$, the algorithm performs a symmetric procedure for each $\b$ (lines 9-15). Finally, it returns $\mathbb{I}_{BD} = \mathbb{I}_{BD}^U \cup \mathbb{I}_{BD}^V$ (line 16).

The correctness of \kw{Build-BD-Index} follows directly from the correctness of its constituent algorithms \kw{Divide-a} and \kw{Divide-b}. For the time complexity of the \kw{Build-BD-Index} algorithm, the most time-consuming part is the invocation of the \kw{Divide-a} and \kw{Divide-b} algorithms. Since each invocation takes $O(|E|^{1.5} \log|U \cup V|)$ time and there are $O(p)$ invocations in total, the overall time complexity is $O(p \cdot |E|^{1.5}\cdot \log|U \cup V|)$. {\color{\mycolorc}As established in \cite{ddbipartite, pseudoarboricity}, the value of $p$ is typically a small constant in real-world graphs. While its worst-case value is $p \le \sqrt{|E|}/2$, this bound is only achieved in complete bipartite graphs, which are uncommon in practice. A similar parameter-dependent complexity also appears in bipartite core decomposition, where the time complexity is $O(\delta \cdot |E|)$ and $\delta$ is the largest integer such that the $(\delta,\delta)$-core is non-empty~\cite{abcore2}. Since $p$ tends to be small in real-world graphs, efficient index construction on large-scale bipartite networks is ensured.}

\begin{algorithm}[t]
    \caption{\invoke{Build-BD-Index}$(G)$} \label{build-bd-index}
    \small
    \KwIn{A bipartite graph $G=(U,V,E)$.}
    \KwOut{\kw{BD-Index}.}
    $p\gets$ the maximum integral such that $D_{p,p}\neq\emptyset$\;
    \For{$\a=0,1,2,\ldots,p$}{
        Invoke algorithm \kw{Divide-a} in \cite{ddbipartite} to compute the non-empty $D_{\a,\b}$ for all $\b$\;
        \ForEach{$x\in U\cup V$}{$r_\a(x)\gets$ the maximum $\b$ such that $x\in D_{\a,\b}$\;}
        Create a node list containing all nodes with $r_{\a} \geq \a$, and sort them in ascending order of $r_{\a}$\;
        \For{$\b=\a,\a+1,\ldots,\max_{x\in (U\cup V)}r_{\a}(x)$}{
            $\mathbb{I}_{BD}^U[\a][\b]\gets$ the first node in the node list with $r_{\a} \geq \b$\;
        }
    }
    \For{$\b=0,1,2,\ldots,p$}{
        Invoke algorithm \kw{Divide-b} in \cite{ddbipartite} to compute the non-empty $D_{\a,\b}$ for all $\a$\;
        \ForEach{$x\in U\cup V$}{$r_\b(x)\gets$ the maximum $\a$ such that $x\in D_{\a,\b}$\;}
        Create a node list containing all nodes with $r_{\b} > \b$, and sort them in ascending order of $r_{\b}$\;
        \For{$\a=\b+1,\b+2,\ldots,\max_{x\in (U\cup V)}r_{\b}(x)$}{
            $\mathbb{I}_{BD}^V[\b][\a]\gets$ the first node in the node list with $r_{\b} \geq \a$\;
        }
    }
    \Return{$\mathbb{I}_{BD}=\mathbb{I}_{BD}^U\cup\mathbb{I}_{BD}^V$}\;
\end{algorithm}

{\color{\mycolora}\stitle{Discussion: index construction and practical efficiency.}
While building the \kw{BD-Index} incurs construction cost, this overhead is quickly offset in practical scenarios for several reasons:\\
(1) \textbf{Fast construction and cost amortization.}
Our experiments Exp-1 and Exp-4 indicate that the index construction time approximates processing 1,000 online queries. However, once built, the index can efficiently support arbitrary $(\alpha,\beta)$ queries. The number of possible queries is much more than the initial cost required to build the index (e.g., for the \kw{HE} dataset in our experiments, there are 549,464 non-empty $D_{\alpha,\beta}$ subgraphs), demonstrating that the index covers a query space far exceeding the number of queries needed to amortize its construction. In addition, in practical applications, query workloads can be very frequent ($\approx$100,000 queries per second \cite{DBLP:conf/cikm/ZhangTWW24}), and the total number of queries soon surpasses this initial cost ($\approx$1,000 online queries), making the construction overhead negligible in the long run.\\
(2) \textbf{Significant query efficiency.}
The online algorithm exhibits slow query response times (up to 60 seconds per query, as shown in Exp-1), whereas our index achieves a maximum response time of only 0.018 seconds. In many real-world applications, the required response time must be returned in real time (e.g., within 0.5 seconds) \cite{amazon,DBLP:conf/cikm/ZhangTWW24,DBLP:conf/sigmod/KersbergenSS22}, highlighting that the online approach is far too slow to be practical. In contrast, our index not only satisfies these stringent latency requirements but also achieves optimal $O(|R|)$ query complexity, allowing it to directly output the nodes contained in the identified community.\\
(3) \textbf{Efficient maintenance for dynamic graphs.}
Our proposed maintenance algorithms in subsequent sections ensure that \kw{BD-Index} is incrementally updatable as the graph evolves through edge insertions and deletions. This design transforms the construction cost into a one-time investment, since ongoing maintenance can be performed efficiently without the need for complete reconstruction.}

{\color{\mycolorb}
\stitle{Remark: extending \kw{BD-Index} to the $(\alpha,\beta)$-core community query problem.} In addition to supporting $(\alpha,\beta)$-dense subgraphs, \kw{BD-Index} can also be adapted to $(\alpha,\beta)$-core community query problem \cite{abcore1,abcore2,abcore3} by defining analogous rank notions and reusing the index framework. Specifically, the $\alpha$-rank of a node $x$ is defined as the largest integer $k$ such that $x$ belongs to the $(\alpha,k)$-core (with $\beta$-rank defined symmetrically). Using these ranks, two symmetric components similar to \kw{BD-Index} can be constructed: for each $\alpha$ (resp. $\beta$) value, we construct node lists sorted by $\alpha$-rank (resp. $\beta$-rank) and pointers $\mathbb{I}_{core}^U[\alpha][\beta]$ and $\mathbb{I}_{core}^V[\beta][\alpha]$ to locate the first qualifying node in each list. This adaptation is feasible because $(\alpha,\beta)$-core also exhibits a hierarchy structure similar to $(\alpha,\beta)$-dense subgraphs \cite{abcore2}, making the indexing approach naturally applicable.
}

\section{SPACE-EFFICIENT INDEX MAINTENANCE}
This section aims to develop space-efficient maintenance algorithms for \kw{BD-Index} under dynamic graph updates, preserving its advantage of linear space complexity $O(|E|)$. A straightforward method is to reconstruct \kw{BD-Index} from scratch after each edge insertion or deletion. However, such recomputation incurs $O(p \cdot |E|^{1.5} \cdot \log|U \cup V|)$ time, which is prohibitive for real-world graphs with frequent edge updates. To address this bottleneck, we first formalize the update theorems, followed by proposing two incremental maintenance algorithms: \textsf{BD-Insert-S} and \textsf{BD-Delete-S}.  

\subsection{Update theorems for $\a$-rank and $\b$-rank}
Since $\mathbb{I}_{BD}^U$ and $\mathbb{I}_{BD}^V$ are two symmetric structures, we focus on maintaining $\mathbb{I}_{BD}^U$ for brevity, and the maintenance of $\mathbb{I}_{BD}^V$ can be implemented symmetrically. In essence, maintaining $\mathbb{I}_{BD}^U$ requires updating the $r_\a$ values of nodes, with emphasis on how $r_\a$ evolve after edge insertions or deletions. To formalize this, we present the following update theorems for $r_\a$ that addresses edge insertion and deletion. \submitversion{Due to space limitations, some proofs are provided in the full version of our paper, which can be found at \url{https://anonymous.4open.science/r/bd-index-F7E2}.} 

\begin{Theorem} \label{insertion-update-theorem}
    \mdbf{(Insertion Update Theorem)}\comment{ Given a graph $G$ and an alpha value $\a$, for the insertion of an edge $(u,v)$, the update of $r_{\alpha}$ follows these rules:
    \begin{enumerate}
        \item If $d_u < \alpha$, then the $r_{\alpha}$ values of all nodes remain unchanged.  
        \item Otherwise, let $x \in N(u)$ be the node whose $r_{\alpha}$ value is the $(\alpha+1)$-th largest among $N(u)\cup\{v\}$, and let $\beta = \min\{r_{\alpha}(x), r_{\alpha}(v)\}$.  
        \begin{itemize}
            \item All nodes with $r_{\alpha} = \beta$ may have their $r_{\alpha}$ updated to $\beta + 1$.  
            \item If $r_{\alpha}(u) < \beta + 1$, then $r_{\alpha}(u)$ may also be updated to $\beta + 1$.  
            \item The $r_{\alpha}$ values of all other nodes remain unchanged.  
        \end{itemize}
    \end{enumerate}}
    Given a graph $G$ and an integer $\a$, for an insertion edge $(u,v)$, let $r^N$ be the $r_\a$ value of the $(\a+1)$-th highest-ranked node in $N(u)\cup\{v\}$ (specifically, if $|N(u)| < \a$, then set $r^N = -1$), and let $\b = \min\{r^N, r_{\alpha}(v)\}$. Then, for all nodes $x \in V$:  
    \begin{enumerate}
       \item If $r_\a(x) = \b$, then $r_\a(x)$ may be updated to $\b+1$.  
        \item If $r_\a(x) \neq \b$, then $r_\a(x)$ remains unchanged.  
    \end{enumerate}
\end{Theorem}
\fullversion{
\begin{proof}
    We divide the proof into three cases.

    \textit{Case 1:} When $|N(u)| < \a$, both $r_N$ and $\b$ equal $-1$. Let $\vec{E}$ be an egalitarian orientation (defined in \defref{egalitarian-orientation}, which we will introduce later) of the given graph $G$ with respect to $\alpha$. We directly insert the edge $(u, v)$ into $\vec{E}$ with direction $(v, u)$. According to the definition of egalitarian orientation, the resulting $\vec{E}$ remains egalitarian. Since the reachability of all nodes remains unchanged, the $\alpha$-rank of each node also remains unchanged, which proves the case.

    \textit{Case 2:} When $|N(u)| \ge \a$ and $r_\a(u) > r_\a(v)$, we have $r^N = r_\a(u)$ and $\b = r_\a(v)$ by \theref{u-update-theorem} (we will prove later). For any $\b' \in [-1, \b]\cup [\b+2, +\infty)$, let $\vec{E}$ be an orientation that satisfies the condition in \defref{ab-dense-subgraph} with alpha value as $\a$ and beta value as $\b'$. We insert the edge $(u,v)$ into $\vec{E}$ with direction towards $v$. It is easy to verify that the updated $\vec{E}$ still satisfies the condition in \defref{ab-dense-subgraph}, so $D_{\a,\b'}$ remains unchanged after insertion. Therefore, only $D_{\a,\b+1}$ may change after insertion. Since $\alpha$-rank values do not decrease, it follows that only the nodes with $r_\a = \b$ may have their $\alpha$-rank increased to $\b+1$, which proves the case.

    \textit{Case 3:} When $|N(u)| \ge \a$ and $r_\a(u) \le r_\a(v)$, we have $r^N \ge r_\a(u)$ and $\b = r^N$. Let $\vec{E}$ be an egalitarian orientation of the given graph $G$ with respect to $\alpha$. We insert the edge $(u,v)$ into $\vec{E}$ with direction toward $u$, so $u$ now has $(\a+1)$ in-neighbors. The value $r^N$ represents the lowest $\alpha$-rank among these $(\a+1)$ neighbors; let $v_N$ be the corresponding neighbor. We then reverse the edge $(v_N, u)$ in $\vec{E}$ to become $(u, v_N)$. For any $\b' \in [-1, \b] \cup [\b+2, +\infty)$, the updated $\vec{E}$ still satisfies the condition in \defref{ab-dense-subgraph} with alpha value as $\a$ and beta value as $\b'$. Therefore, only the nodes in $V$ with $r_\a = \b$ may have their $\alpha$-rank increased to $\b+1$, which completes the proof.
\end{proof}
}

\begin{Theorem} \label{deletion-update-theorem}
    \mdbf{(Deletion Update Theorem)}\comment{ Given a graph $G$ and an alpha value $\a$, for the deletion of an edge $(u,v)$, let $\b = \min\{r_\a(u), r_\a(v)\}$. The update of $r_{\alpha}$ follows these rules:  
    \begin{enumerate}
        \item All nodes with $r_{\alpha} = \beta$ may have their $r_{\alpha}$ updated to $\beta - 1$.  
        \item For $u$, after the deletion, we have:  
        \begin{itemize}
            \item If $d_u' \leq \a$, then $r_{\alpha}'(u) = -1$.  
            \item Otherwise, $r_{\alpha}'(u)$ is updated to the $r_\a'$ value of the $(\a+1)$-th highest-ranked neighbor in $N'(u)$.  
        \end{itemize}
        \item The $r_{\alpha}$ values of all other nodes remain unchanged.
    \end{enumerate}}
    Given a graph $G$ and an integer $\a$, for an deletion edge $(u,v)$, let $\b = \min\{r_\a(u), r_\a(v)\}$. Then, for all nodes $x \in V$:  
    \begin{enumerate}
        \item If $r_\a(x) = \b$, then $r_\a(x)$ may be updated to $\b - 1$.  
        \item If $r_\a(x) \neq \b$, then $r_\a(x)$ remains unchanged.  
    \end{enumerate}
\end{Theorem}
\fullversion{
\begin{proof}
    Let $\vec{E}$ be an egalitarian orientation of the given graph $G$ and alpha $\a$. We divide the proof into three cases.

    \textit{Case 1:} When $|N(u)| \le \a$, it follows from \theref{degree} that $r_\a(u) = -1$. According to the definition of egalitarian orientation, all neighbors of $u$ are directed toward $u$ in $\vec{E}$, so the directed edge $(v,u)$ can be directly removed. It is clear that the orientation remains egalitarian, and all nodes’ $r_\a$ values remain unchanged, which completes the proof.

    \textit{Case 2:} When $|N(u)| > \a$ and the edge $(u,v)$ is directed toward $v$ in $\vec{E}$. According to \lemref{egalitarian-property} (we will prove later), we have $r_\a(u) \ge r_\a(v)$, and thus $\b = r_\a(v)$. We directly remove the edge $(u,v)$ from $\vec{E}$. For any beta value $\b' \in [-1, \b-1] \cup [\b+1, +\infty)$, it is easy to verify that $\vec{E}$ still satisfies the condition in \defref{ab-dense-subgraph} with alpha value $\a$ and beta value $\b'$. This implies that only the subgraph $D_{\a,\b}$ is affected by the edge deletion. Since deleting an edge cannot increase any node’s $\alpha$-rank, it follows that only the nodes with $r_\a = \b$ may have their rank reduced to $\b - 1$, which proves the case.

    \textit{Case 3:} When $|N(u)| > \a$ and the edge $(u,v)$ is directed toward $u$ in $\vec{E}$. According to \lemref{egalitarian-property}, we have $r_\a(u) \le r_\a(v)$, and thus $\b = r_\a(u)$. From the proof of \theref{u-update-theorem}, there exists a neighbor $v_2 \in N(u)$ such that $(u, v_2)$ is in $\vec{E}$ and $r_\a(v_2) = r_\a(u)$. We first remove the edge $(v, u)$ from $\vec{E}$, and then reverse the edge $(u, v_2)$ to become $(v_2, u)$. For any beta value $\b' \in [-1, \b-1] \cup [\b+1, +\infty)$, it is easy to verify that the updated $\vec{E}$ still satisfies the condition in \defref{ab-dense-subgraph} with alpha value $\a$ and beta value $\b'$. This shows that only the nodes in $V$ with $r_\a = \b$ may have their rank decreased to $\b - 1$, which completes the proof.
\end{proof}
}

The above two update theorems indicate that when an edge is inserted or deleted, there exists a value $\b$ such that in $V$, only nodes with $r_\a = \b$ require updates, and their $r_\a$ values change by at most $1$ (note that this property does not hold for nodes in $U$). In the next theorem, we describe how to compute the $r_\a$ values of nodes in $U$ given the $r_\a$ values of nodes in $V$. Based on this theorem, we can first update the $r_\a$ values of nodes in $V$ and then use these updated values to update the $r_\a$ values of nodes in $U$, forming the foundation of our maintenance algorithms.

\begin{Theorem} \label{u-update-theorem}
    Given a graph $G$ and a value $\a$, for any $u \in U$, we have:  
    \begin{enumerate}
        \item If $|N(u)| \leq \a$, then $r_\a(u) = -1$.  
        \item If $|N(u)| > \a$, then $r_\a(u)$ is equal to the $r_\a$ value of the $(\a+1)$-th highest-ranked node in $N(u)$.  
    \end{enumerate}
\end{Theorem}
\fullversion{
\begin{proof}
    When $|N(u)| \le \a \Rightarrow d_u \le \a$, it follows from \lemref{degree} that $u \notin D_{\a,0}$, and thus $r_\a(u) = -1$. Otherwise, if $|N(u)| > \a$, let $\vec{E}$ be an egalitarian orientation of the given graph $G$ with respect to $\alpha$. According to the definition of egalitarian orientation, the indegree of $u$ is exactly $\a$, so there are $\a$ neighbors in $N(u)$ with directed edges pointing to $u$. By \lemref{egalitarian-property}, these neighbors must have $\alpha$-rank values greater than or equal to $r_\a(u)$.  

    In addition, since $r_\a(u)$ is the $\alpha$-rank of $u$, $u$ must be able to reach a node $v_1 \in V$ whose indegree is greater than $r_\a(u)$. Let $(u, v_2)$ be the first edge on the path $u\rightsquigarrow v_1$. It follows that $r_\a(v_2) = r_\a(u)$, and by contradiction, we can show that $u$ cannot reach any node with an $\alpha$-rank greater than $r_\a(u)$.  

    In summary, $N(u)$ contains $\a$ neighbors pointing to $u$, all of whose $\alpha$-ranks are greater than or equal to $r_\a(u)$. On the other hand, among the neighbors pointed to by $u$, the one with the highest $\alpha$-rank is $v_2$, whose rank is exactly $r_\a(u)$. Therefore, the $(\a+1)$-th largest $\alpha$-rank among the neighbors of $u$ is $r_\a(u)$, which completes the proof. 
\end{proof}
}

\fullversion{\begin{Example} \label{update-example}
    For the insertion case, consider the scenario where $\a=0$ and we insert $(u_6, v_3)$ into $G$ in \figref{example-graph-1}. By \theref{insertion-update-theorem}, we first examine $N(u_6) \cup \{v_3\} = \{v_3, v_5\}$. From the \kw{BD-Index} in \figref{bd-index-example}, we have $r_\a(v_3) = 3$ and $r_\a(v_5) = 1$, thus obtaining $r^N = r_\a(v_3) = 3$. This results in $\b = 3$, indicating that nodes in $V$ with $r_\a = 3$, namely $v_2$ and $v_3$, may have their $r_\a$ values updated to $4$, while the $r_\a$ values of all other nodes in $V$ remain unchanged. In fact, after inserting $(u_6, v_3)$, only $r_\a(v_3)$ increases to $4$. Next, we update the $r_\a$ values of nodes in $U$ according to \theref{u-update-theorem}. As a result, the $r_\a$ values of $\{u_1, u_2, u_3, u_4, u_6\}$ are updated to $4$. Thus, we obtain the updated $r_\a$ values for all nodes.  
    
    For the deletion case, suppose that $\a=0$ and we delete $(u_5, v_4)$ from $G$. By \theref{deletion-update-theorem}, we compute $\b = \min\{2, 2\} = 2$. This means that nodes in $V$ with $r_\a = 2$, namely $v_1$ and $v_4$, may have their $r_\a$ values updated to $1$. In fact, after this deletion, only $v_4$ undergoes an update, with $r_\a(v_4)$ decreasing to $1$, while the $r_\a$ values of all other nodes in $V$ remain unchanged. Next, according to \theref{u-update-theorem}, the $r_\a$ value of $u_5$ is updated to $-1$.
\end{Example}}

\comment{\stitle{Discussion.} According to \theref{insertion-update-theorem} and \theref{deletion-update-theorem}, after a single edge insertion or deletion, the $r_\a$ values of nodes in $V$ change by at most $1$. However, as shown in \expref{update-example}, the $r_\a$ values of nodes in $U$ do not necessarily follow the same pattern. For example, $r_\a(u_6)$ changes from $1$ to $4$, and $r_\a(u_5)$ changes from $2$ to $-1$. This significant change of $r_\a$ is essentially different from rank updates in traditional unipartite graphs \cite{dd}, where the rank of any node can only change by at most $1$. As a result, maintaining rank updates in bipartite graphs exhibits inherent asymmetry and complexity, making it a significant challenge.}

\begin{algorithm}[t]
    \caption{\invoke{BD-Insert-S}$(G,\mathbb{I}_{BD},u,v)$} \label{bd-insert-s}
    \small
    \KwIn{A bipartite graph $G$ and its \kw{BD-Index} $\mathbb{I}_{BD}$, the edge $(u,v)$ to be inserted.}
    \KwOut{The updated $G$ and $\mathbb{I}_{BD}$.}
    \For{$\a=0,1,\ldots,p$}{
        Compute $\b$ according to \theref{insertion-update-theorem}\;
        Invoke \kw{DSS++} algorithm in \cite{ddbipartite} to compute the $D_{\a,\b+1}$ of $G\cup\{(u,v)\}$\;
        \ForEach{$x\in V\cap D_{\a,\b+1}$ and $r_\a(x)=\b$}{
            $r_\a(x)\gets\b+1$, update $\mathbb{I}_{BD}^U$ accordingly\;
        }
        \ForEach{$x\in U$}{
            Compute $r_\a(x)$ according to \theref{u-update-theorem}, update $\mathbb{I}_{BD}^U$ accordingly\;
        }
    }
    Update $\mathbb{I}_{BD}^V$ similarly\;
    $G\gets G\cup\{(u,v)\}$\;
    \Return{$G$, $\mathbb{I}_{BD}$}\;
\end{algorithm}

\subsection{The incremental maintenance algorithms}
Building on \theref{insertion-update-theorem}, \theref{deletion-update-theorem}, and \theref{u-update-theorem}, we present the incremental maintenance algorithms, \textsf{BD-Insert-S} and \textsf{BD-Insert-D}, to handle edge insertion and deletion, as depicted in \algref{bd-insert-s} and \algref{bd-delete-s}.

\stitle{The \textsf{BD-Insert-S} algorithm for edge insertion.} In \algref{bd-insert-s}, \kw{BD-Insert-S} processes each $\a$ iteration independently (line 1). First, it determines the value $\b$ according to \theref{insertion-update-theorem} (line 2). Since nodes with $r_\a = \b$ may increment their $r_\a$ to $\b+1$, the algorithm computes $D_{\a,\b+1}$ to identify these affected nodes (line 3). It then iterates over all $x \in V$: if $x \in D_{\a,\b+1}$ and $r_\a(x) = \b$, it updates $r_\a(x)$ to $\b+1$ and maintains $\mathbb{I}_{BD}^U$ accordingly (lines 4-5). Once all nodes in $V$ are processed, the algorithm applies \theref{u-update-theorem} to update the nodes in $U$ (lines 6-7). After processing all $\a$ values, the algorithm completes updating $\mathbb{I}_{BD}^U$ and symmetrically maintains $\mathbb{I}_{BD}^V$ (line 8). Finally, the updated $G$ and \kw{BD-Index} are returned (line 10).

\stitle{The \textsf{BD-Delete-S} algorithm for edge deletion.} The workflow of \textsf{BD-Delete-S} mirrors \textsf{BD-Insert-S}. The algorithm first iterates over each $\a$ (line~1) and derives $\b$ via \theref{deletion-update-theorem} (line 2). To identify nodes with $r_\a = \b$ that require updating to $\b - 1$, it computes $D_{\a,\b}$ after the edge deletion (line 3). For each node $x \in V$, if $r_\a(x) = \b$ and $x \notin D_{\a,\b}$, the algorithm reduces $r_\a(x)$ to $\b - 1$ and updates $\mathbb{I}_{BD}^U$ (lines 4-5). Subsequently, it updates the $r_\a$ values of nodes in $U$ following \theref{u-update-theorem} (lines 6-7). After processing all $\a$, the algorithm finalizes updates to $\mathbb{I}_{BD}^U$ and symmetrically maintains $\mathbb{I}_{BD}^V$ (line~8). Finally, the updated $G$ and $\mathbb{I}_{BD}$ are returned (line 10).

The correctness of \kw{BD-Insert-S} and \kw{BD-Delete-S} is directly established by \theref{insertion-update-theorem} and \theref{deletion-update-theorem}. Next, we analyze their time and space complexities.

\begin{Theorem}
    The \kw{BD-Insert-S} and \kw{BD-Delete-S} algorithms have a time complexity of $O(p \cdot |E|^{1.5})$ and a space complexity of $O(|E|)$.
\end{Theorem}
\fullversion{\begin{proof}
    In the \textsf{BD-Insert-S} algorithm, the dominant computational cost lies in invoking \textsf{DSS++}, which has a worst-case time complexity of $O(|E|^{1.5})$ per call \cite{ddbipartite}. As \textsf{BD-Insert-S} executes \textsf{DSS++} $O(p)$ times, the algorithm derives its overall time complexity of $O(p \cdot |E|^{1.5})$. Regarding space complexity, the most memory-intensive operation is \kw{DSS++}, which requires $O(|E|)$ space \cite{ddbipartite}. Therefore, the space complexity of \kw{BD-Insert-S} is $O(|E|)$. The \textsf{BD-Delete-S} algorithm operates analogously to \textsf{BD-Insert-S}; applying the same analysis, its time complexity is $O(p \cdot |E|^{1.5})$, and its space complexity is $O(|E|)$.
\end{proof}}

\begin{algorithm}[t]
    \caption{\invoke{BD-Delete-S}$(G,\mathbb{I}_{BD},u,v)$} \label{bd-delete-s}
    \small
    \KwIn{A bipartite graph $G$ and its \kw{BD-Index} $\mathbb{I}_{BD}$, the edge $(u,v)$ to be deleted.}
    \KwOut{The updated $G$ and $\mathbb{I}_{BD}$.}
    \For{$\a=0,1,\ldots,p$}{
        Compute $\b$ according to \theref{deletion-update-theorem}\;
        Invoke \kw{DSS++} algorithm in \cite{ddbipartite} to compute the $D_{\a,\b}$ of $G\setminus\{(u,v)\}$\;
        \ForEach{$x\in V\setminus D_{\a,\b}$ and $r_\a(x)=\b$}{
            $r_\a(x)\gets\b-1$, update $\mathbb{I}_{BD}^U$ accordingly\;
        }
        \ForEach{$x\in U$}{
            Compute $r_\a(x)$ according to \theref{u-update-theorem}, update $\mathbb{I}_{BD}^U$ accordingly\;
        }
    }
    Update $\mathbb{I}_{BD}^V$ similarly\;
    $G\gets G\setminus\{(u,v)\}$\;
    \Return{$G$, $\mathbb{I}_{BD}$}\;
\end{algorithm}

Compared to the baseline method of recomputing from scratch, \kw{BD-Insert-S} and \kw{BD-Delete-S} retain the space-efficient advantage by retaining $O(|E|)$ space. Meanwhile, they reduce the time complexity of updating a single edge from $O(p\cdot  |E|^{1.5}\cdot \log|U\cup V|)$ to $O(p\cdot |E|^{1.5})$. Unlike the baseline, which recomputes all layers, \kw{BD-Insert-S} and \kw{BD-Delete-S} update only a single node layer in $V$, resulting in superior performance. As shown by our experiments, they achieve approximately one order of magnitude of speedup compared to the baseline approach.

\section{TIME-EFFICIENT INDEX MAINTENANCE}
In this section, we focus on developing time-efficient algorithms for maintaining \kw{BD-Index}. Recall that if an orientation satisfies the structural constraints in \defref{ab-dense-subgraph}, the corresponding dense subgraph can be derived dirctly. This insight motivates our idea: by maintaining the orientation, we can efficiently compute the dense subgraph (i.e., node ranks), thereby enabling efficient \kw{BD-Index} maintenance. Building on this, we first define the concept of egalitarian orientation (\defref{egalitarian-orientation}), and then show how to efficiently compute the \kw{BD-Index} from egalitarian orientation (\algref{orientationtorank}), thereby transforming the task of maintaining the \kw{BD-Index} into maintaining the egalitarian orientation. We then present the algorithms \kw{BD-Insert-T} and \kw{BD-Delete-T}, which are designed to maintain the egalitarian orientations efficiently. 

\subsection{A novel concept: egalitarian orientation}
In this subsection, we introduce the concept of egalitarian orientation based on $\a$ for maintaining $\mathbb{I}_{BD}^U$. The maintenance of $\mathbb{I}_{BD}^V$ follows symmetrically by replacing $\alpha$ with $\beta$ and swapping the constraints between $U$ and $V$ in Definition \ref{egalitarian-orientation}.

\begin{Definition} \label{egalitarian-orientation}
    \mdbf{(Egalitarian orientation)} Given a bipartite graph $G$ and an alpha value $\a$, an orientation $\vec{E}$ is called an \emph{egalitarian orientation} if it satisfies the following conditions:  
    \begin{enumerate}
        \item For any node $u \in U$, if $d_u(G) > \a$, then $\vec{d}_u(\vec{E}) = \a$; otherwise, if $d_u(G) \le \a$, then $\vec{d}_u(\vec{E}) = d_u(G)$.  
        \item There exists no path $v_s \rightsquigarrow v_t$ in $\vec{E}$ such that $v_s, v_t \in V$ and $\vec{d}_{v_t}(\vec{E}) - \vec{d}_{v_s}(\vec{E}) \ge 2$.
    \end{enumerate}
\end{Definition}


For example, the orientation shown in \subfigref{example-graph-2} is an egalitarian orientation given $\a = 1$. Each node in $U$ has an indegree of exactly $\a = 1$, and no path $v_s \rightsquigarrow v_t$ exists among nodes in $V$ that violates the condition of Definition~\ref{egalitarian-orientation}. 

\begin{algorithm}[t]
    \caption{\invoke{OrientationToRank}$(\a,\vec{E})$} \label{orientationtorank}
    \small
    \KwIn{An alpha value $\a$ and an egalitarian orientation.}
    \KwOut{The $\a$-rank of all nodes.}
    $\vec{d}_{\max} \gets \max_{v \in V} \vec{d}_v(\vec{E})$, $vis\gets \emptyset$\;
    \ForEach{$k=\vec{d}_{\max}-1,\vec{d}_{\max}-2,\ldots,0$}{
        $T\gets\{v\in V\setminus vis\mid\vec{d}_v(\vec{E})=k+1\}$\;
        \ForAll{$x\in(U\cup V)\setminus vis$, $x\in T$ or $x$ can reach a node in $T$}{
            $r_\a(x)\gets k$, $vis\gets vis\cup \{x\}$\;
        }
    }
    \lForAll{$x\in (U\cup V)\setminus vis$}{$r_\a(x)\gets -1$}
    \Return{$r_\a$}\;
\end{algorithm}


\stitle{Intuition behind egalitarian orientation.} The term ``egalitarian'' in egalitarian orientation refers to the balancing of indegrees among nodes in $V$. Suppose there exists a path $v_s \rightsquigarrow v_t$ in an orientation $\vec{E}$ such that $v_s, v_t \in V$ and $\vec{d}_{v_t}(\vec{E}) - \vec{d}_{v_s}(\vec{E}) \ge 2$. By reversing this path (i.e., reversing the direction of all edges along the path), the indegree of $v_t$ decreases by one, the indegree of $v_s$ increases by one, and the indegrees of all other nodes remain unchanged. This reversing operation balances the indegrees of $v_t$ and $v_s$, leading to a more even distribution of indegrees among nodes in $V$. An egalitarian orientation guarantees that no such path exists, intuitively meaning the indegrees of nodes in $V$ are already as ``egalitarian'' as possible. 

Next, we introduce how to compute the rank of nodes based on the egalitarian orientation. The algorithm \kw{OrientationToRank} for computing $\alpha$-rank from an egalitarian orientation is shown in \algref{orientationtorank}. The algorithm first computes the maximum indegree $\vec{d}_{\max}$ among all nodes in $V$, and initializes the set $vis$ to record visited nodes (line 1). In each round of the ``foreach'' loop (line 2), the algorithm identifies the nodes whose $r_\alpha$ is equal to $k$. Specifically, it first collects nodes in $V \setminus vis$ with indegree equal to $k+1$ into a set $T$ (line 3). Then, for each node in $T$ and those that can reach nodes in $T$, their $r_\alpha$ values are set to $k$, and they are added to the visited set $vis$ (lines 4--5). Finally, all unvisited nodes are assigned $r_\alpha = -1$ (line 6), and the algorithm returns the $r_\alpha$ values for all nodes (line 7). Below is an example of running algorithm \kw{OrientationToRank}.

\begin{figure}[t]
    \centering
    \includegraphics[width=0.95\linewidth]{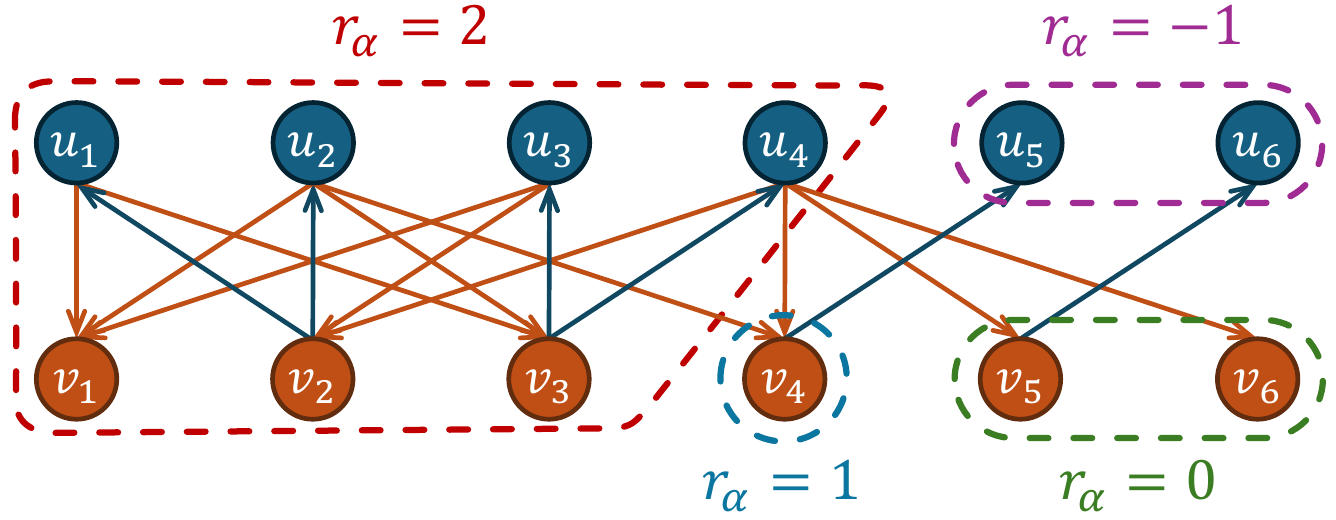}
    \vspace*{-0.3cm}
    \caption{Example of computing $r_\a$ from an egalitarian orientation.} \label{orientation-to-rank-example}
\end{figure}

\fullversion{\begin{Example}
    We set $\alpha = 1$ and use the egalitarian orientation $\vec{E}$ in \subfigref{example-graph-2} as the input to run the \kw{OrientationToRank} algorithm. The procedure is illustrated in \figref{orientation-to-rank-example}. The algorithm first determines that $\vec{d}_{\max} = \vec{d}_{v_1}(\vec{E}) = 3$. Then, it enters the loop with $k = 2$ and identifies all nodes in $V \setminus vis$ with indegree 3, i.e., $T = \{v_1\}$. Next, the algorithm computes the set of nodes that can reach $v_1$, which is $\{v_2, v_3, u_1, u_2, u_3, u_4\}$. As a result, all nodes in $\{v_1, v_2, v_3, u_1, u_2, u_3, u_4\}$ are assigned $r_\alpha = 2$ and added to the set $vis$. In the next iteration with $k = 1$, the algorithm considers the nodes outside of $vis$ and computes $T = \{v_4\}$. Since no nodes can reach $v_4$, the only node with $r_\alpha = 1$ is $v_4$. In the following iteration with $k = 0$, we have $T = \{v_5, v_6\}$. Again, there are no nodes that can reach $v_5$ or $v_6$, so these two nodes are the only ones with $r_\alpha = 0$. Finally, the nodes that have not been visited, i.e., $\{u_5, u_6\}$, are assigned $r_\alpha = -1$.
\end{Example}}

Next, we prove the correctness and analyze the time complexity of the \kw{OrientationToRank} algorithm.

\begin{Theorem} \label{orientation-to-rank}
    The \kw{OrientationToRank} algorithm can correctly return $\a$-rank within $O(|E|)$ time.
\end{Theorem}
\fullversion{\begin{proof}
    We first prove the correctness of the algorithm. We begin by showing that in the egalitarian orientation $\vec{E}$, the path $s \rightsquigarrow t$ described in \defref{ab-dense-subgraph} does not exist. Assume for contradiction that such a path $s \rightsquigarrow t$ exists. First, consider any node $u \in U$ with $d_u(G) \le \a$. In the egalitarian orientation, all its incident edges are directed toward $u$, and its indegree cannot exceed $\a$. Hence, $u$ cannot be either $s$ or $t$ in the path. Next, for any node $u \in U$ with $d_u(G) > \a$, its indegree is exactly $\a$, so again $u$ cannot be $s$ or $t$. Therefore, both $s$ and $t$ must be in $V$. Since $\vec{d}_s < \b$ and $\vec{d}_t > \b$, it follows that $\vec{d}_t - \vec{d}_s \ge 2$, which contradicts the definition of an egalitarian orientation. Thus, such a path $s \rightsquigarrow t$ cannot exist. 

    It follows that in an egalitarian orientation, $D_{\a,\b} = T \cup \{x \mid x \text{ can reach a node in } T \text{ in } \vec{E}\}$, where $T = \{v \in V \mid \vec{d}_v(\vec{E}) > \b\}$. Next, we prove that if a node $x \in D_{\a,\b} \setminus D_{\a,\b+1}$, then the algorithm \kw{OrientationToRank} correctly computes $r_\a(x) = \b$. Let $T_1 = \{v \in V \mid \vec{d}_v(\vec{E}) > \b+1\}$ and $T_2 = \{v \in V \mid \vec{d}_v(\vec{E}) > \b\}$. Since $x \in D_{\a,\b} \setminus D_{\a,\b+1}$, we know that $x \notin T_1$ and cannot reach any node in $T_1$. Therefore, in the ``foreach'' loop of the algorithm, $x$ will not be visited in any iteration where $k \ge \b + 1$. On the other hand, $x \in T_2$ or $x$ can reach a node in $T_2$, so in the iteration where $k = \b$, $x$ will be visited and assigned $r_\a(x) = \b$. According to the definition of $\alpha$-rank, this confirms that \kw{OrientationToRank} computes the correct $\alpha$-rank.

    Next, we analyze the time complexity. In each round of the ``foreach'' loop, computing the nodes that can reach $T$ can be done via a breadth-first search. Once a node is visited, it is added to the set $vis$, and no node is revisited. Therefore, the total cost of all breadth-first searches is bounded by $O(|E|)$, which shows that the overall time complexity of \kw{OrientationToRank} is $O(|E|)$.
\end{proof}}

By using the \kw{OrientationToRank} algorithm, once we have an egalitarian orientation, the $r_\a$ values of all nodes can be computed efficiently in linear time. Therefore, we transform the task of maintaining $r_\a$ and the \kw{BD-Index} into the task of maintaining egalitarian orientation. When an edge is inserted or deleted, we can first update the egalitarian orientation, from which we can then compute the updated $r_\a$ values and correspondingly update the \kw{BD-Index}. Compared to directly maintaining the \kw{BD-Index} (as done in \kw{BD-Insert-S} and \kw{BD-Delete-S}), maintaining the egalitarian orientation is significantly more efficient. Next, we present the algorithms \kw{BD-Insert-T} and \kw{BD-Delete-T}, which provide efficient methods for maintaining egalitarian orientations.

\subsection{The insertion algorithm \textsf{BD-Insert-T}}

According to \defref{egalitarian-orientation} and \theref{orientation-to-rank}, each egalitarian orientation corresponds to a unique $\alpha$ value. The index $\mathbb{I}_{BD}^U$ maintains $(p+1)$ distinct $\alpha$ values, each associated with a node list, thus requiring $(p+1)$ corresponding egalitarian orientations. Similarly, $\mathbb{I}_{BD}^V$ requires another $(p+1)$ orientations. Consequently, the total number of required egalitarian orientations is $(2p + 2)$. We represent the set of all such orientations as $\vec{\mathbb{E}}$.

\begin{algorithm}[t]
    \caption{\invoke{BD-Insert-T}$(G,\mathbb{I}_{BD},\vec{\mathbb{E}},u,v)$} \label{bd-insert-t}
    \small
    \KwIn{A bipartite graph $G$ and its \kw{BD-Index} $\mathbb{I}_{BD}$, the egalitarian orientations, the edge $(u,v)$ to be inserted.}
    \KwOut{The updated $G$, $\mathbb{I}_{BD}$, and egalitarian orientations.}
    \For{$\a=0,1,\ldots,p$}{
        Let $\vec{E}\in \vec{\mathbb{E}}$ be the egalitarian orientation for the current $\a$\;
        $\vec{E}\gets{\vec{E}} \cup (v,u)$\;
        \lIf{$d_u(G)<\a$}{
            continue to the next $\a$ value
        }
        Let $v_{\min}\gets\mathop{\arg\min}\limits_{v'\in V,~v'~\text{can reach}~u}\vec{d}_{v'}(\vec{E})$\;
        Reverse the path $v_{\min}\rightsquigarrow u$\;
        $r_\a(\cdot)\gets$ \kw{OrientationToRank}$(\a,\vec{E})$, and update $\mathbb{I}_{BD}^U$ accordingly\;
    }
    Update $\mathbb{I}_{BD}^V$ similarly\;
    $G\gets G\cup\{(u,v)\}$\;
    \Return{$G$, $\mathbb{I}_{BD}$, $\vec{\mathbb{E}}$}\;
\end{algorithm}

Given all egalitarian orientations $\vec{\mathbb{E}}$, we propose the \kw{BD-Insert-T} algorithm for edge insertion in \algref{bd-insert-t}. Similar to \kw{BD-Insert-S} (\algref{bd-insert-s}), \kw{BD-Insert-T} processes each $\a$ in $\mathbb{I}_{BD}^U$ individually (line 1). For each $\a$, it first retrieves the corresponding egalitarian orientation $\vec{E}$, and inserts the directed edge $(v, u)$ into $\vec{E}$ (lines 2-3). If $d_u(G)<\a$, the orientation $\vec{E}$ remains egalitarian with unchanged $r_\a$ values, and \kw{BD-Insert-T} directly proceeds to the next $\a$ (line 4). Otherwise, if $d_u(G)\ge\a$, $\vec{E}$ may violate the egalitarian conditions, requiring the following adjustments. First, let $v_{\min}$ be the node in $V$ with the minimum indegree among all nodes that can reach $u$ (line 5). By definition of reachability, there is a path $v_{\min}\rightsquigarrow u$, and the algorithm reverses this path in $\vec{E}$ by reversing all edge directions along this path (line 6). This reversal restores $\vec{E}$ to an egalitarian orientation. Thus, the algorithm updates the $r_\a$ values of all nodes using algorithm \kw{OrientationToRank} and correspondingly maintains $\mathbb{I}_{BD}^U$ (line 7). After processing all values of $\a$ (for $\mathbb{I}_{BD}^U$ ) and $\b$ (for $\mathbb{I}_{BD}^V$) using the above method, \kw{BD-Insert-T} returns the updated graph $G$, \kw{BD-Index}, and egalitarian orientations $\vec{\mathbb{E}}$ (line 10).

\fullversion{\begin{Example}
    Consider the orientation shown in \subfigref{example-graph-2}, which by definition forms an egalitarian orientation for $\a=1$. We analyze the execution of the \kw{BD-Insert-T} algorithm when inserting edge $(u_6,v_3)$. After processing the $\a=0$ case, the algorithm proceeds to $\a=1$. First, it inserts the directed edge $(v_3,u_6)$ into the orientation. Since $d_{u_6}(G)=1=\a$, \kw{BD-Insert-T} cannot directly continue to the next $\a$. It identifies all nodes in $V$ that can reach $u_6$, i.e., $v_2,v_3,v_5$. Among these nodes, $v_5$ has the lowest indegree, so we have $v_{\min}=v_5$. Then, the algorithm reverses the path $v_5\rightsquigarrow u_6$, which contains a single directed edge $(v_5,u_6)$, changing its direction to $(u_6,v_5)$. After this reversal, the orientation once again becomes egalitarian. Using algorithm \kw{OrientationToRank}, \kw{BD-Insert-T} then computes the updated $r_\a$ values for all nodes and determines that only nodes $u_6$ and $v_5$ have changed their $r_\a$ values, from $0$ to $1$. The index $\mathbb{I}_{BD}^U$ is maintained accordingly, completing the $\a=1$ case processing for this edge insertion.
\end{Example}}

\submitversion{Next, we prove the correctness and analyze the complexity of the \kw{BD-Insert-T} algorithm. }

\fullversion{Next, we propose the following lemma, which serves as the foundation for analyzing the correctness and complexity of the \kw{BD-Insert-T} algorithm.

\begin{Lemma} \label{egalitarian-property}
    Given a graph $G$, an alpha value $\a$, and an egalitarian orientation $\vec{E}$, we have the following properties: (1) For any node $v \in V$, it holds that $\vec{d}_v(\vec{E}) \in \{r_\a(v), r_\a(v) + 1\}$; (2) For any nodes $x, y \in (U \cup V)$, if $r_\a(x) > r_\a(y)$, then the directed edge $(x, y)$ in $\vec{E}$ must point to $y$.
\end{Lemma}
\begin{proof}
    We first prove property (1). According to the definition of egalitarian orientation and \theref{orientation-to-rank}, we know that the indegree of $v$ $\vec{d}_v(\vec{E}) \le r_\a(v) + 1$; otherwise, $v$ would have a higher $\alpha$-rank value. Moreover, we have that either $v$ has an indegree greater than $r_\a(v)$, or $v$ can reach a node in $V$ with indegree greater than $r_\a(v)$. If $v$ has indegree greater than $r_\a(v)$, the property holds. Otherwise, by condition (2) in the definition of egalitarian orientation, $v$ must have indegree greater than $r_\a(v) - 1$, and thus the property also holds.

    Next, we prove property (2). According to the construction of $D_{\a,\b}$ in the proof of \theref{orientation-to-rank}, let $T = \{v \in V \mid \vec{d}_v(\vec{E}) > \b\}$. Then, $D_{\a,\b}$ consists of $T$ together with all nodes that can reach $T$, which implies that all edges between $D_{\a,\b}$ and $(U \cup V) \setminus D_{\a,\b}$ are directed from $D_{\a,\b}$ to $(U \cup V) \setminus D_{\a,\b}$. Therefore, if $r_\a(x) > r_\a(y)$, node $x$ belongs to a denser subgraph than $y$, and the edge $(x, y)$ must be directed toward $y$, which proves the property.
\end{proof}
}

\begin{Theorem} \label{bd-insert-t-correctness}
    The \kw{BD-Insert-T} algorithm can correctly maintain \kw{BD-Index}.
\end{Theorem}
\fullversion{
\begin{proof}
    Since the algorithm processes each node list (i.e., each alpha value) separately, we only need to prove that within the loop for a given alpha value (lines 2–7), the algorithm correctly updates the egalitarian orientation. Once the orientation is correctly maintained, the corresponding $r_\a$ values and the \kw{BD-Index} can be correctly updated based on \theref{orientation-to-rank}.  

    Therefore, let the current alpha value be $\a$. If $d_u(G) < \a$ (line 4), it is clear that the egalitarian orientation is updated correctly. Otherwise, if $d_u(G) > \a$, let $(v_1, u)$ be the last edge on the path $v_{\min} \rightsquigarrow u$. Since $v_{\min}$ can reach $v_1$ or $v_{\min} = v_1$, it follows from \theref{egalitarian-property} that $r_\a(v_{\min}) \ge r_\a(v_1)$. Additionally, since $\vec{d}_{v_{\min}} \le \vec{d}_{v_1}$, we also have $r_\a(v_{\min}) \le r_\a(v_1)$. Therefore, $r_\a(v_{\min}) = r_\a(v_1)$. Similarly, we can conclude that all nodes along the path $v_{\min} \rightsquigarrow v_1$ have the same $\alpha$-rank as $v_{\min}$. Now, suppose we first reverse the path $v_{\min} \rightsquigarrow v_1$ in the original orientation. It is easy to see that the resulting orientation remains egalitarian. Then, we reverse the edge $(v_1, u)$ (i.e., the full path $v_{\min} \rightsquigarrow u$ has been reversed). Since $v_{\min}$ is the node with the smallest indegree among those that can reach $u$, the resulting orientation is still egalitarian. Hence, \kw{BD-Insert-T} correctly maintains the egalitarian orientation, and by \theref{orientation-to-rank}, it also correctly updates the \kw{BD-Index}.
\end{proof}
}

\begin{Theorem}
    The time complexity and space complexity of the \kw{BD-Insert-T} algorithm are both $O(p\cdot|E|)$.
\end{Theorem}
\fullversion{\begin{proof}
    The most computationally intensive steps of the algorithm are line 5 and line 7. Line 5 can be implemented by performing a breadth-first search from node $u$, requiring $O(|E|)$ time. Similarly, line 7 can be executed in $O(|E|)$ time according to \theref{orientation-to-rank}. Consequently, processing a single $\a$ value requires $O(|E|)$ time. Since the algorithm handles $(p+1)$ distinct values of $\a$, maintaining $\mathbb{I}_{BD}^U$ incurs $O(p \cdot|E|)$ time complexity. By symmetry, the maintenance of $\mathbb{I}_{BD}^V$ has identical complexity. Therefore, the total time complexity of \kw{BD-Insert-T} is $O(p\cdot|E|)$. For space complexity, the dominant cost comes from storing the input egalitarian orientations $\vec{\mathbb{E}}$, occupying $O(p \cdot |E|)$ space.
\end{proof}}

\subsection{The deletion algorithm \textsf{BD-Delete-T}}

\begin{algorithm}[t]
    \caption{\invoke{BD-Delete-T}$(G,\mathbb{I}_{BD},\vec{\mathbb{E}},u,v)$} \label{bd-delete-t}
    \small
    \KwIn{A bipartite graph $G$ and its \kw{BD-Index} $\mathbb{I}_{BD}$, the egalitarian orientations, the edge $(u,v)$ to be deleted.}
    \KwOut{The updated $G$, $\mathbb{I}_{BD}$, and egalitarian orientations.}
    \For{$\a=0,1,\ldots,p$}{
        Let $\vec{E}\in \vec{\mathbb{E}}$ be the egalitarian orientation for the current $\a$\;
        \If{$d_u(G)\le\a$}{
            $\vec{E}\gets\vec{E}\setminus(v,u)$\;
            Continue to the next $\a$ value\;
        }
        \eIf{$(v,u)\in\vec{E}$}{
            Let $v_{\max}\gets\mathop{\arg\max}\limits_{v'\in V,~u~\text{can reach}~v'}\vec{d}_{v'}(\vec{E})$\;
            Reverse the path $u\rightsquigarrow v_{\max}$\;
        }
        (\tcp*[h]{$(u,v)\in\vec{E}$})
        {
            Let $v_{\max}\gets\mathop{\arg\max}\limits_{v'\in \{v'\in V|v~\text{can reach}~v'\}\cup\{v\}}\vec{d}_{v'}(\vec{E})$\;
            \lIf{$v_{\max}\neq v$}{reverse the path $v\rightsquigarrow v_{\max}$}
        }
        $\vec{E}\gets\vec{E}\setminus(u,v)$ or $\vec{E}\gets\vec{E}\setminus(v,u)$\;
        $r_\a(\cdot)\gets$ \kw{OrientationToRank}$(\a,\vec{E})$, and update $\mathbb{I}_{BD}^U$ accordingly\;
    }
    Update $\mathbb{I}_{BD}^V$ similarly\;
    $G\gets G\setminus\{(u,v)\}$\;
    \Return{$G$, $\mathbb{I}_{BD}$, $\vec{\mathbb{E}}$}\;
\end{algorithm}

The pseudo-code of our \kw{BD-Delete-T} algorithm is outlined in \algref{bd-delete-t}. Similar to \kw{BD-Insert-T}, the \kw{BD-Delete-T} algorithm processes each $\a$ value individually (line 1). For each $\a$, it first retrieves the corresponding egalitarian orientation $\vec{E}$ (line 2). The algorithm then diverges into two cases based on whether $d_u(G)\le \a$. In the first case where $d_u(G)\le \a$, the algorithm simply removes the edge $(v,u)$ from $\vec{E}$ (by definition, edge $(v,u)$ must be directed toward $u$ in this case) while preserving the orientation's egalitarian property and all $r_\a$ values. In the second case where $d_u(G)>\a$, the algorithm examines the direction of edge $(u,v)$ in $\vec{E}$. If oriented toward $u$ (line 6), it identifies $v_{\max}$ as the maximum-indegree node reachable from $u$, and reverses the path $u\rightsquigarrow v_{\max}$ (lines 7-8). Conversely, if oriented toward $v$ (line 9), the algorithm determines $v_{\max}$ as either $v$ itself or the node with the highest indegree among all nodes in $V$ reachable from $v$ (line 10), and performs path reversal only in the latter case (line 11). Next, the algorithm removes either $(u,v)$ or $(v,u)$ from $\vec{E}$, depending on whether it points toward $u$ or $v$ (line 12). At this point, the resulting $\vec{E}$ is guaranteed to be an egalitarian orientation. Subsequently, the algorithm updates the $r_\a$ values of all nodes and updates the \kw{BD-Index} using algorithm \kw{OrientationToRank} (line 13). After processing all $\a$ values for $\mathbb{I}_{BD}^U$, the algorithm performs analogous updates for $\mathbb{I}_{BD}^V$ (line 14). Ultimately, the \kw{BD-Delete-T} algorithm returns $G$, $\mathbb{I}_{BD}$, and $\vec{\mathbb{E}}$ (line 16).

Next, we prove the correctness of \kw{BD-Delete-T} and analyze its complexity.

\begin{Theorem} \label{bd-delete-t-correctness}
    The \kw{BD-Delete-T} algorithm can correctly maintain \kw{BD-Index}.
\end{Theorem}
\fullversion{
\begin{proof}
    Similar to the correctness proof of \kw{BD-Insert-T}, the correctness of \kw{BD-Delete-T} also only requires showing that the algorithm correctly maintains the egalitarian orientation for each alpha value $\a$. First, if $d_u(G) \le \a$, it is straightforward to verify the correctness based on the definition of egalitarian orientation. Otherwise, if $d_u(G) > \a$, we consider two cases:

    \textit{Case 1:} $(u,v) \in \vec{E}$ (line 6). In this case, according to the proof of \theref{u-update-theorem}, the neighbor $v_1 \in N(u)$ that is pointed to by $u$ and has the highest $r_\a$ satisfies $r_\a(v_1) = r_\a(u)$. Since $v_1$ can reach $v_{\max}$ or is equal to $v_{\max}$, we have $r_\a(v_1) \ge r_\a(v_{\max})$. On the other hand, since $v_{\max}$ has a higher indegree than $v_1$, it follows that $r_\a(v_{\max}) \ge r_\a(v_1)$. Therefore, $r_\a(v_1) = r_\a(v_{\max})$. Similarly, we can conclude that all nodes along the path $v_1 \rightsquigarrow v_{\max}$ have the same $\alpha$-rank as $v_{\max}$. Hence, by reversing the path $u \rightsquigarrow v_{\max}$ and deleting the edge $(u,v)$ or $(v,u)$ from $\vec{E}$, it is easy to see that the resulting orientation remains egalitarian.

    \textit{Case 2:} $(v,u) \in \vec{E}$ (line 6). This case can be further divided into two subcases: (1) If $v_{\max} = v$, then the algorithm simply removes the edge $(u,v)$ without reversing any path. This implies that $v$ cannot reach any node in $V$ with a higher indegree. According to the definition of egalitarian orientation, any node in $V$ that can reach $v$ must have indegree at least $\vec{d}_v - 1$, so decreasing $v$'s indegree by removing $(u,v)$ does not violate the egalitarian condition. (2) If $v_{\max} \neq v$, we can use a similar argument as in Case 1 to show that the resulting orientation remains egalitarian after the path reversal and edge deletion.

    Therefore, the algorithm correctly maintains the egalitarian orientation, and by \theref{orientation-to-rank}, it also correctly maintains the \kw{BD-Index}.
\end{proof}
}

\begin{Theorem}
    The time complexity and space complexity of the \kw{BD-Delete-T} algorithm are both $O(p\cdot|E|)$.
\end{Theorem}
\fullversion{\begin{proof}
    Similar to \kw{BD-Insert-T}, \kw{BD-Delete-T} can identify $v_{\max}$ through a single breadth-first search in $O(|E|)$ time. The rank update in line 13 also requires $O(|E|)$ time according to \theref{orientation-to-rank}. Thus, maintaining $\mathbb{I}_{BD}^U$ for a single $\a$ takes $O(|E|)$ time. With $(p+1)$ distinct $\a$ values, the total time complexity for maintaining $\mathbb{I}_{BD}^U$ becomes $O(p \cdot |E|)$. The same complexity applies to maintaining $\mathbb{I}_{BD}^V$. Hence, the overall time complexity of \kw{BD-Delete-T} is $O(p \cdot |E|)$. Regarding space complexity, the dominant factor is storing the egalitarian orientations, which occupies $O(p \cdot |E|)$ space.
\end{proof}}

By leveraging the egalitarian orientation, \kw{BD-Insert-T} and \kw{BD-Delete-T} can update a node list in \kw{BD-Index} in $O(|E|)$ time, avoiding the computationally expensive \kw{DSS++} algorithm ($O(|E|^{1.5})$ time) required by \kw{BD-Insert-S} and \kw{BD-Delete-S}. This complexity reduction enables \kw{BD-Insert-T} and \kw{BD-Delete-T} to significantly outperform their counterparts (\kw{BD-Insert-S} and \kw{BD-Delete-S}) when processing dynamic graphs.

\comment{\stitle{Discussion.} The concept of egalitarian orientation also exists in unipartite graphs \cite{dd}, but it differs fundamentally from the egalitarian orientation we propose for bipartite graphs in the following aspects: (1) Unlike the uniform treatment of nodes in unipartite graphs, we explicitly distinguish between node sets $U$ and $V$. Specifically, we introduce a novel constraint that limits the indegree of nodes in $U$ to at most $\a$, which has no counterpart in unipartite graphs; (2) While egalitarian orientation in unipartite graphs is parameter-free, our definition incorporates parameters $\a$ and $\b$, which is specifically designed to maintain \kw{BD-Index}; (3) Both \kw{BD-Insert-T} (line 4) and \kw{BD-Delete-T} (lines 3–5) require special handling for low-degree nodes $u$ while ensuring nodes in $U$ never exceed the indegree threshold $\a$ during maintenance, representing unique challenges in bipartite graphs. These fundamental differences establish our bipartite egalitarian orientation as a novel structure that cannot be substituted by its unipartite counterpart.}

{\color{\mycolorb}\stitle{Discussion: novelty and applicability of \kw{BD-Index}.} Compared to prior work \cite{ddbipartite}, our contributions lie in proposing a novel index \kw{BD-Index} along with two maintenance strategies that enable efficient query processing and dynamic graph updates. These innovations significantly enhance the practicality of $(\alpha,\beta)$-dense subgraph search in real-world scenarios (as shown in our case study in Section 6.3). In addition to supporting $(\alpha,\beta)$-dense subgraphs, we also observe that the same indexing framework can be extended to $(\alpha,\beta)$-core decomposition; the details of this extension are discussed in the remark in Section~3.4.
}

\newcounter{exp}
\newcommand{\expnumber}{\refstepcounter{exp}\theexp}
\section{EXPERIMENTS}

\stitle{Algorithms.} {\color{\mycolorb}For $(\a,\b)$-dense subgraph queries, we implement three algorithms: the baseline online algorithm \kw{Online}, its optimized variant \kw{Online++}, and the index-based algorithm \kw{Query-BD-Index} (\algref{query-bd-index}). The \kw{Online} algorithm processes each query by invoking the state-of-the-art $(\a,\b)$-dense subgraph search algorithm \kw{DSS++} \cite{ddbipartite}. Building on this, \kw{Online++} accelerates query processing through two optimizations: (1) it reuses the orientation obtained from the previous query rather than reinitializing it before the max-flow computation; and (2) it caches the results of the most recent 10 queries. When processing a new query with parameters $(\alpha,\beta)$, \kw{Online++} checks whether any cached result $D_{\alpha^+,\beta^+}$ satisfies $\alpha^+ \ge \alpha$ and $\beta^+ \ge \beta$, or $D_{\alpha^-,\beta^-}$ satisfies $\alpha^- \le \alpha$ and $\beta^- \le \beta$. By exploiting the hierarchical property of $(\a,\b)$-dense subgraphs, the computation can then be restricted to the subgraph outside $D_{\alpha^+,\beta^+}$ or inside $D_{\alpha^-,\beta^-}$.}

For index construction, we implement the \kw{Build-BD-Index} algorithm (\algref{build-bd-index}). To maintain the \kw{BD-Index} dynamically, we implement the space-efficient algorithms \kw{BD-Insert-S} (\algref{bd-insert-s}) and \kw{BD-Delete-S} (\algref{bd-delete-s}), as well as the time-efficient algorithms \kw{BD-Insert-T} (\algref{bd-insert-t}) and \kw{BD-Delete-T} (\algref{bd-delete-t}). As a baseline comparison, we consider recomputing the \kw{BD-Index} from scratch using \kw{Build-BD-Index} after each update, denoted as \kw{Recomputing}. All algorithms are implemented in C++ with O3 optimization. Our experiments are conducted on a Linux system with a 2.2GHz AMD 3990X 64-Core CPU and 256GB of memory.

\stitle{Datasets.} As shown in \tabref{dataset}, we evaluate the proposed algorithms on 10 real-world datasets: \kw{Actor} (\kw{AC}), \kw{IMDB} (\kw{IM}), \kw{Hepph} (\kw{HE}), \kw{Amazon} (\kw{AM}), \kw{Flickr} (\kw{FL}), \kw{Epinions} (\kw{EP}), \kw{Patent} (\kw{PA}), \kw{Pokec} (\kw{PO}), \kw{Wiki} (\kw{WI}), and \kw{Livejournal} (\kw{LI}). All datasets are publicly available from the Koblenz Network Collection (\url{http://www.konect.cc/}).

\begin{table}[t]
\caption{Statistics of datasets.} \label{dataset}
\vspace*{-0.3cm}
\setlength{\tabcolsep}{8pt}
\small
1K=1,000, 1M=1,000,000\\
\begin{tabular}{c|c|cccc}
\hline
Dataset & Category    & $|U|$  & $|V|$  & $|E|$   & $p$   \\ \hline
\kw{AC} & affiliation & 127.8K & 383.6K & 1.5M    & 12    \\
\kw{IM} & affiliation & 303.6K & 896.3K & 3.8M    & 20    \\
\kw{HE} & citation    & 24.5K  & 28.1K  & 4.6M    & 371   \\
\kw{AM} & rating      & 2.1M   & 1.2M   & 5.7M    & 23    \\
\kw{FL} & affiliation & 396.0K & 103.6K & 8.5M    & 134   \\
\kw{EP} & rating      & 120.5K & 755.8K & 13.7M   & 120   \\
\kw{PA} & citation    & 2.1M   & 3.3M   & 16.5M   & 35    \\
\kw{PO} & social      & 1.6M   & 1.2M   & 22.3M   & 25    \\
\kw{WI} & authorship  & 953.5K & 5.9M   & 30.6M   & 156   \\
\kw{LI} & affiliation & 3.2M   & 7.5M   & 112.3M  & 104   \\ \hline
\end{tabular}
\end{table}

\subsection{Performance studies on static graphs}

\begin{figure}[t]
    \centering
    \includegraphics[width=0.95\linewidth]{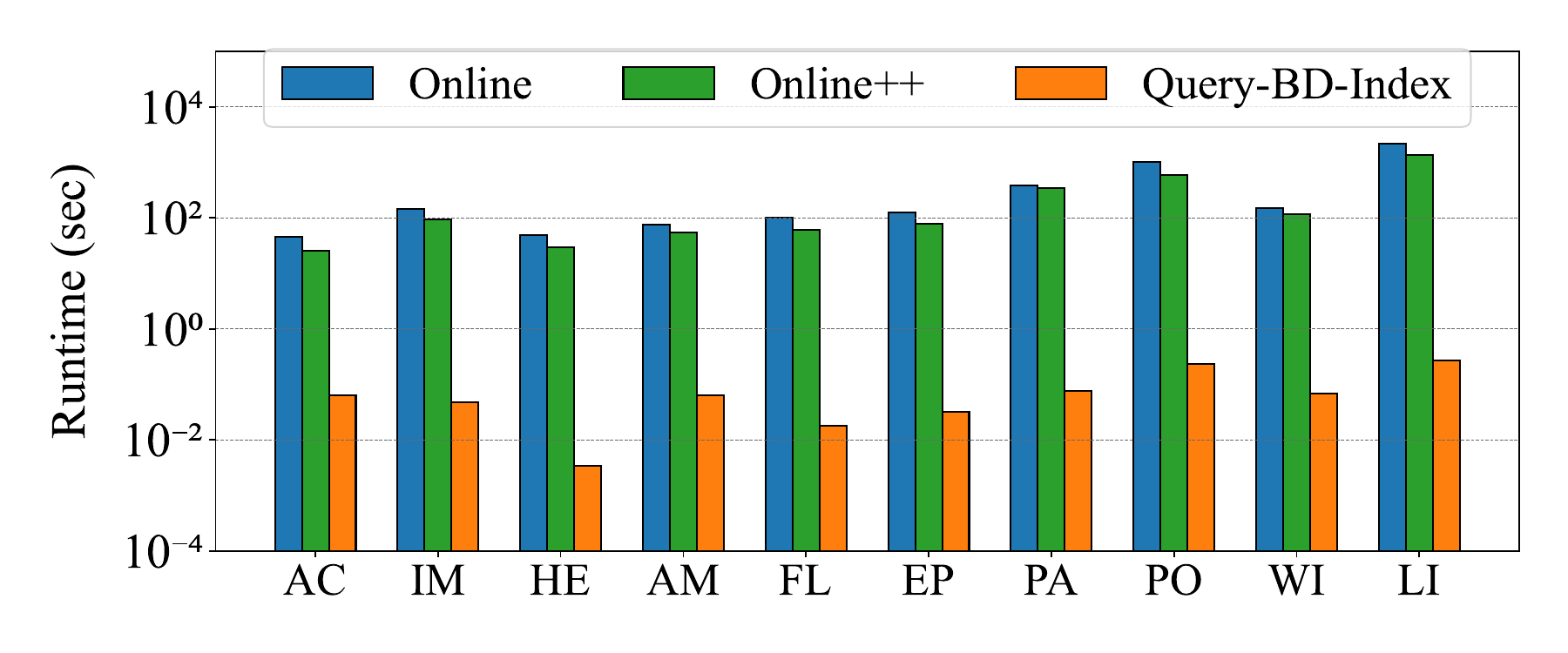}
    \vspace*{-0.5cm}
    \caption{\color{\mycolorb}Runtime of different query algorithms (total time of 100 random queries).} \label{query}
\end{figure}

{\color{\mycolorb}\stitle{Exp-\expnumber: Query processing time of different algorithms.} In this experiment, we evaluate the performance of $(\a,\b)$-dense subgraph query algorithms, including \kw{Online}, \kw{Online++}, and \kw{Query-BD-Index}. For each dataset, we execute 100 queries with $\a$ and $\b$ parameters uniformly sampled from $[0, p]$, measuring the total processing time. The results are shown in \figref{query}. \kw{Query-BD-Index} achieves a speedup of 3 to 4 orders of magnitude over \kw{Online} and \kw{Online++}. For example, on the \kw{HE} dataset, the running times of \kw{Online}, \kw{Online++}, and \kw{Query-BD-Index} are 49.84 seconds, 29.82 seconds, and 0.0034 seconds, respectively. On the large-scale dataset \kw{LI} with over 100 million edges, the maximum query time of \kw{Online}, \kw{Online++}, and \kw{Query-BD-Index} across all 100 queries are 60.27 seconds, 39.86 seconds, and only 0.018 seconds, respectively. This striking gap highlights that both online algorithms fail to satisfy the stringent real-time response requirements of practical applications (typically under 0.5 seconds \cite{amazon}), whereas our index-based method fully meets the requirement. Although \kw{Online++} reduces redundant computations and achieves moderate improvements over \kw{Online}, its performance gain is far from sufficient to meet real-time demands, further underscoring the necessity of our index-based approach.
}

\begin{figure}[t]
  \captionsetup[subfigure]{justification=centering}
  \begin{subfigure}{0.48\linewidth}
    \centering
    \includegraphics[width=1.0\linewidth]{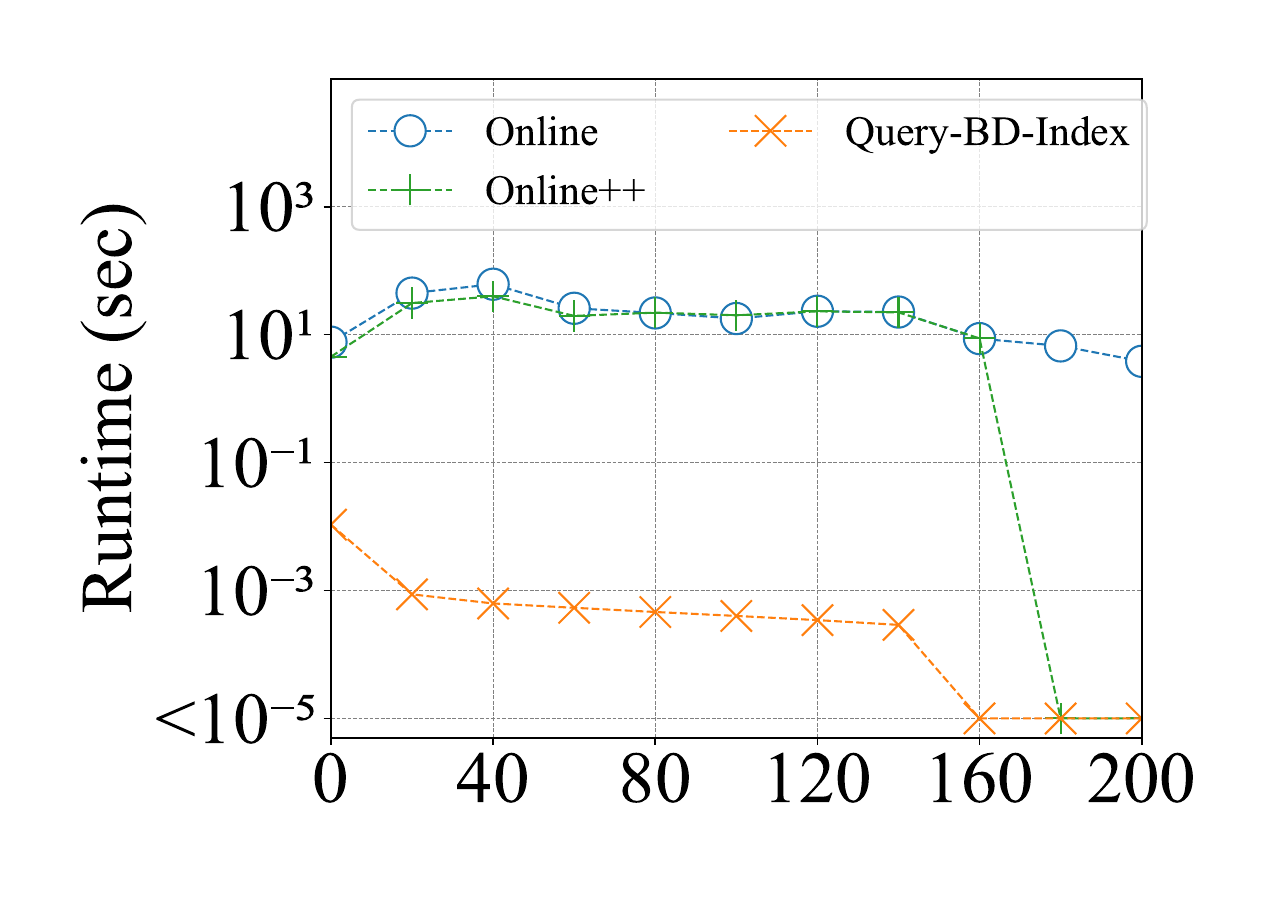}
    \vspace*{-0.8cm}
    \subcaption[font=scriptsize]{\color{\mycolorb}\textmd{$\a=100$, varying $\b$.}} \label{varyb}
  \end{subfigure}
  \begin{subfigure}{0.48\linewidth}
    \centering
    \includegraphics[width=1.0\linewidth]{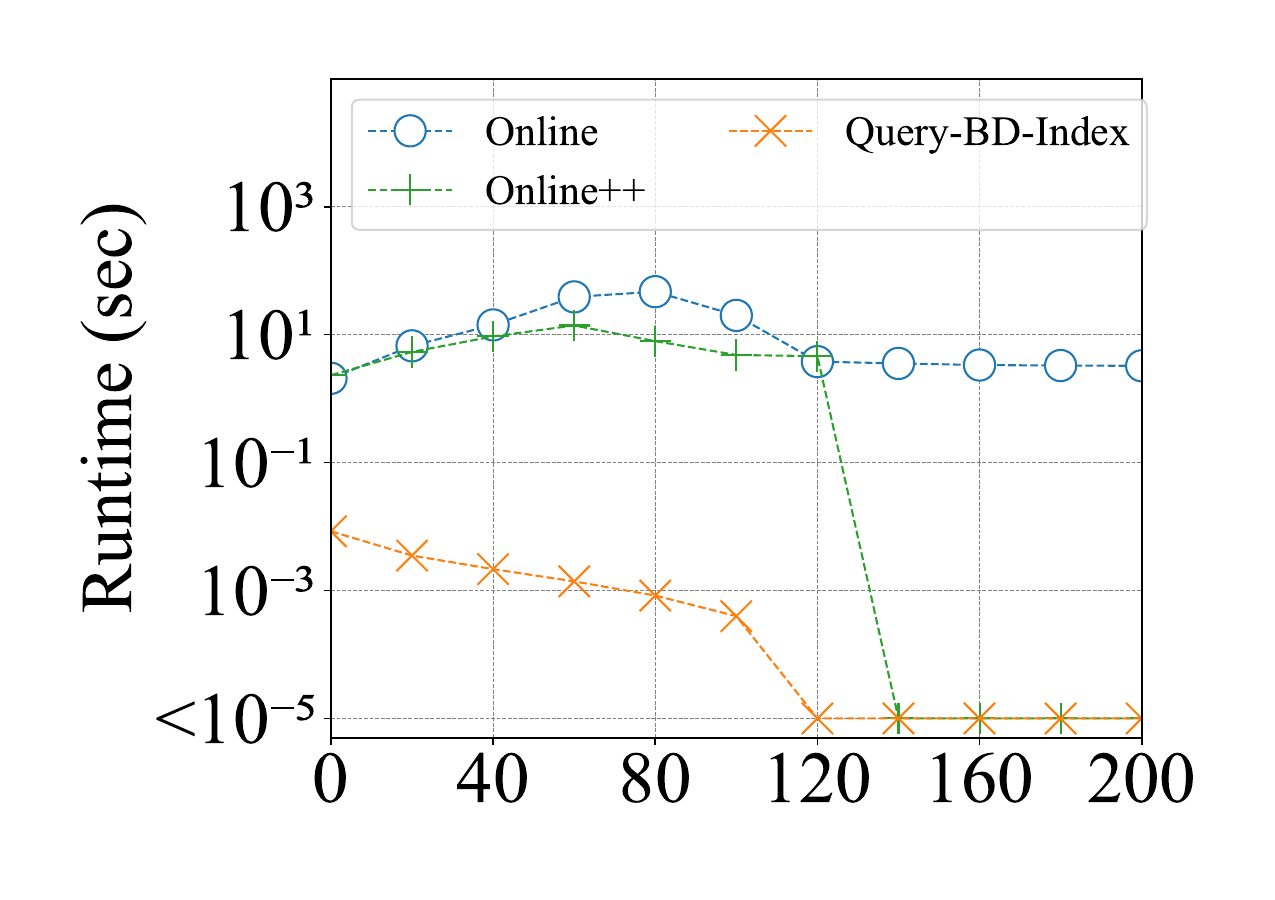}
    \vspace*{-0.8cm}
    \subcaption[font=scriptsize]{\color{\mycolorb}\textmd{$\b=100$, varying $\a$.}} \label{varya}
  \end{subfigure}
  \vspace*{-0.2cm}
  \caption{\color{\mycolorb}Running time with varying $\a$ and $\b$ on dataset \kw{LI}.} \label{queryab}
\end{figure}

{\color{\mycolorb}\stitle{Exp-\expnumber: Query processing time of different algorithms with varying $\a$ and $\b$.} In this experiment, we evaluate the performance of \kw{Online}, \kw{Online++}, and \kw{Query-BD-Index} across different $(\a, \b)$ values. The results on the \kw{LI} dataset are shown in \figref{queryab} (other datasets exhibit similar trends). As observed, the index-based algorithm \kw{Query-BD-Index} consistently outperforms the online algorithms across all parameter settings, achieving a speedup of 2 to 5 orders of magnitude. For \kw{Online++}, it is query-efficient only when $(\alpha,\beta)$ are large enough such that $D_{\alpha,\beta}=\emptyset$ (where its cache-based optimization is effective); when $(\alpha,\beta)$ are smaller, its performance is as slow as \kw{Online}. These results demonstrate that, compared to the online algorithms, \kw{Query-BD-Index} maintains high efficiency across all $(\alpha, \beta)$ combinations.
}

\begin{figure}[t]
    \centering
    \includegraphics[width=0.95\linewidth]{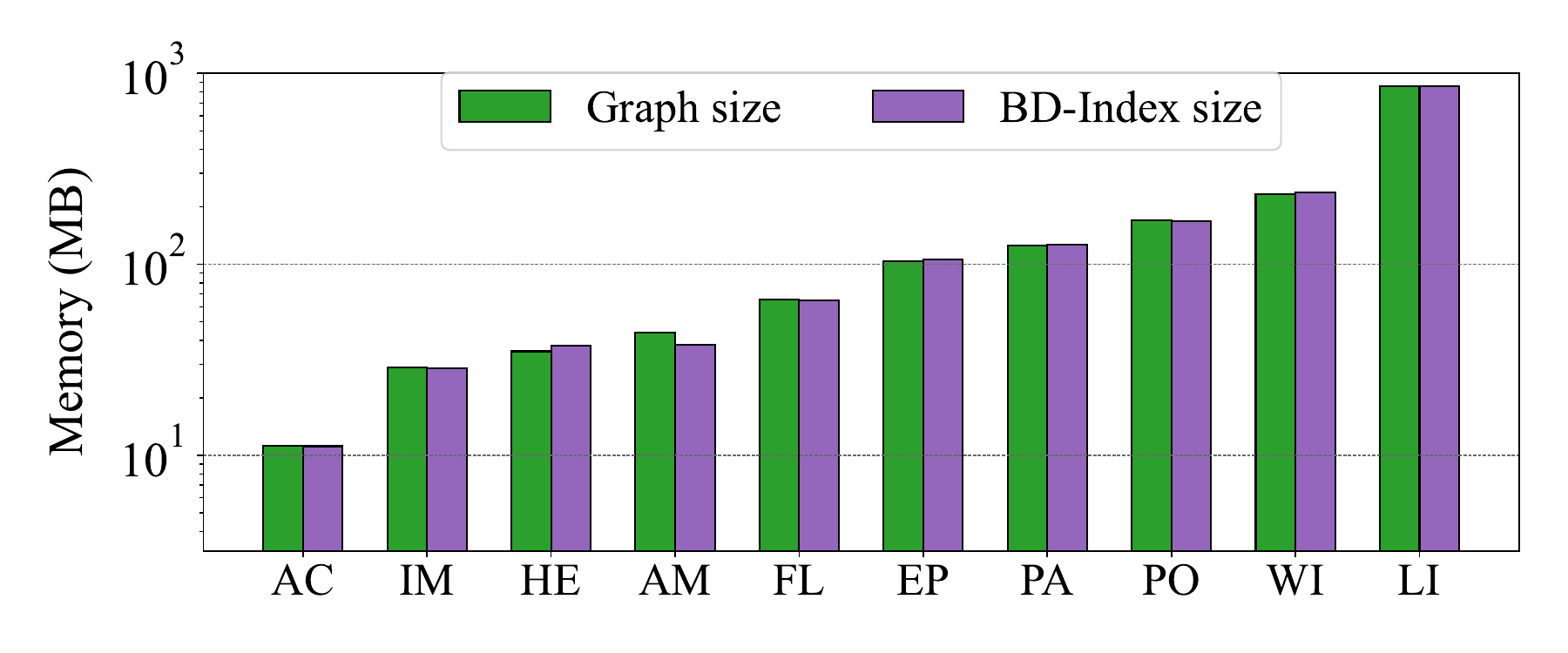}
    \vspace*{-0.5cm}
    \caption{Memory usage of graph and \kw{BD-Index}.} \label{memory}
\end{figure}

\stitle{Exp-\expnumber: Index space usage.} This experiment evaluates the memory consumption of \kw{BD-Index} compared to the graph size. Given that each edge requires storing two endpoints (4 bytes each), we compute the graph size as $8|E|$ bytes. \figref{memory} shows the comparative results. As seen, \kw{BD-Index} exhibits near-identical memory requirements to the graph size across all datasets. For example, on dataset \kw{LI}, the graph occupies 856.3MB while \kw{BD-Index} requires 856.8MB. These results align well with the theoretical $O(|E|)$ space complexity of \kw{BD-Index}, highlighting its highly space-efficient advantage.

\begin{table}[t]
\caption{The construction time of \kw{BD-Index}.} \label{construction}
\vspace*{-0.3cm}
\small
\begin{tabular}{cc|cc}
\hline
Dataset & Runtime (sec) & Dataset & Runtime (sec)  \\ \hline
\kw{AC} & 76.2    & \kw{EP} & 3,952.0  \\
\kw{IM} & 469.2   & \kw{PA} & 2,198.0  \\
\kw{HE} & 2,646.6 & \kw{PO} & 5,250.7  \\
\kw{AM} & 405.7   & \kw{WI} & 7,659.8  \\
\kw{FL} & 3,010.5 & \kw{LI} & 65,194.9 \\ \hline
\end{tabular}
\end{table}

\stitle{Exp-\expnumber: Index construction time.} \tabref{construction} presents the construction times of \kw{BD-Index} using the \kw{Build-BD-Index} algorithm acorss all datasets. As shown, \kw{BD-Index} can be efficiently constructed at various scales. For smaller datasets such as \kw{AC} and \kw{IM}, the index is built in under 500 seconds. For mid-sized graphs like \kw{HE}, \kw{AM}, and \kw{PA}, construction completes in a few thousand seconds. Notably, even on the largest dataset \kw{LI}, which contains over 112 million edges, the index is constructed in approximately 18 hours (65,194.9 seconds). These results demonstrate that \kw{BD-Index} can be constructed within reasonable time even for large-scale graphs, and the practical performance of \kw{Build-BD-Index} significantly outperforms its worst-case time complexity of $O(p \cdot |E|^{1.5} \cdot \log |U\cup V|)$.

\begin{figure}[t]
  \captionsetup[subfigure]{justification=centering}
  \begin{subfigure}{0.45\linewidth}
    \centering
    \includegraphics[width=1.0\linewidth]{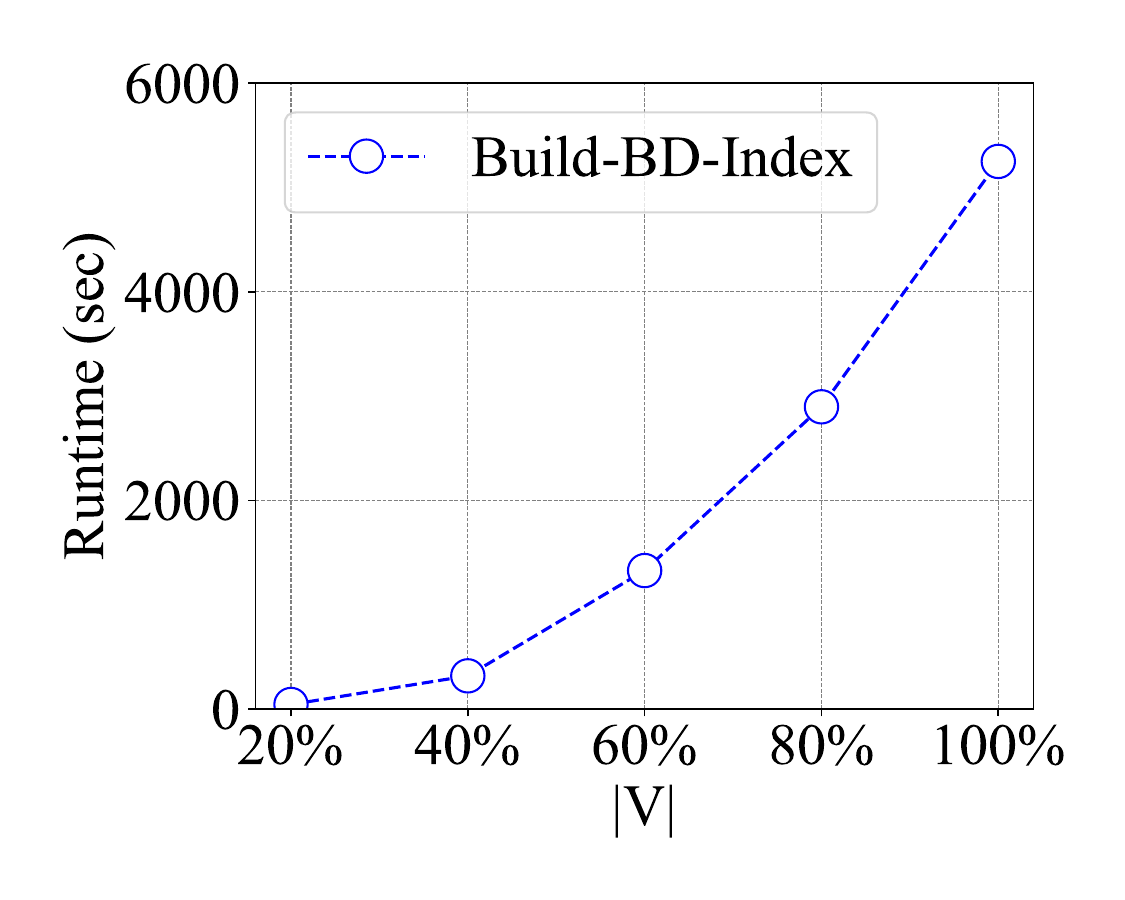}
    \vspace*{-0.8cm}
    \subcaption[font=scriptsize]{\textmd{Dataset \kw{PO}, varying $|V|$.}} \label{pokecV}
  \end{subfigure}
  \begin{subfigure}{0.45\linewidth}
    \centering
    \includegraphics[width=1.0\linewidth]{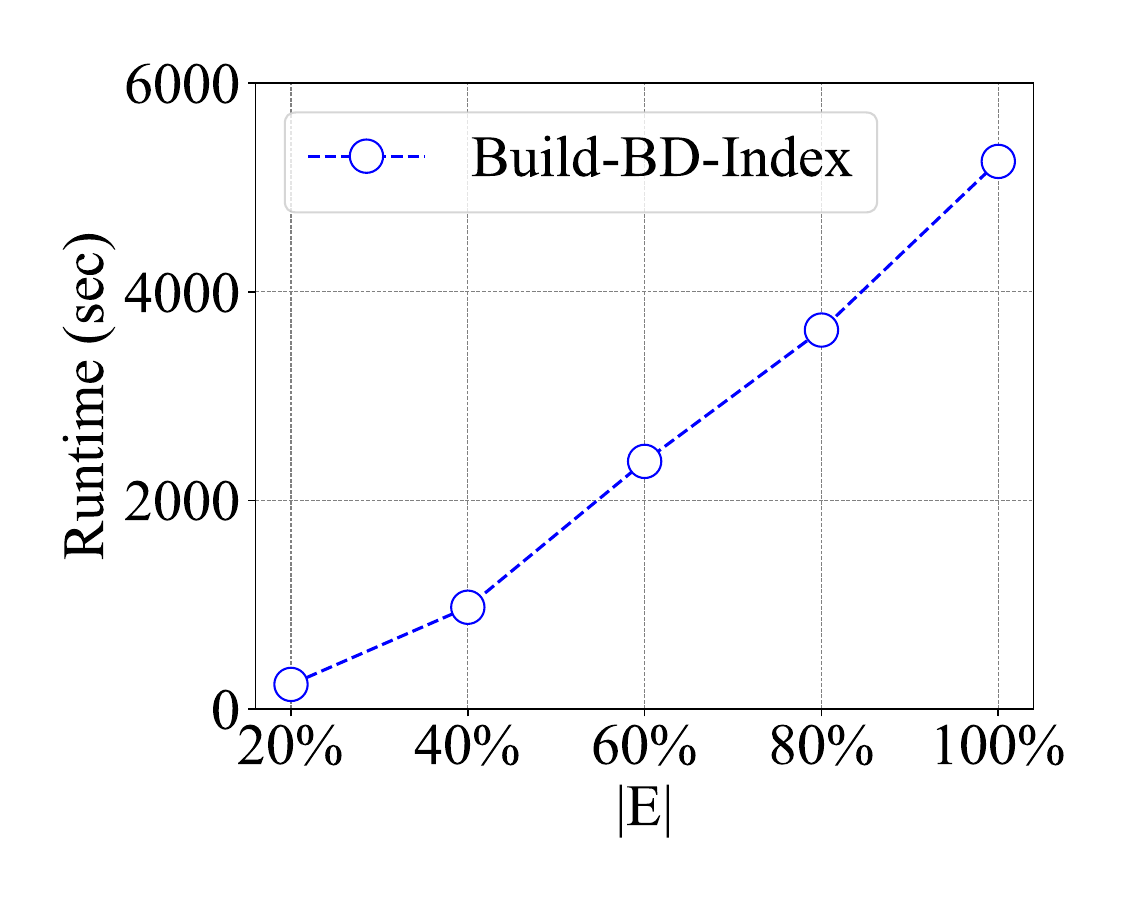}
    \vspace*{-0.8cm}
    \subcaption[font=scriptsize]{\textmd{Dataset \kw{PO}, varying $|E|$.}} \label{pokecE}
  \end{subfigure}
  \begin{subfigure}{0.45\linewidth}
    \centering
    \includegraphics[width=1.0\linewidth]{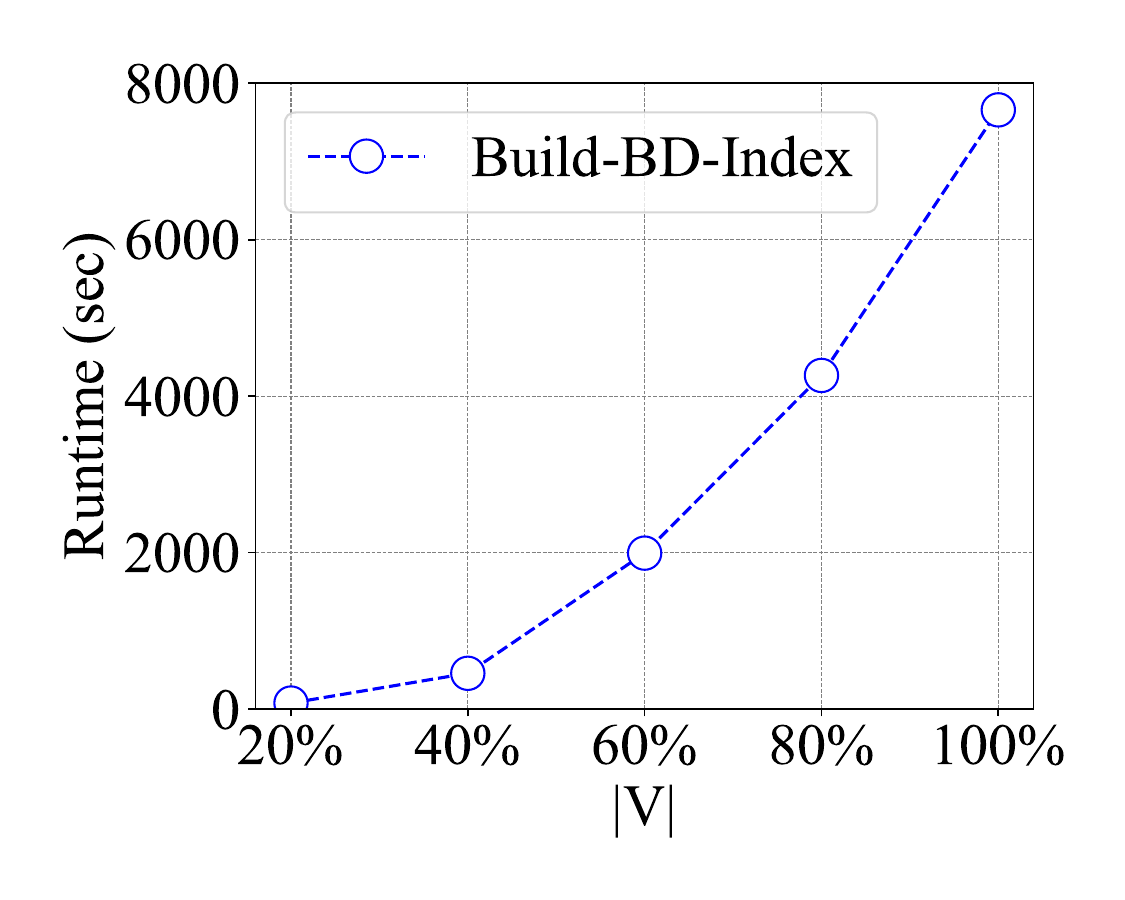}
    \vspace*{-0.8cm}
    \subcaption[font=scriptsize]{\textmd{Dataset \kw{WI}, varying $|V|$.}} \label{wikiesV}
  \end{subfigure}
  \begin{subfigure}{0.45\linewidth}
    \centering
    \includegraphics[width=1.0\linewidth]{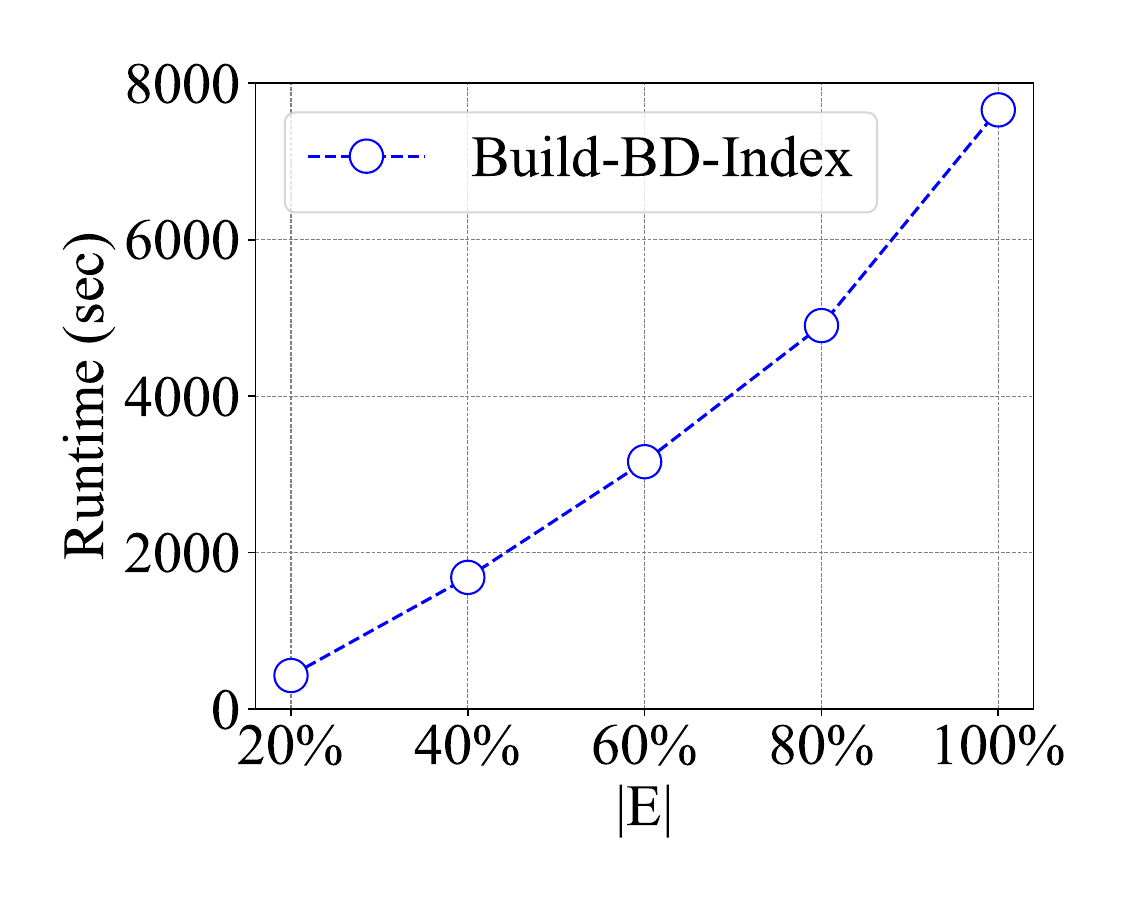}
    \vspace*{-0.8cm}
    \subcaption[font=scriptsize]{\textmd{Dataset \kw{WI}, varying $|E|$.}} \label{wikiesE}
  \end{subfigure}
  \vspace*{-0.2cm}
  \caption{Scalability test of \kw{Build-BD-Index} algorithm.} \label{scal}
\end{figure}

\stitle{Exp-\expnumber: Scalability test of index construction.} We evaluate the scalability of the \kw{Build-BD-Index} algorithm by constructing \kw{BD-Index} for subgraphs containing $\{20\%, 40\%, 60\%, 80\%, 100\%\}$ of the original vertices or edges. \figref{scal} shows the runtime results for \kw{PO} and \kw{WI}, with other datasets exhibiting similar trends. The runtime increases smoothly and predictably with the growth of $|V|$ and $|E|$, indicating that \kw{Build-BD-Index} scales well with both graph size dimensions. In all cases, the growth trend remains stable, with no sudden spikes or inefficiencies observed. These results confirm the strong scalability of our \kw{Build-BD-Index} algorithm.

\subsection{Index maintenance on dynamic graphs}

\begin{figure}[t]
    \centering
    \includegraphics[width=0.98\linewidth]{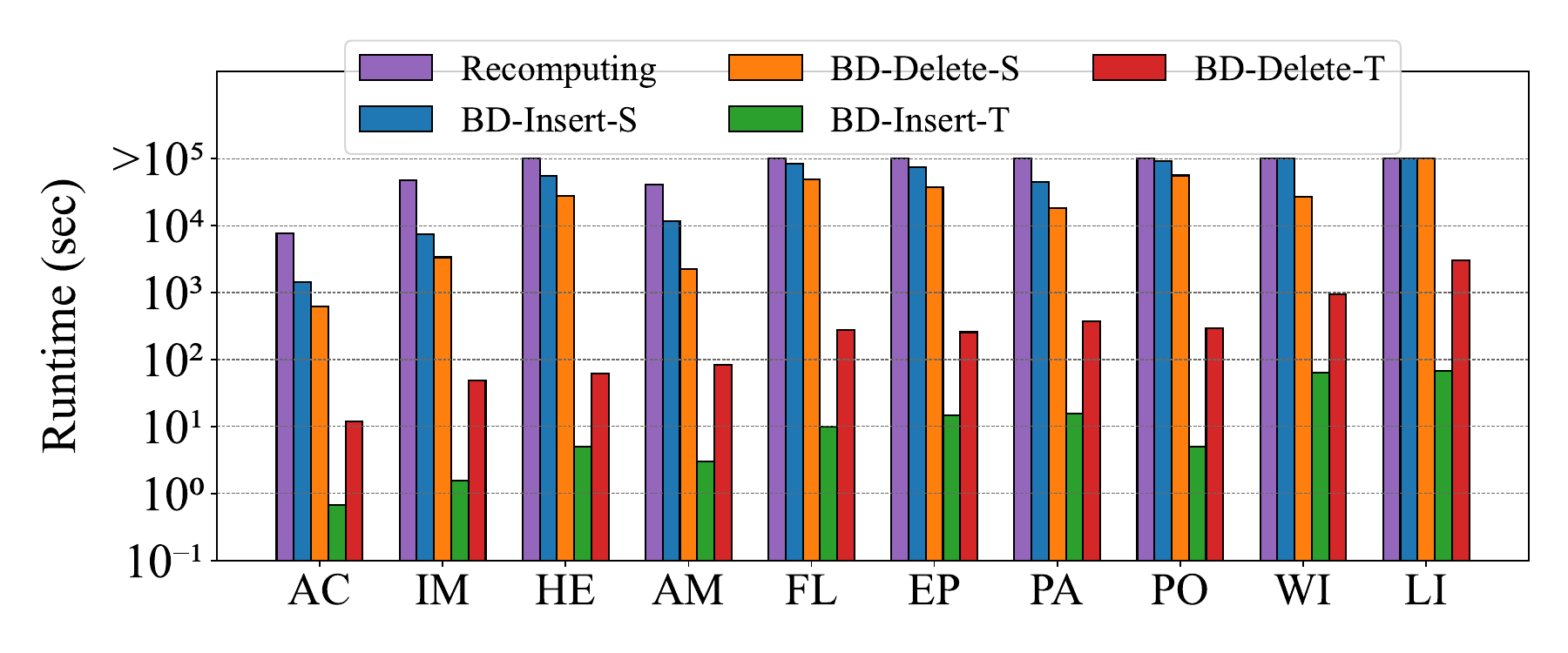}
    \vspace*{-0.5cm}
    \caption{Runtime of maintenance algorithms of \kw{BD-Index} (total time for processing 100 random edge deletions and insertions).} \label{dynamic}
\end{figure}

\stitle{Exp-\expnumber: Runtime of index maintenance algorithms.} Here we evaluate the runtime of maintenance algorithms for \kw{BD-Index}, including the baseline \kw{Recomputing}, space-efficient algorithms (\kw{BD-Insert-S} and \kw{BD-Delete-S}), and time-efficient algorithms (\kw{BD-Insert-T} and \kw{BD-Delete-T}). For each dataset, we perform 100 random edge updates (deletions followed by re-insertions) and measure the total processing time. The results are presented in \figref{dynamic}. 

As seen, the baseline algorithm \kw{Recomputing} is extremely slow, completing within the $10^5$-second runtime limit only for datasets \kw{AC}, \kw{IM}, and \kw{AM}. It is approximately one order of magnitude slower than the space-efficient algorithms and 2–4 orders of magnitude slower than the time-efficient algorithms. For space-efficient algorithms, \kw{BD-Insert-S} and \kw{BD-Delete-S} exhibit comparable runtimes, but both exceed the $10^5$-second limit on the large dataset \kw{LI}. In contrast, the time-efficient algorithms \kw{BD-Insert-T} and \kw{BD-Delete-T} are substantially faster, achieving 3–4 orders of magnitude and 1–2 orders of magnitude speedups, respectively, over their space-efficient counterparts. For example, on dataset \kw{PO}, the total runtimes of \kw{BD-Insert-S}, \kw{BD-Delete-S}, \kw{BD-Insert-T}, and \kw{BD-Delete-T} for processing 100 edge deletions and insertions are 91,860 seconds, 55,972 seconds, 4.9 seconds, and 291 seconds, respectively, corresponding to speedups of 18,747$\times$ and 192$\times$ for insertion and deletion. These results demonstrate the superior efficiency of our proposed \kw{BD-Insert-T} and \kw{BD-Delete-T}.

Additionally, we observe that \kw{BD-Insert-T} is about an order of magnitude faster than \kw{BD-Delete-T}. This performance difference stems from the inherent complexity of edge deletion operations in \kw{BD-Delete-T}, compared to the relatively straightforward implementation of \kw{BD-Insert-T}. These results show that maintaining \kw{BD-Index} for edge insertions is more efficient than for deletions, highlighting our algorithm's practical advantages in real-world applications where graph updates primarily consist of edge insertions.

\begin{figure}[t]
    \centering
    \includegraphics[width=0.95\linewidth]{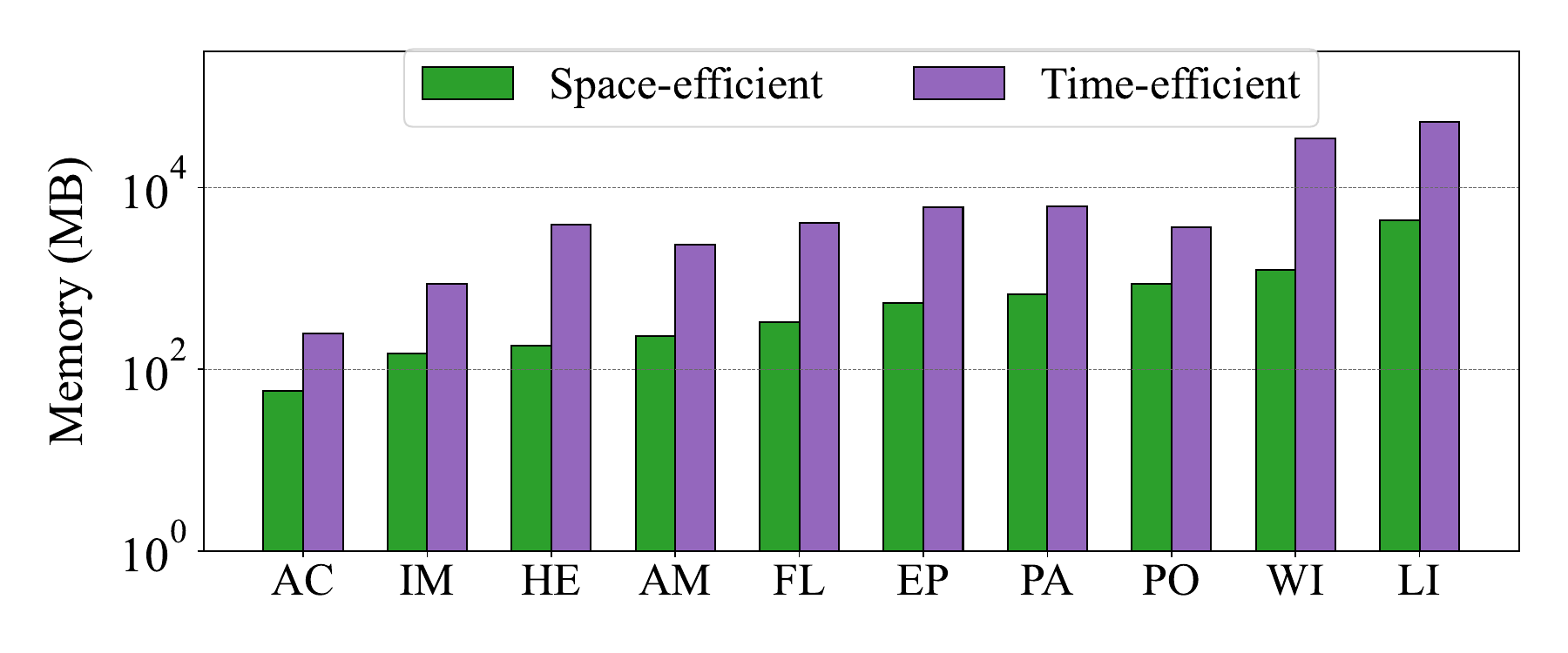}
    \vspace*{-0.5cm}
    \caption{Memory usage of two maintenance approaches.} \label{dynamic_memory}
\end{figure}

\stitle{Exp-\expnumber: Memory usage of two maintenance approach.} We evaluate the memory overhead of two maintenance strategies for \kw{BD-Index}: the space-efficient approach and the time-efficient approach. The total memory consumption comprises the \kw{BD-Index} size, the maintenance process overhead, and the storage for egalitarian orientations (used only in the time-efficient approach). The results are shown in \figref{dynamic_memory}. As seen, the time-efficient approach consumes 4 to 28 times more memory than the space-efficient one. For example, on the largest dataset \kw{LI}, the space-efficient method uses 4,413 MB, while the time-efficient method requires 52,597 MB. Notably, even for massive graphs like \kw{LI} (containing over 100 million edges), the time-efficient method's memory overhead remains practical at approximately 51 GB, well within modern server-grade hardware capacities. On the other hand, in terms of runtime, the time-efficient approach achieves up to four orders of magnitude speedup over the space-efficient method. This highlights that our time-efficient approach offers a favorable time-space trade-off.


%

\begin{figure}[t]
  \captionsetup[subfigure]{justification=centering}
  \begin{subfigure}{0.48\linewidth}
    \centering
    \includegraphics[width=1.0\linewidth]{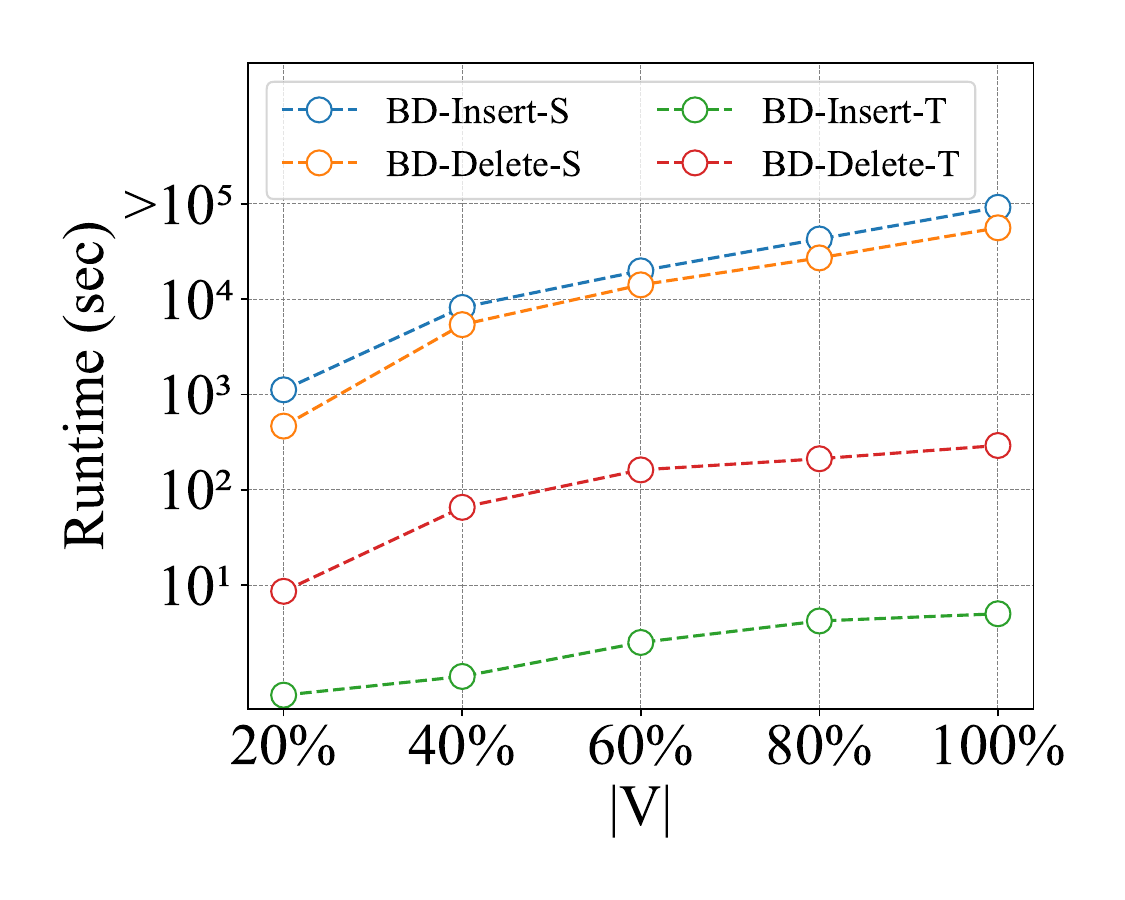}
    \vspace*{-0.8cm}
    \subcaption[font=scriptsize]{\textmd{Dataset \kw{PO}, varying $|V|$.}} \label{dynamic_pokecV}
  \end{subfigure}
  \begin{subfigure}{0.48\linewidth}
    \centering
    \includegraphics[width=1.0\linewidth]{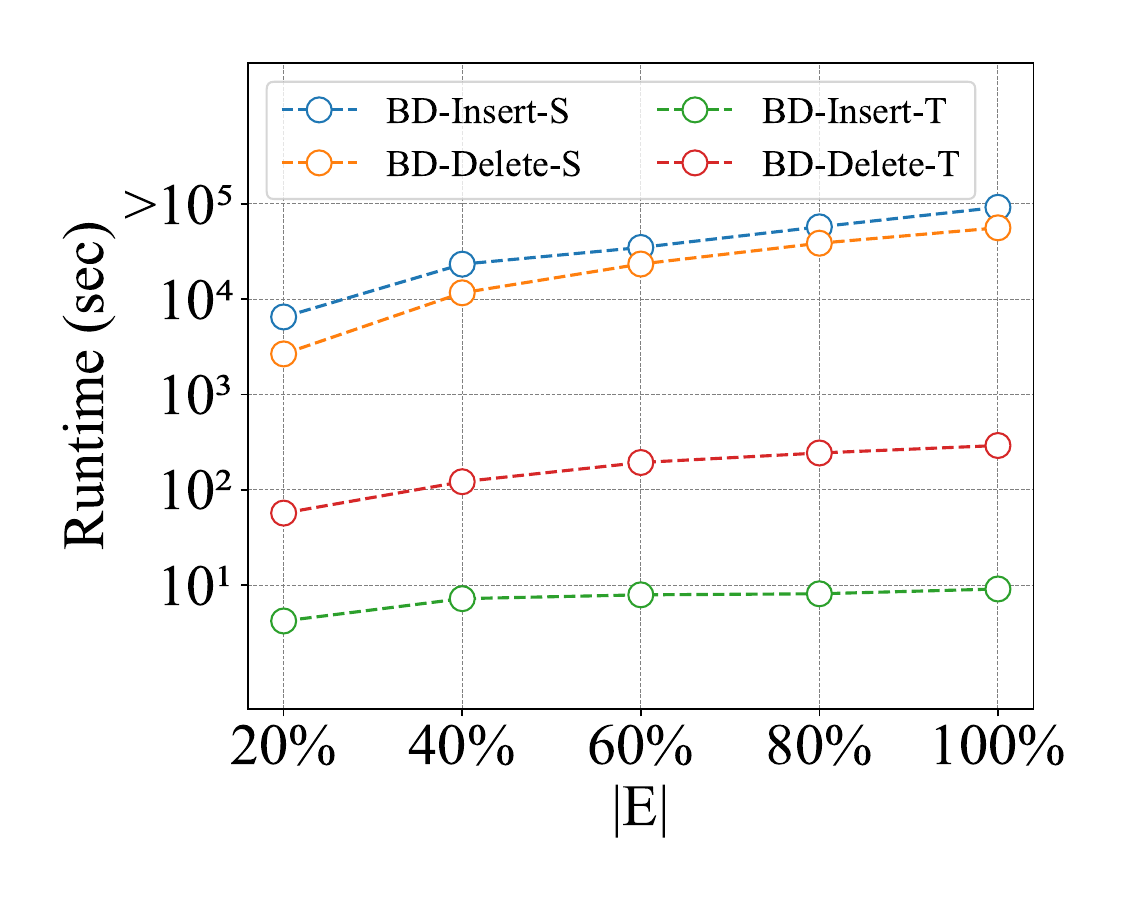}
    \vspace*{-0.8cm}
    \subcaption[font=scriptsize]{\textmd{Dataset \kw{PO}, varying $|E|$.}} \label{dynamic_pokecE}
  \end{subfigure}
  \begin{subfigure}{0.48\linewidth}
    \centering
    \includegraphics[width=1.0\linewidth]{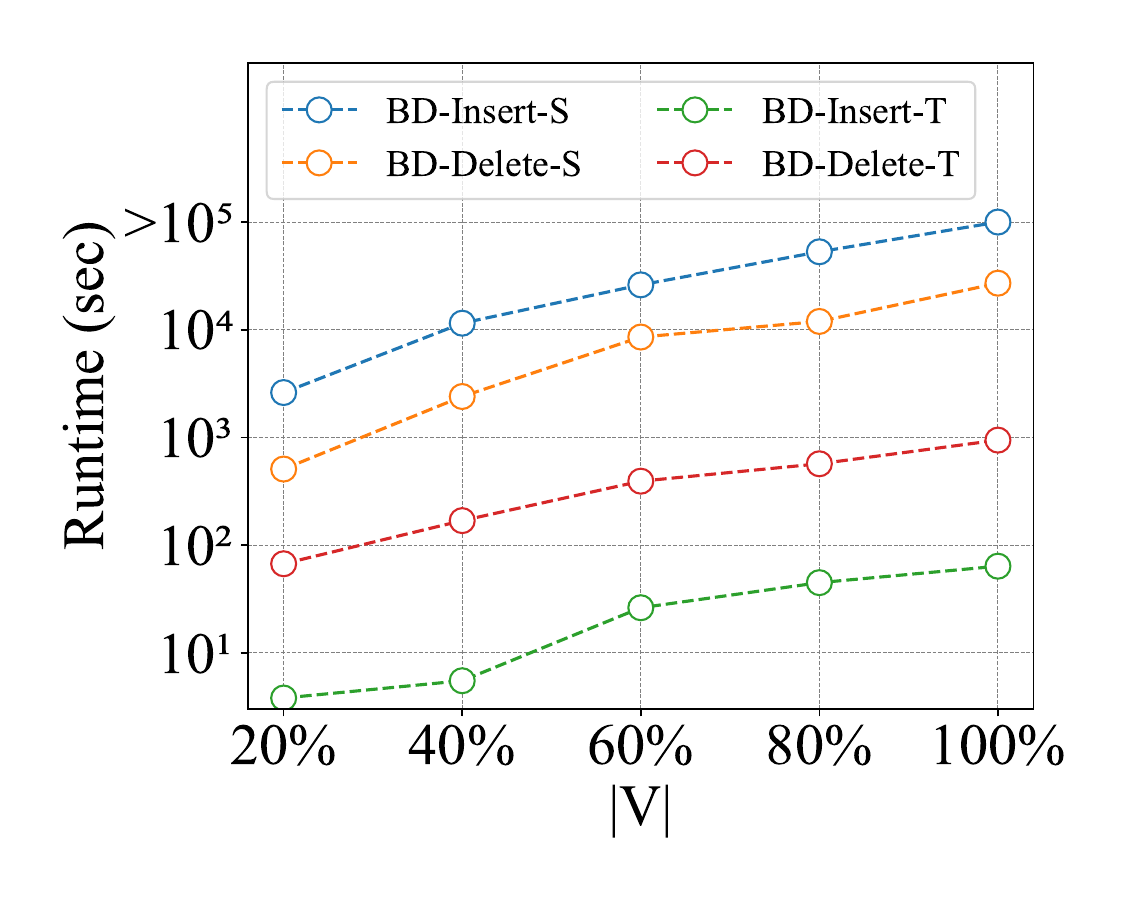}
    \vspace*{-0.8cm}
    \subcaption[font=scriptsize]{\textmd{Dataset \kw{WI}, varying $|V|$.}} \label{dynamic_wikiesV}
  \end{subfigure}
  \begin{subfigure}{0.48\linewidth}
    \centering
    \includegraphics[width=1.0\linewidth]{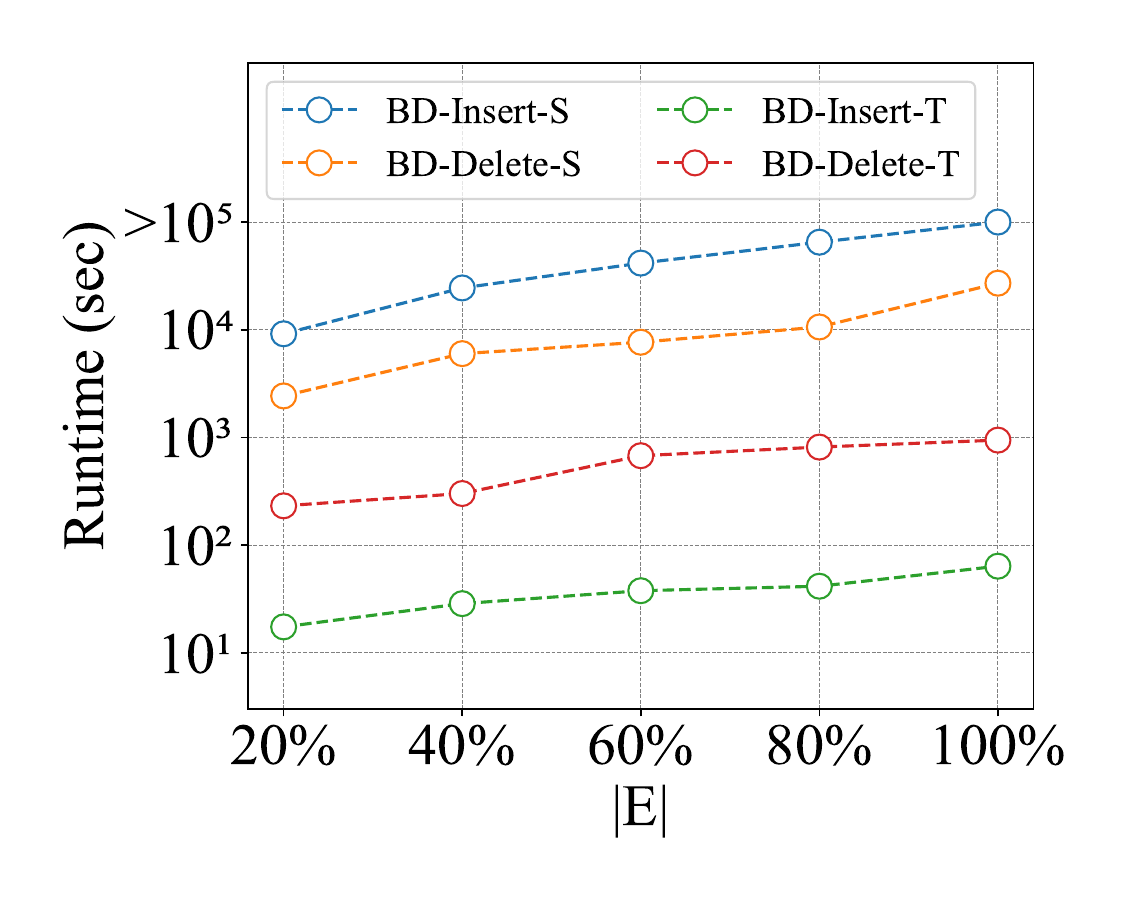}
    \vspace*{-0.8cm}
    \subcaption[font=scriptsize]{\textmd{Dataset \kw{WI}, varying $|E|$.}} \label{dynamic_wikiesE}
  \end{subfigure}
  \vspace*{-0.2cm}
  \caption{Scalability test of maintenance algorithms (measuring by the total time for processing 100 random edge deletions and insertions).} \label{dynamic_scal}
\end{figure}

\stitle{Exp-\expnumber: Scalability test of index maintenance algorithms.} This experiment evaluates the scalability of our maintenance algorithms using the subgraphs generated in Exp-5. For each subgraph, we perform 100 random edge deletions followed by re-insertions, measuring the total runtime. \figref{dynamic_scal} shows the results on \kw{PO} and \kw{WI}, with other datasets exhibiting similar trends. As shown, the space-efficient algorithms \kw{BD-Insert-S} and \kw{BD-Delete-S} exhibit slow runtimes, and even encounter timeout issues (>$10^5$ seconds) when the graph becomes large. In contrast, the time-efficient algorithms \kw{BD-Insert-T} and \kw{BD-Delete-T} maintain fast performance with gradual runtime increases as graphs grow. These findings demonstrate the superior scalability of the time-efficient algorithms in handling edge updates on large-scale graphs.

\stitle{Exp-\expnumber: Effect of different edge update strategies.} This experiment evaluates the maintenance algorithms under different edge update strategies. We define the degree of an edge as the sum of its endpoints' degrees. The edges are then partitioned into three categories: low-degree (edges with lowest-$1/3$ degree), medium-degree (edges with middle-$1/3$ degree), and high-degree (edges with highest-$1/3$ degree), representing regions with different density in the graph. For each category, we randomly select 100 edges for deletion and re-insertion, forming distinct update strategies. \figref{strategy} shows the total runtime across all datasets under different update strategies. The time-efficient algorithms \kw{BD-Insert-T} and \kw{BD-Delete-T} consistently outperform the space-efficient algorithms \kw{BD-Insert-S} and \kw{BD-Delete-S} by 1-5 orders of magnitude. Notably, our proposed algorithms exhibit minimal runtime variation across different edge selection strategies. In particular, the time-efficient algorithms maintain high performance regardless of update edge type, confirming their robustness.

\begin{figure}[t]
  \captionsetup[subfigure]{justification=centering}
  \begin{subfigure}{0.95\linewidth}
    \centering
    \includegraphics[width=1.0\linewidth]{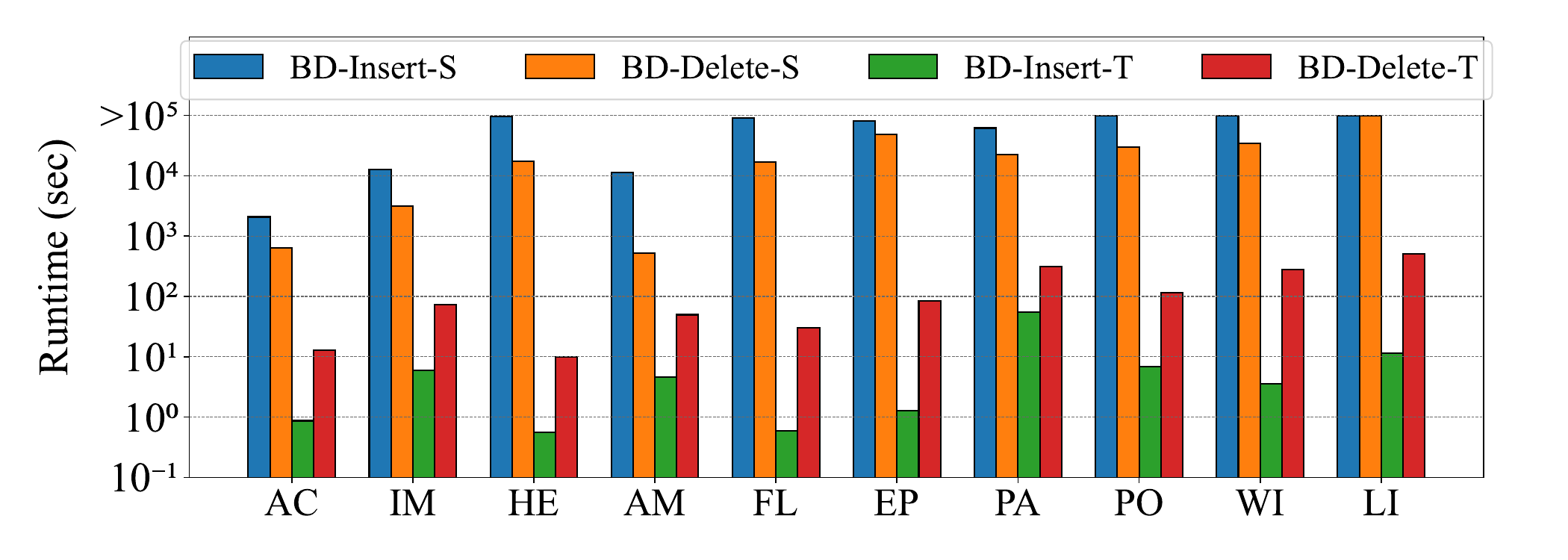}
    \vspace*{-0.8cm}
    \subcaption[font=scriptsize]{\textmd{Low-degree edges.}} \label{low}
  \end{subfigure}
  \begin{subfigure}{0.95\linewidth}
    \centering
    \includegraphics[width=1.0\linewidth]{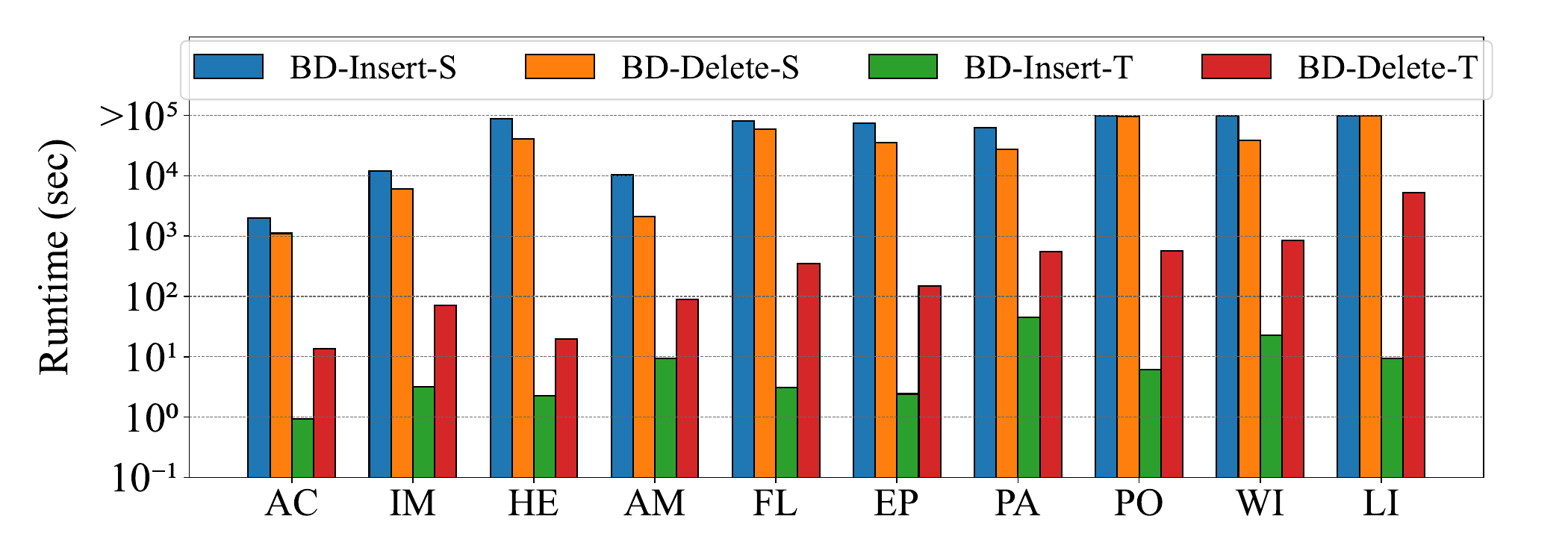}
    \vspace*{-0.8cm}
    \subcaption[font=scriptsize]{\textmd{Medium-degree edges.}} \label{mid}
  \end{subfigure}
  \begin{subfigure}{0.95\linewidth}
    \centering
    \includegraphics[width=1.0\linewidth]{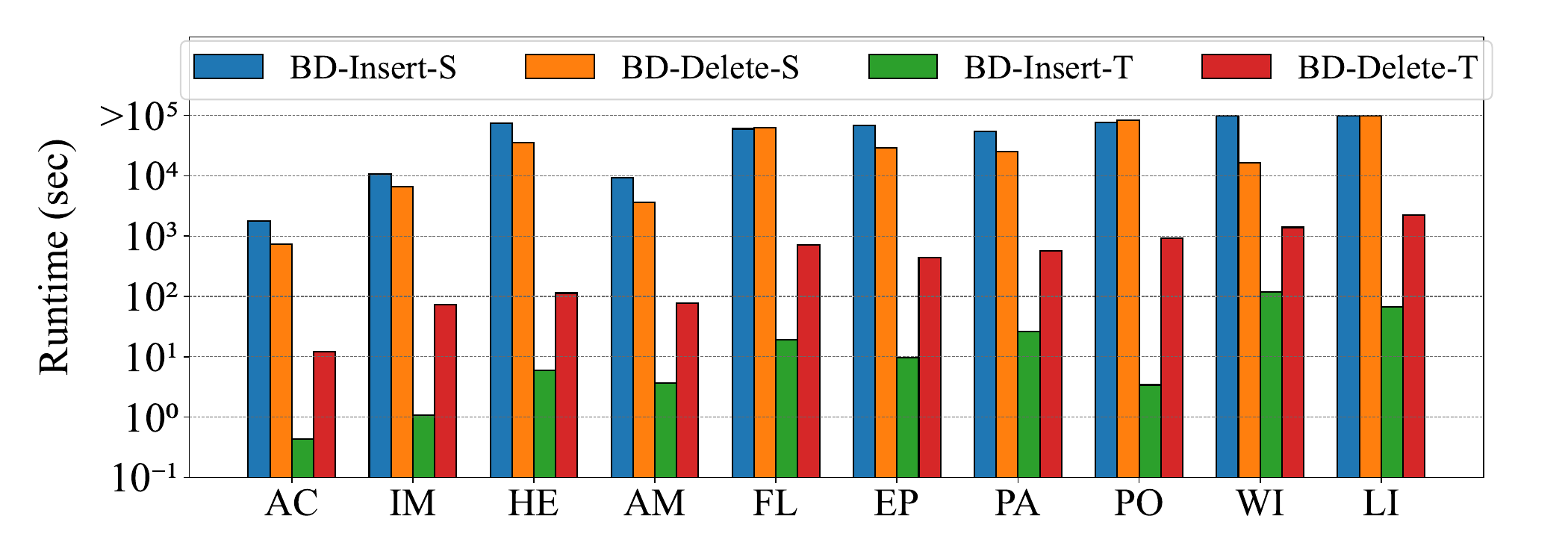}
    \vspace*{-0.8cm}
    \subcaption[font=scriptsize]{\textmd{High-degree edges.}} \label{high}
  \end{subfigure}
  \vspace*{-0.2cm}
  \caption{Runtime of maintenance algorithms with different edge selection strategies (total time of 100 random edge deletions and insertions).} \label{strategy}
\end{figure}

{\color{\mycolorb}

\subsection{Case studies}

\stitle{Performance of different algorithms under mixed updates and queries.} To better simulate the update and query patterns of real streaming systems, we conduct a case study on the real-world temporal \kw{AM-app} dataset\footnote{\color{\mycolorb}Data source: \url{https://amazon-reviews-2023.github.io/data_processing/0core.html}.}, which is an Amazon e-commerce user–product review network for the appliances category. The dataset contains $|E|=2.1$M edges, $|U|=1.7$M users, and $|V|=94.3$K products, where each edge is associated with a timestamp representing a review. 

First, we simulate edge updates in an insertion-only mode: specifically, the latest 2,000 edges (covering the period from 2023-08-13 to 2023-09-13) are incrementally inserted into the graph in chronological order. During the insertion process, we generate 20 query sessions at random time points. Each session lasts 10 seconds with queries distributed randomly over its duration. By varying the number of queries per session, we simulate different query frequencies. In our setup, the system processes queries and updates sequentially in chronological order, and subsequent operations enter a queue and await prior operations' completion. All reported times are therefore measured as turnaround times, defined as the duration from issue to completion of each query or update. The results are presented in \subfigref{insonly} (note that the online algorithms do not require update time). \kw{BD-Query-S} and \kw{BD-Update-S} denote the query and update times of \kw{BD-Index} under the space-efficient maintenance strategy, while \kw{BD-Query-T} and \kw{BD-Update-T} represent the corresponding times under the time-efficient strategy. We observe that as queries become frequent ($\approx$100 queries per session, equivalent to $\approx$10 queries per second), the average query turnaround times of the online algorithms \kw{Online} and \kw{Online++} quickly rises beyond 10 seconds, and their maximum turnaround times exceeds 40 seconds, indicating severe queuing delays that render them impractical for frequent-query workloads. In contrast, the space-efficient strategy maintains low average query times ($\approx$0.3 seconds), but suffers from high maximum query times ($\approx$10 seconds) due to its slow update processing (queries issued immediately after an update must wait for the update to complete). The time-efficient strategy, on the other hand, achieves consistently low query and update times, with the average and maximum query times remaining below $0.01$ seconds and $0.1$ seconds, respectively.

\begin{figure}[t]
  \captionsetup[subfigure]{justification=centering}
  \begin{subfigure}{0.9\linewidth}
    \centering
    \includegraphics[width=1.0\linewidth]{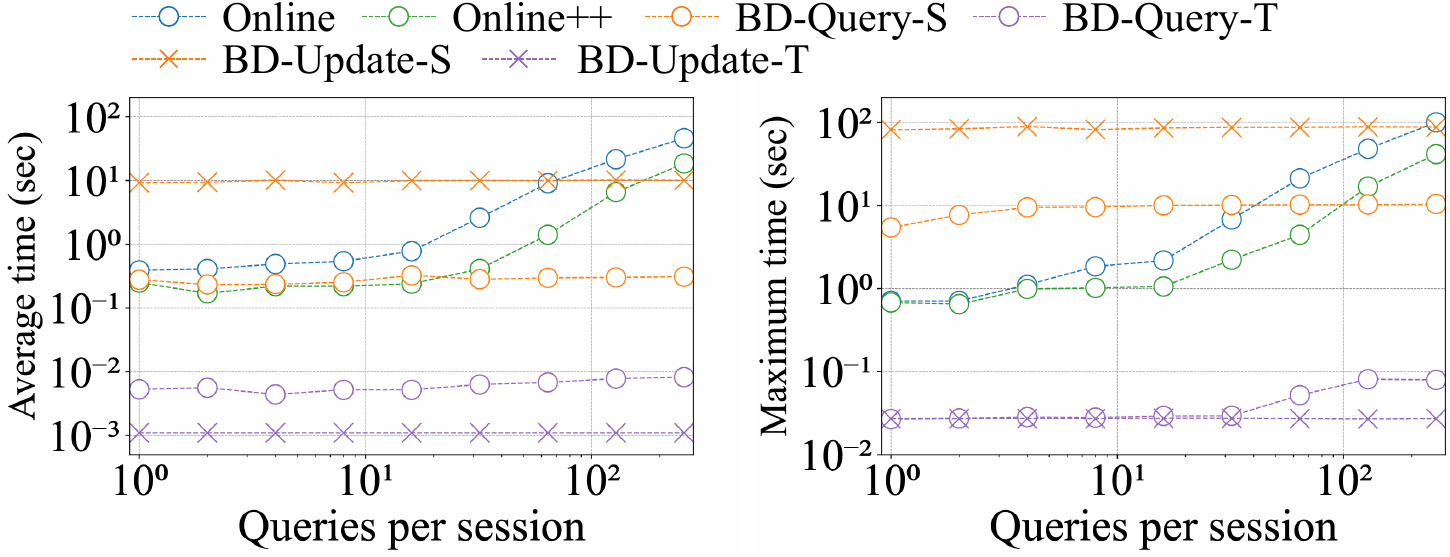}
    \vspace*{-0.5cm}
    \subcaption[font=scriptsize]{\color{\mycolorb}\textmd{\scalebox{0.97}{Average and maximum turnaround time in insertion-only case.}}} \label{insonly}
  \end{subfigure}
  \begin{subfigure}{0.9\linewidth}
    \centering
    \includegraphics[width=1.0\linewidth]{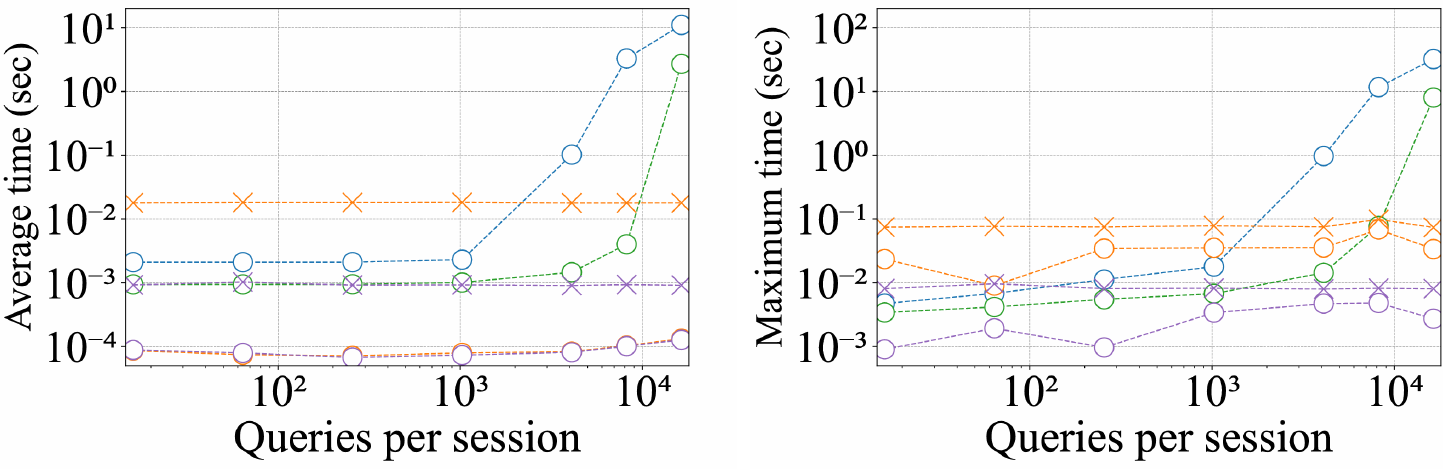}
    \vspace*{-0.5cm}
    \subcaption[font=scriptsize]{\color{\mycolorb}\textmd{\scalebox{0.97}{Average and maximum turnaround time in first-in-first-out case.}}} \label{fifo}
  \end{subfigure}
  \vspace*{-0.2cm}
  \caption{\color{\mycolorb}Query turnaround time and update turnaround time of different algorithms on dataset \kw{AM-app}.} \label{case}
\end{figure}

Second, we consider another update mode, first-in-first-out. Starting from an empty graph, we insert 100,000 edges in chronological order (spanning from 2002-11-18 to 2014-07-01) while maintaining a sliding time window of one year (i.e., edges older than one year are removed during the insertion process). The query generation follows the same procedure as in the insertion-only case, and the results are presented in \subfigref{fifo}. As the query frequency increases ($\approx$$10^4$ queries per session, equivalent to $\approx$$10^3$ queries per second), the maximum query turnaround time of \kw{Online} and \kw{Online++} rise dramatically to 32 seconds and 8 seconds, respectively. In contrast, the query turnaround times of the space-efficient and time-efficient strategies never exceed 0.1 seconds and 0.01 seconds, respectively. This efficiency stems from the fact that \kw{BD-Index} achieves optimal $O(|R|)$ query time, enabling direct retrieval of the target subgraph without redundant computation. \comment{These results indicate that the online algorithms are too slow to handle frequent queries, potentially causing timeouts (e.g., exceeding 100 seconds). For the space-efficient strategy, its query turnaround time mainly depends on its update overhead; when the graph is very large (e.g., the insertion-only case where $|E|\approx 2$M), some queries may be delayed due to waiting for updates, but most queries can still be answered within 0.3 seconds. For the time-efficient strategy, all query turnaround times in our tests remained below 0.1 seconds, fully meeting real-time requirements in practical applications. Moreover, the extra orientations maintained by this strategy required only <185 MB of additional space, making it the best-performing approach in our case study.}

\begin{figure}[t]
  \captionsetup[subfigure]{justification=centering}
  \begin{subfigure}{0.48\linewidth}
    \centering
    \includegraphics[width=1.0\linewidth]{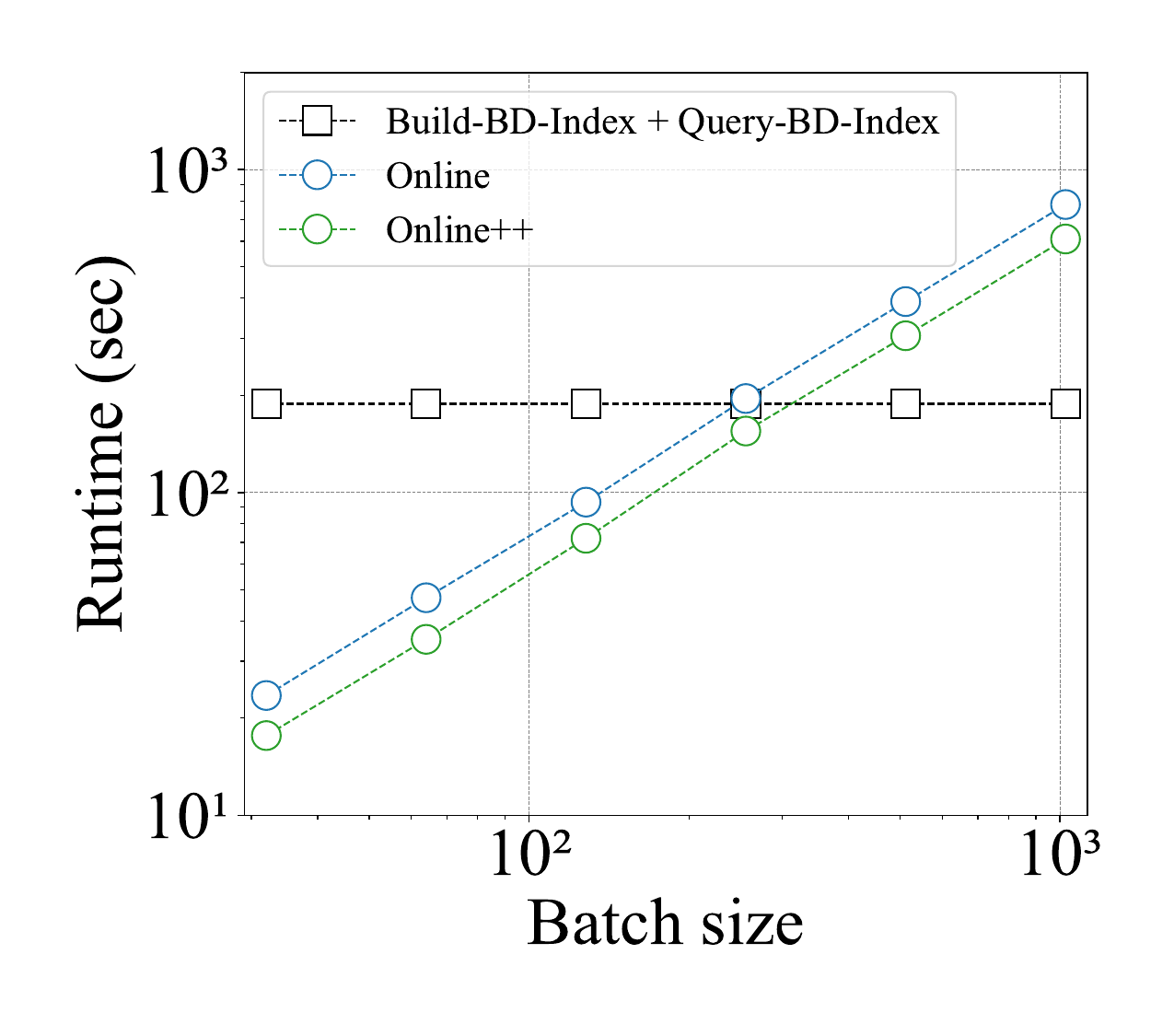}
    \vspace*{-0.8cm}
    \subcaption[font=scriptsize]{\color{\mycolorb}\textmd{Batch-query results.}} \label{batch-query}
  \end{subfigure}
  \begin{subfigure}{0.48\linewidth}
    \centering
    \includegraphics[width=1.0\linewidth]{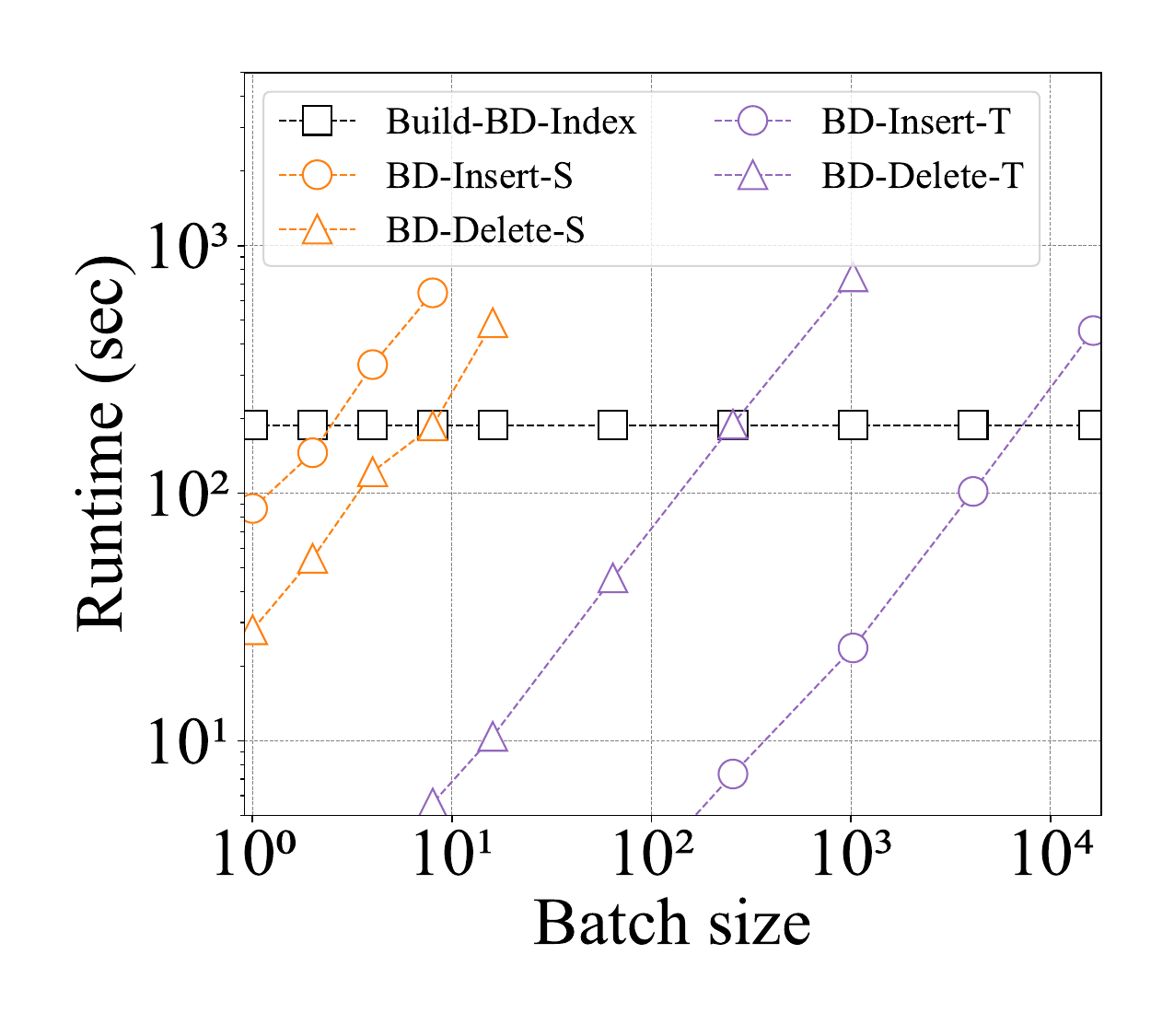}
    \vspace*{-0.8cm}
    \subcaption[font=scriptsize]{\color{\mycolorb}\textmd{Batch-update results.}} \label{batch-update}
  \end{subfigure}
  \vspace*{-0.2cm}
  \caption{\color{\mycolorb}Batch-query and batch-update results on the temporal \kw{AM} dataset.} \label{batch}
\end{figure}

\stitle{Performance of \kw{BD-Index} under batch queries and updates.} In this case study, we conduct tests on batch updates and batch queries using the real-world temporal Amazon review dataset \kw{AM} (details in \tabref{dataset}). For the batch-query test, we evaluate whether it is more efficient to process a batch of queries using the online algorithms or to construct the index and process queries based on it. The results are shown in \subfigref{batch-query}. We observe that when the query batch size exceeds 300, constructing the \kw{BD-Index} and answering queries on the index outperforms computing all queries directly using the online algorithms. For instance, at 1,024 queries, the index-based approach is over 3.24$\times$ faster than the online baseline \kw{Online++}. For the batch-update test, we evaluate whether it is more efficient to maintain the index incrementally using dynamic maintenance algorithms or to rebuild it from scratch upon receiving a batch of updates. The results are shown in \subfigref{batch-update}. As seen, the maintenance speed of the space-efficient strategy is relatively slow: the time to rebuild the index is roughly equivalent to inserting 3 edges and deleting 8 edges. In contrast, the time-efficient strategy demonstrates strong batch-update capability: the time to rebuild the index is roughly equivalent to inserting 5,000 edges and deleting 250 edges. This indicates that the time-efficient strategy can handle large-scale updates much more effectively, making it highly suitable for dynamic environments where frequent and sizable updates are expected.

\stitle{Summary.} The above case studies demonstrate that our \kw{BD-Index}-based approach consistently outperforms online algorithms in query processing, particularly under high-frequency query workloads. In realistic recommendation systems, query rates can reach 100,000 queries per second \cite{DBLP:conf/cikm/ZhangTWW24}, and real-time systems typically require responses within 0.5 seconds \cite{amazon,DBLP:conf/cikm/ZhangTWW24,DBLP:conf/sigmod/KersbergenSS22}. Under such stringent conditions, our index-based approach is far more suitable than online methods, as \kw{BD-Index} achieves optimal $O(|R|)$ query complexity and thus provides minimal latency. In our case studies simulating realistic workloads, the time-efficient maintenance strategy of \kw{BD-Index} processes all queries with turnaround times below 0.1 seconds, demonstrating exceptional capability for supporting $(\alpha,\beta)$-dense subgraph queries in highly dynamic environments. The space-efficient strategy, while more memory-friendly, incurs slower update performance and may therefore be preferable only when updates are relatively infrequent or memory constraints are strict.
}

\section{RELATED WORKS}

\stitle{Cohesive subgraph models in bipartite graphs.} There exists a wide range of cohesive models for bipartite graphs, such as biclique \cite{biclique1, biclique2, biclique3, biclique4} and its relaxed variant, the $k$-biplex \cite{biplex1, biplex2, biplex3}, as well as quasi-biclique models \cite{quasibiclique1,quasiclique2,quasiclique3}. In addition, butterfly-based model $k$-bitruss \cite{bitruss1, bitruss2, bitruss3} and degree-based model $(\alpha,\beta)$-core \cite{abcore1, abcore2, abcore3} have also been studied. For identifying the densest regions of a graph, the densest subgraph model is widely used \cite{densest1, densest2, densest3, densest4}. In terms of connectivity, the $k$-neighbor connectivity model has been proposed \cite{connectivity}. Recently, the $(\alpha,\beta)$-dense subgraph model was introduced \cite{ddbipartite} to effectively capture bipartite graph density structures with relatively efficient computation. For dynamic graphs, there exist algorithms for maintaining bicliques \cite{biclique_dynamic} and $(\alpha,\beta)$-cores \cite{abcore3}. To the best of our knowledge, we are the first to investigate efficient solutions for $(\alpha,\beta)$-dense subgraph search in both static and dynamic graphs.

\stitle{Density-based subgraph models in unipartite graphs.} In unipartite graphs, the most fundamental problems related to graph density is the densest subgraph search problem \cite{goldberg1, densest1, flowless, supermodularity}, which aims to find a subgraph that maximizes the ratio between the number of edges and the number of nodes. To address diverse application requirements, this problem has been extended to several variants, such as locally-dense decomposition \cite{locally-dense}, top-$k$ densest subgraphs \cite{DBLP:conf/kdd/QinLCZ15, DBLP:journals/pvldb/MaCLH22}, anchored densest subgraphs \cite{DBLP:conf/sigmod/DaiQC22, DBLP:conf/kdd/YeLLLLW24}, densest $k$-subgraph \cite{bourgeois2013exact, DBLP:journals/jal/AsahiroITT00}, and density decomposition \cite{dd}. In dynamic settings, many algorithms have been proposed for maintaining density-based structures, such as for the densest subgraph \cite{densest_dynamic1, DBLP:conf/stoc/SawlaniW20, mapreduce, DBLP:conf/wdag/SarmaLNT12, DBLP:conf/waw/BahmaniGM14}, top-$k$ densest subgraphs \cite{DBLP:conf/cikm/NasirGMG17}, and density decomposition \cite{dd}. \comment{Although the $(\alpha,\beta)$-dense subgraph search and maintenance problems studied in this paper are related to density-based models in unipartite graphs, existing techniques designed for unipartite settings cannot be directly applied to bipartite graphs.} However, existing techniques designed for unipartite settings cannot be directly applied to the $(\alpha,\beta)$-dense subgraph search problem.

\section{CONCLUSION}

This paper studies the problem of efficient $(\alpha,\beta)$-dense subgraph search and maintenance in bipartite graphs. First, leveraging the hierarchical property of $(\alpha,\beta)$-dense subgraphs, we introduce the concepts of $\alpha$-rank and $\beta$-rank to capture their inclusion relationships. Using these ranks, we organize nodes into $p$ compact node lists, forming a novel index structure called \kw{BD-Index}, which achieves optimal query time $O(|D_{\a, \b}|)$ with linear space complexity $O(|E|)$. We also propose an index construction algorithm with time complexity $O(p \cdot |E|^{1.5} \cdot \log |U \cup V|)$. To handle dynamic updates, we establish several novel update theorems\comment{ that characterize how $\alpha$-ranks and $\beta$-ranks change under edge insertions and deletions}. Building upon these theorems, we present two maintenance strategies. The space-efficient method maintains the index in $O(p \cdot |E|^{1.5})$ time and $O(|E|)$ space. The time-efficient method employs egalitarian orientations to reduce update time to $O(p \cdot |E|)$ while using $O(p \cdot |E|)$ space. Experiments on 10 real-world datasets demonstrate the efficiency and scalability of our solutions in both static and dynamic graphs.

\balance
\bibliographystyle{ACM-Reference-Format}
\bibliography{main}

\comment{
\appendix
\section*{APPENDIX}

In this appendix, we provide the missing proofs of the theorems presented in the paper, including \theref{insertion-update-theorem}, \theref{deletion-update-theorem}, \theref{u-update-theorem}, \theref{bd-insert-t-correctness}, and \theref{bd-delete-t-correctness}. We begin by proving the following lemma, which serves as the foundation for our subsequent proofs.

\begin{Lemma} \label{egalitarian-property}
    Given a graph $G$, an alpha value $\a$, and an egalitarian orientation $\vec{E}$, we have the following properties: (1) For any node $v \in V$, it holds that $\vec{d}_v(\vec{E}) \in \{r_\a(v), r_\a(v) + 1\}$; (2) For any nodes $x, y \in (U \cup V)$, if $r_\a(x) > r_\a(y)$, then the directed edge $(x, y)$ in $\vec{E}$ must point to $y$.
\end{Lemma}
\begin{proof}
    We first prove property (1). According to the definition of egalitarian orientation and \theref{orientation-to-rank}, we know that the indegree of $v$ $\vec{d}_v(\vec{E}) \le r_\a(v) + 1$; otherwise, $v$ would have a higher $\alpha$-rank value. Moreover, we have that either $v$ has an indegree greater than $r_\a(v)$, or $v$ can reach a node in $V$ with indegree greater than $r_\a(v)$. If $v$ has indegree greater than $r_\a(v)$, the property holds. Otherwise, by condition (2) in the definition of egalitarian orientation, $v$ must have indegree greater than $r_\a(v) - 1$, and thus the property also holds.

    Next, we prove property (2). According to the construction of $D_{\a,\b}$ in the proof of \theref{orientation-to-rank}, let $T = \{v \in V \mid \vec{d}_v(\vec{E}) > \b\}$. Then, $D_{\a,\b}$ consists of $T$ together with all nodes that can reach $T$, which implies that all edges between $D_{\a,\b}$ and $(U \cup V) \setminus D_{\a,\b}$ are directed from $D_{\a,\b}$ to $(U \cup V) \setminus D_{\a,\b}$. Therefore, if $r_\a(x) > r_\a(y)$, node $x$ belongs to a denser subgraph than $y$, and the edge $(x, y)$ must be directed toward $y$, which proves the property.
\end{proof}

Based on the properties of egalitarian orientation established in \lemref{egalitarian-property}, we first prove \theref{u-update-theorem}, followed by the proofs of \theref{insertion-update-theorem} and \theref{deletion-update-theorem}.

\paragraph{\textsc{\upshape Proof of \theref{u-update-theorem}}}

When $|N(u)| \le \a \Rightarrow d_u \le \a$, it follows from \lemref{degree} that $u \notin D_{\a,0}$, and thus $r_\a(u) = -1$. Otherwise, if $|N(u)| > \a$, let $\vec{E}$ be an egalitarian orientation of the given graph $G$ with respect to $\alpha$. According to the definition of egalitarian orientation, the indegree of $u$ is exactly $\a$, so there are $\a$ neighbors in $N(u)$ with directed edges pointing to $u$. By \lemref{egalitarian-property}, these neighbors must have $\alpha$-rank values greater than or equal to $r_\a(u)$.  

In addition, since $r_\a(u)$ is the $\alpha$-rank of $u$, $u$ must be able to reach a node $v_1 \in V$ whose indegree is greater than $r_\a(u)$. Let $(u, v_2)$ be the first edge on the path $u\rightsquigarrow v_1$. It follows that $r_\a(v_2) = r_\a(u)$, and by contradiction, we can show that $u$ cannot reach any node with an $\alpha$-rank greater than $r_\a(u)$.  

In summary, $N(u)$ contains $\a$ neighbors pointing to $u$, all of whose $\alpha$-ranks are greater than or equal to $r_\a(u)$. On the other hand, among the neighbors pointed to by $u$, the one with the highest $\alpha$-rank is $v_2$, whose rank is exactly $r_\a(u)$. Therefore, the $(\a+1)$-th largest $\alpha$-rank among the neighbors of $u$ is $r_\a(u)$, which completes the proof. 
\hfill\qedsymbol{}\vspace{0.2cm}

\paragraph{\textsc{\upshape Proof of \theref{insertion-update-theorem}}}

We divide the proof into three cases.

\textit{Case 1:} When $|N(u)| < \a$, both $r_N$ and $\b$ equal $-1$. Let $\vec{E}$ be an egalitarian orientation (defined in \defref{egalitarian-orientation}, which we will introduce later) of the given graph $G$ with respect to $\alpha$. We directly insert the edge $(u, v)$ into $\vec{E}$ with direction $(v, u)$. According to the definition of egalitarian orientation, the resulting $\vec{E}$ remains egalitarian. Since the reachability of all nodes remains unchanged, the $\alpha$-rank of each node also remains unchanged, which proves the case.

\textit{Case 2:} When $|N(u)| \ge \a$ and $r_\a(u) > r_\a(v)$, we have $r^N = r_\a(u)$ and $\b = r_\a(v)$ by \theref{u-update-theorem}. For any $\b' \in [-1, \b]\cup [\b+2, +\infty)$, let $\vec{E}$ be an orientation that satisfies the condition in \defref{ab-dense-subgraph} with alpha value as $\a$ and beta value as $\b'$. We insert the edge $(u,v)$ into $\vec{E}$ with direction towards $v$. It is easy to verify that the updated $\vec{E}$ still satisfies the condition in \defref{ab-dense-subgraph}, so $D_{\a,\b'}$ remains unchanged after insertion. Therefore, only $D_{\a,\b+1}$ may change after insertion. Since $\alpha$-rank values do not decrease, it follows that only the nodes with $r_\a = \b$ may have their $\alpha$-rank increased to $\b+1$, which proves the case.

\textit{Case 3:} When $|N(u)| \ge \a$ and $r_\a(u) \le r_\a(v)$, we have $r^N \ge r_\a(u)$ and $\b = r^N$. Let $\vec{E}$ be an egalitarian orientation of the given graph $G$ with respect to $\alpha$. We insert the edge $(u,v)$ into $\vec{E}$ with direction toward $u$, so $u$ now has $(\a+1)$ in-neighbors. The value $r^N$ represents the lowest $\alpha$-rank among these $(\a+1)$ neighbors; let $v_N$ be the corresponding neighbor. We then reverse the edge $(v_N, u)$ in $\vec{E}$ to become $(u, v_N)$. For any $\b' \in [-1, \b] \cup [\b+2, +\infty)$, the updated $\vec{E}$ still satisfies the condition in \defref{ab-dense-subgraph} with alpha value as $\a$ and beta value as $\b'$. Therefore, only the nodes in $V$ with $r_\a = \b$ may have their $\alpha$-rank increased to $\b+1$, which completes the proof.
\hfill\qedsymbol{}\vspace{0.2cm}

\paragraph{\textsc{\upshape Proof of \theref{deletion-update-theorem}}}

Let $\vec{E}$ be an egalitarian orientation of the given graph $G$ and alpha $\a$. We divide the proof into three cases.

\textit{Case 1:} When $|N(u)| \le \a$, it follows from \theref{degree} that $r_\a(u) = -1$. According to the definition of egalitarian orientation, all neighbors of $u$ are directed toward $u$ in $\vec{E}$, so the directed edge $(v,u)$ can be directly removed. It is clear that the orientation remains egalitarian, and all nodes’ $r_\a$ values remain unchanged, which completes the proof.

\textit{Case 2:} When $|N(u)| > \a$ and the edge $(u,v)$ is directed toward $v$ in $\vec{E}$. According to \lemref{egalitarian-property}, we have $r_\a(u) \ge r_\a(v)$, and thus $\b = r_\a(v)$. We directly remove the edge $(u,v)$ from $\vec{E}$. For any beta value $\b' \in [-1, \b-1] \cup [\b+1, +\infty)$, it is easy to verify that $\vec{E}$ still satisfies the condition in \defref{ab-dense-subgraph} with alpha value $\a$ and beta value $\b'$. This implies that only the subgraph $D_{\a,\b}$ is affected by the edge deletion. Since deleting an edge cannot increase any node’s $\alpha$-rank, it follows that only the nodes with $r_\a = \b$ may have their rank reduced to $\b - 1$, which proves the case.

\textit{Case 3:} When $|N(u)| > \a$ and the edge $(u,v)$ is directed toward $u$ in $\vec{E}$. According to \lemref{egalitarian-property}, we have $r_\a(u) \le r_\a(v)$, and thus $\b = r_\a(u)$. From the proof of \theref{u-update-theorem}, there exists a neighbor $v_2 \in N(u)$ such that $(u, v_2)$ is in $\vec{E}$ and $r_\a(v_2) = r_\a(u)$. We first remove the edge $(v, u)$ from $\vec{E}$, and then reverse the edge $(u, v_2)$ to become $(v_2, u)$. For any beta value $\b' \in [-1, \b-1] \cup [\b+1, +\infty)$, it is easy to verify that the updated $\vec{E}$ still satisfies the condition in \defref{ab-dense-subgraph} with alpha value $\a$ and beta value $\b'$. This shows that only the nodes in $V$ with $r_\a = \b$ may have their rank decreased to $\b - 1$, which completes the proof.
\hfill\qedsymbol{}\vspace{0.2cm}

\paragraph{\textsc{\upshape Proof of \theref{bd-insert-t-correctness}}}

Since the algorithm processes each node list (i.e., each alpha value) separately, we only need to prove that within the loop for a given alpha value (lines 2–7), the algorithm correctly updates the egalitarian orientation. Once the orientation is correctly maintained, the corresponding $r_\a$ values and the \kw{BD-Index} can be correctly updated based on \theref{orientation-to-rank}.  

Therefore, let the current alpha value be $\a$. If $d_u(G) < \a$ (line 4), it is clear that the egalitarian orientation is updated correctly. Otherwise, if $d_u(G) > \a$, let $(v_1, u)$ be the last edge on the path $v_{\min} \rightsquigarrow u$. Since $v_{\min}$ can reach $v_1$ or $v_{\min} = v_1$, it follows from \theref{egalitarian-property} that $r_\a(v_{\min}) \ge r_\a(v_1)$. Additionally, since $\vec{d}_{v_{\min}} \le \vec{d}_{v_1}$, we also have $r_\a(v_{\min}) \le r_\a(v_1)$. Therefore, $r_\a(v_{\min}) = r_\a(v_1)$. Similarly, we can conclude that all nodes along the path $v_{\min} \rightsquigarrow v_1$ have the same $\alpha$-rank as $v_{\min}$. Now, suppose we first reverse the path $v_{\min} \rightsquigarrow v_1$ in the original orientation. It is easy to see that the resulting orientation remains egalitarian. Then, we reverse the edge $(v_1, u)$ (i.e., the full path $v_{\min} \rightsquigarrow u$ has been reversed). Since $v_{\min}$ is the node with the smallest indegree among those that can reach $u$, the resulting orientation is still egalitarian. Hence, \kw{BD-Insert-T} correctly maintains the egalitarian orientation, and by \theref{orientation-to-rank}, it also correctly updates the \kw{BD-Index}.
\hfill\qedsymbol{}\vspace{0.2cm}

\paragraph{\textsc{\upshape Proof of \theref{bd-delete-t-correctness}}}

Similar to the correctness proof of \kw{BD-Insert-T}, the correctness of \kw{BD-Delete-T} also only requires showing that the algorithm correctly maintains the egalitarian orientation for each alpha value $\a$. First, if $d_u(G) \le \a$, it is straightforward to verify the correctness based on the definition of egalitarian orientation. Otherwise, if $d_u(G) > \a$, we consider two cases:

\textit{Case 1:} $(u,v) \in \vec{E}$ (line 6). In this case, according to the proof of \theref{u-update-theorem}, the neighbor $v_1 \in N(u)$ that is pointed to by $u$ and has the highest $r_\a$ satisfies $r_\a(v_1) = r_\a(u)$. Since $v_1$ can reach $v_{\max}$ or is equal to $v_{\max}$, we have $r_\a(v_1) \ge r_\a(v_{\max})$. On the other hand, since $v_{\max}$ has a higher indegree than $v_1$, it follows that $r_\a(v_{\max}) \ge r_\a(v_1)$. Therefore, $r_\a(v_1) = r_\a(v_{\max})$. Similarly, we can conclude that all nodes along the path $v_1 \rightsquigarrow v_{\max}$ have the same $\alpha$-rank as $v_{\max}$. Hence, by reversing the path $u \rightsquigarrow v_{\max}$ and deleting the edge $(u,v)$ or $(v,u)$ from $\vec{E}$, it is easy to see that the resulting orientation remains egalitarian.

\textit{Case 2:} $(v,u) \in \vec{E}$ (line 6). This case can be further divided into two subcases: (1) If $v_{\max} = v$, then the algorithm simply removes the edge $(u,v)$ without reversing any path. This implies that $v$ cannot reach any node in $V$ with a higher indegree. According to the definition of egalitarian orientation, any node in $V$ that can reach $v$ must have indegree at least $\vec{d}_v - 1$, so decreasing $v$'s indegree by removing $(u,v)$ does not violate the egalitarian condition. (2) If $v_{\max} \neq v$, we can use a similar argument as in Case 1 to show that the resulting orientation remains egalitarian after the path reversal and edge deletion.

Therefore, the algorithm correctly maintains the egalitarian orientation, and by \theref{orientation-to-rank}, it also correctly maintains the \kw{BD-Index}.
\hfill\qedsymbol{}\vspace{0.2cm}
}

\end{document}